\documentclass[preprint,sort&compress,3p]{elsarticle}
\usepackage{amsfonts}
\usepackage{amsthm}
\usepackage{amsmath}

\newtheorem{theorem}{Theorem}[section]
\newtheorem{corollary}{Corollary}[section]
\theoremstyle{definition}
\newtheorem{remark}{Remark}[section]
\newtheorem{example}{Example}[section]

\newcommand{\braket}[1]{\langle{#1}\rangle}

\newcommand{\lshad}{[\![}
\newcommand{\rshad}{]\!]}
\DeclareMathOperator{\tr}{tr}
\DeclareMathOperator{\Ham}{Ham}
\renewcommand{\d}[1]{\mathrm{d}{#1}\,}

\journal{Annals of Physics}

\begin{document}

\begin{frontmatter}

\title{Phase space quantum mechanics}
\author{Maciej B{\l}aszak}
\ead{blaszakm@amu.edu.pl}
\author{Ziemowit Doma{\'n}ski}
\ead{ziemowit@amu.edu.pl}
\address{Faculty of Physics, Adam Mickiewicz University, Umultowska 85, 61-614 Pozna{\'n}, Poland}

\begin{abstract}
The paper develope the alternative formulation of quantum mechanics known as the phase space quantum mechanics or deformation quantization.
It is shown that the quantization naturally arises as an appropriate deformation of the classical Hamiltonian mechanics. More precisely, the deformation
of the point-wise product of observables to an appropriate noncommutative $\star$-product and the deformation of the Poisson bracket to an appropriate Lie
bracket is the key element in introducing the quantization of classical Hamiltonian systems.

The formalism of the phase space quantum mechanics is presented in a very systematic way for the case of any smooth Hamiltonian function and for a very
wide class of deformations. The considered class of deformations and the corresponding $\star$-products contains as a special cases all deformations which
can be found in the literature devoted to the subject of the phase space quantum mechanics.

Fundamental properties of $\star$-products of observables, associated with the considered deformations are presented as well. Moreover, a space of states
containing all admissible states is introduced, where the admissible states are appropriate pseudo-probability distributions defined on the phase space.
It is proved that the space of states is endowed with a structure of a Hilbert algebra with respect to the $\star$-multiplication.

The most important result of the paper shows that developed formalism is more fundamental then the axiomatic ordinary quantum mechanics which appears in
the presented approach as the intrinsic element of the general formalism.
The equivalence of two formulations of quantum mechanics is proved by observing that the Wigner-Moyal transform has all properties of the tensor product. This observation allows writing many previous results found in
the literature in a transparent way, from which the equivalence of the two formulations of quantum mechanics follows naturally.

In addition, examples of a free particle and a simple harmonic oscillator
illustrating the formalism of the deformation quantization and its classical limit are given.
\end{abstract}


\end{frontmatter}

\section{Introduction}
In the 1970's \citet{Bayen:1978a, Bayen:1978b} laid the foundations of an alternative description of quantum mechanics, referred to as the \emph{phase
space quantum mechanics} or \emph{deformation quantization}, see also \cite{Dito.Sternheimer:2002, Gutt:2000, Weinstein:1994b} for recent reviews. Their
work was based on earlier works of \citet{Weyl:1927, Weyl:1931}, \citet{Wigner:1932}, \citet{Groenewold:1946}, \citet{Moyal:1949} and \citet{Berezin:1975a,
Berezin:1975b, Berezin:1975c} on the physical side and of \citet{Gerstenhaber:1963, Gerstenhaber:1964, Gerstenhaber:1966, Gerstenhaber:1968,
Gerstenhaber:1974}, and \citet{Gerstenhaber:1988} deformation theory of associative algebras on the mathematical side. Since then many efforts have been
made in order to develop the phase space quantum mechanics \cite{Fairlie:1964, Omori:1991, Curtright:1998, Curtright:1999a, Zachos:2002, Zachos:2005,
Hirshfeld:2002a, Hirshfeld:2002b, Bordemann:1998, Waldmann:2005, Dias:2004a, Dias:2004b, Sternheimer:1998, Nasiri:2006, Dias:2001, Gosson:2005, Gosson:2006,
Gosson:2008, Bracken:2010, Smith:2006, Dias:2010, Cohen:1966}. It became one of the most attractive and successful quantization theories, both from the
mathematical and physical point of view, with wide range of applications beyond the original quantization problem, including quantum optics, theory of
dynamical systems, field theory \cite{Curtright:1999b, Antonsen:1997a, Dito:2002, Dito:1990, Dito:1992} and M-theory \cite{Fairlie:1998, Baker:1999,
Pinzul:2001}. Recently the most prominent application come from noncommutative geometry \cite{Connes:1994} and the noncommutative field theory models
arising from it \cite{Bahns:2003, Connes:1998, Jurco.Moller:2001, Jurco:2000, Jurco:2001, Seiberg:1999}.

The ordinary description of quantum mechanics is given by a set of axioms from which the connection with classical mechanics is not evident. The natural
formulation of quantum mechanics seems to be a result of a generalization of classical Hamiltonian mechanics, in such a way that the new formulation of
quantum mechanics should smoothly reduce to the formulation of classical mechanics as the Planck constant $\hbar$ goes to 0. The phase space quantum
mechanics is such a natural formulation of quantum theory.

The idea behind the phase space quantum mechanics relies on a deformation, with respect to some parameter $\hbar$ (the Planck constant), of a geometrical
structure of a classical phase space $M$. Since, for a given manifold $M$ geometrical structures on it (in our case the Poisson tensor) can be defined
in terms of the algebra $C^\infty(M)$ (the algebra of smooth complex valued functions on $M$), the deformation of the algebra $C^\infty(M)$ results in
the deformation of the geometrical structures on the manifold $M$. By the deformation of $C^\infty(M)$ is meant the deformation of the point-wise product
between functions to some noncommutative product, commonly denoted by $\star$, as well as the deformation of the complex conjugation of functions to
an appropriate involution, denoted hereafter by $\dagger$. The deformation of the algebra $C^\infty(M)$ results in the deformation of the Poisson tensor,
and consequently the Poisson bracket, to some deformed Poisson tensor as well as some deformed Poisson bracket, denoted hereafter by
$\lshad \,\cdot\, , \,\cdot\, \rshad$. The elements of $C^\infty(M)$ are, like in a classical case, the observables. The deformed algebra $C^\infty(M)$
will be denoted hereafter by $\mathcal{A}_Q$.

In the majority of papers, see e.g. \cite{Bayen:1978a, Waldmann:2005, Sternheimer:1998}, the deformation of the product between observables is performed
formally by considering as the algebra of observables the $\mathbb{C}[[\hbar]]$-linear module $C^\infty(M)[[\hbar]]$ of formal power series in the
deformation parameter $\hbar$ with coefficients from the space $C^\infty(M)$ of all smooth complex valued functions defined on the phase space $M$, where
$\mathbb{C}[[\hbar]]$ is the ring of formal power series in the parameter $\hbar$ with coefficients from $\mathbb{C}$. It is however possible to perform
strict deformation of the product between observables by endowing the algebra of observables with appropriate topology and requiring the $\star$-product
to continuously depend on $\hbar$. In the framework of $C^*$-algebras this was done by \citet{Rieffel:1993, Natsume:2001, Natsume:2003}, see also
\cite{Weinstein:1994a, Bonneau:2004}, whereas in the framework of Fr{\'e}chet algebras this was studied by \citet{Omori:2000, Omori:2007}. In the
following paper the problem of convergence in the algebra $\mathcal{A}_Q$ will not be considered since, although it is important mathematical problem,
it is not crucial from the physical point of view.

In formal deformation quantization one has very strong existence and classification results for the $\star$-products. For the symplectic manifolds the
general existence was shown first by \citet{DeWilde:1983, DeWilde:1988}, later by \citet{Fedosov:1985, Fedosov:1986, Fedosov:1989, Fedosov:1994}
and \citet{Omori:1991}. The case where the classical phase space is a general Poisson manifold turned out to be much more difficult and was finally
solved by \citet{Kontsevich:2003}, see also \cite{Kontsevich:1997, Kontsevich:1999, Kontsevich:2001}. The classification of $\star$-products was
obtained first for the symplectic case by \citet{Nest:1995a, Nest:1995b}, \citet{Bertelson:1997}, \citet{Weinstein:1998}. The classification for
the Poisson case was given by \citet{Kontsevich:2003}.

In addition to the deformation of the product between observables also the definition of states have to be changed slightly. The admissible states of the
classical Hamiltonian system are defined as probability distributions on the phase space. In the phase space quantum mechanics states have to be defined
as \emph{pseudo-probability distributions} on the phase space (the distributions do not have to be positive define). This reflects the
quantum character of the states. It is possible to introduce a Hilbert space $\mathcal{H}$ which contains all admissible states. Moreover, the
$\star$-product can be extended to a noncommutative product between functions from $\mathcal{H}$ creating from $\mathcal{H}$ a Hilbert algebra. The
action of some observable from $\mathcal{A}_Q$ on some state from $\mathcal{H}$ can also be expressed by the $\star$-product. Furthermore, the expectation
values of observables and the time evolution equation are defined like in the classical case, except the fact that the point-wise product and the Poisson
bracket are replaced by the $\star$-product and the Lie bracket $\lshad \,\cdot\, , \,\cdot\, \rshad$.

In the paper the case of the deformation quantization of the Hamiltonian system which $\star$-product is related to the canonical
Poisson bracket is considered. Even in this case one can introduce infinitely many $\star$-products inducing proper deformations of a classical
Hamiltonian system. In the majority of papers only the special case of the $\star$-product is considered, namely the \emph{Moyal product}. In some papers
also other particular $\star$-products, \emph{gauge equivalent} to the Moyal product, are discussed. Moreover, quite often (especially in quantum optics)
the deformation quantization in \emph{holomorphic coordinates} is also considered. In the paper a very broad family of $\star$-products, equivalent to
the Moyal product, is constructed. This family of $\star$-products contains in particular all examples of $\star$-products which can be found in the
variety of papers devoted to the phase space quantum mechanics.

The formulation of quantum mechanics on the phase space is equivalent to the ordinary formulation of quantum mechanics. In the majority of papers this
is proved with the help of the so called Wigner map and its inverse, see e.g. \cite{Dias:2004a, Dias:2004b}. Actually, it is a morphism between the
algebra $\mathcal{A}_Q$ with noncommutative star multiplication and the algebra of linear operators in a Hilbert space with the multiplication being
simple composition.

Another approach to construction of the phase space quantum mechanics is to extend ordinary Schr{\"o}dinger equation to the so called
\emph{Schr{\"o}dinger equation in phase space} with related eigenfunctions being appropriate distributions \cite{Torres-Vega:1990, Torres-Vega:1993,
Li:2004}. As was shown recently it is equivalent to an eigenvalue problem of classical Hamiltonian and an appropriate star multiplication
\cite{Chruscinski:2005}.

These points of view are however a little bit miss leading since one could thought that both descriptions of quantum mechanics are completely different
and that the phase space quantum mechanics is just some representation of the ``more fundamental'' ordinary quantum mechanics. In this paper is presented
an alternative point of view on the relation between the phase space quantum mechanics and the ordinary description of quantum mechanics. It is shown that
the ordinary description of quantum mechanics appears as a natural consequence of the phase space quantum mechanics. Moreover, from the presented
construction it is evident that the phase space quantum mechanics is the most fundamental formulation of quantum mechanics.

To summarize, the aim of this paper is to present in a systematic way the phase space quantum mechanics in a canonical regime as a natural deformation
of classical Hamiltonian mechanics for a very general class of gauge equivalent deformations. Moreover, it is shown that from the phase space quantum
mechanics naturally appears, at least in the canonical regime, the ordinary description of quantum mechanics. In addition, the physical equivalence of
the presented family of deformations is discussed.

The paper is organized as follows. In Section \ref{sec:7} the classical Hamiltonian mechanics is reviewed. Since the quantization scheme presented in
the paper is closely related to the classical Hamiltonian mechanics the review of classical mechanics is done using the same language as that from quantum
theory. In this section the basic concepts of the Hamiltonian mechanics are given including the definitions of a phase space, observables and Hamiltonian
systems. Moreover, the definitions of pure and mixed states and expectation values of observables are given together with their basic properties. Also the
time evolution of pure and mixed states and observables is presented.

In Section \ref{sec:6} the formulation of the phase space quantum mechanics is presented. This section starts with the review of some basics of the
deformation quantization of general Hamiltonian systems followed by the full description of a one-parameter family of quantization schemes of classical
Hamiltonian systems in canonical regime. The results in this part of the paper are well known in the literature and are given
here for completeness. In particular, there are presented fundamental properties of canonical $\star_\sigma$-products together with the systematic
construction of the space of states.

Next, the one-parameter family of quantization schemes is extended to a very broad family of quantizations sharing the same properties. More general
star-products and quantizations related to them were introduced by \citet{Bayen:1978a, Bayen:1978b} and further studied by \citet{Dias:2001}.
Already \citet{Cohen:1966} introduced a very general class of star-products and quasi-distributions. Then, many authors considered various particular
cases of the Cohen's construction (see review article \cite{Lee:1995}). However, we have to stress that presented in the paper quantizations are much
more general then these considered by \citet{Cohen:1966} and others.

In this section also the definitions of pure and mixed states and expectation values of observables are given together with their basic properties.
Moreover, the time evolution equations of states and observables are presented.

In Section \ref{sec:9} the equivalence of the ordinary formulation and the formulation on the phase space of quantum mechanics is presented. In the
existing literature the equivalence of the phase space formulation of quantum mechanics and the Schr\"odinger, Dirac and Heisenberg formulation of quantum
mechanics (referred hereafter as the ordinary description of quantum mechanics), in which observables and states are defined as operators on the Hilbert
space $L^2(\mathbb{R})$, is derived, using the powerful tools of pseudo-differential operators. See for example \cite{Dias:2004a, Dias:2004b}. In the
paper the equivalence of these two formulations of quantum mechanics will be proved by observing that the Wigner-Moyal transform (which definition and
properties can be found in \cite{Gosson:2005}) have all properties of the tensor product. This observation allows writing many previous results found in
the literature in a lucid and elegant way, from which the equivalence of the two formulations of quantum mechanics is easier to see. Moreover, this
observation will also provide an argument to treat the phase space quantum mechanics as a more fundamental formalism of quantum mechanics, than the one
developed by Schr\"odinger, Dirac and Heisenberg. In fact, it is proved that the observables and states, represented as operators of the form $A \star {}$
and $\rho \star {}$ where $A$ and $\rho$ are functions on the phase space, can be written in the form of a tensor product:
$A \star {} = \hat{1} \otimes A(\hat{q},\hat{p})$ and $\rho \star {} = \hat{1} \otimes \hat{\rho}$, where $A(\hat{q},\hat{p})$ and $\hat{\rho}$ are
operators on the Hilbert space $L^2(\mathbb{R})$. Also it is proved that the action of the observable $A \star {}$ on the state $\rho \star {}$ is
expressed by the formula $A \star \rho \star {} = \hat{1} \otimes A(\hat{q},\hat{p}) \hat{\rho}$. The presented approach is an extension and reformulation
of the idea found in \cite{Gosson:2005}, see also \cite{Gosson:2006, Gosson:2008}, where it was proved, for the Moyal case, that
$A \star {} = U_\varphi A(\hat{q},\hat{p}) U_\varphi^{-1}$. $U_\varphi$ is the Wigner wave-packet transform which maps wave functions on a configuration
space to wave functions on a phase space. These two approaches are equivalent after noting that
$U_\varphi \psi = W(\varphi^*, \psi) = \varphi^* \otimes \psi$, where $W$ is the Wigner-Moyal transform.

Section \ref{sec:9} contains also proofs of some properties from former sections.

In Section \ref{sec:12} examples of a free particle and a simple harmonic oscillator are presented illustrating the formalism of the deformation
quantization. In the example of a free particle the time evolution of a free particle initially in some given state is derived. In the example of a
harmonic oscillator stationary states and coherent states of the harmonic oscillator are derived. Moreover, it is proved that the presented examples of
quantum states converge to appropriate classical states as Planck constant goes to zero.

Section \ref{sec:15} contains a discussion on the physical equivalence of the presented family of deformations, while Section \ref{sec:14} contains final
comments and conclusions. The notation and conventions used in the paper as well as longer technical proofs of some theorems from a main text can be
found in \ref{sec:8} -- \ref{sec:2}.

\section{Classical Hamiltonian mechanics}
\label{sec:7}
In this section the basics of the classical Hamiltonian mechanics will be reviewed and adopted to the language of quantum theory. For further use during the quantization procedure, in what follows,
complex tensor fields and functions on a manifold will be considered. Nevertheless, all below considerations can be made using only real tensor fields
and functions. More details on classical Hamiltonian mechanics can be found in \cite{Libermann:1987, Fecko:2006}.

\subsection{Classical algebra of observables}
\label{subsec:7.1}

In the classical Hamiltonian mechanics pure states are represented by points in a \emph{phase space}, which in turn is represented by a \emph{Poisson
manifold}. The \emph{Poisson manifold} is a smooth manifold $M$ endowed with a two times contravariant antisymmetric (real) tensor field $\mathcal{P}$
satisfying the below relation
\begin{equation}
\mathcal{L}_{\zeta_f} \mathcal{P} = 0,
\label{eq:7.1.1}
\end{equation}
($\mathcal{L}_{\zeta_f}$ denotes a Lie derivative in the direction $\zeta_f$) for every vector fields $\zeta_f$ defined as
\begin{equation*}
\zeta_f := \mathcal{P} \d{f}, \quad f \in C^\infty(M).
\end{equation*}
The tensor field $\mathcal{P}$ is called a \emph{Poisson tensor} and the vector fields $\zeta_f$ are called \emph{Hamiltonian fields}. The space of all
Hamiltonian fields on $M$ will be denoted by $\Ham(M)$. In the rest of the paper it will be assumed that the Poisson tensor $\mathcal{P}$ is
non-degenerate. In this case it can be proved that the Poisson manifold $M$ is even-dimensional.

Using the Poisson tensor $\mathcal{P}$ a structure of a Lie algebra can be added to the space $C^\infty(M)$ of all (complex valued) smooth functions on
$M$, namely a Lie bracket can be defined as
\begin{equation}
\{f,g\}_{\mathcal{P}} := \mathcal{P}(\d{f},\d{g}), \quad f,g \in C^\infty(M).
\label{eq:7.1.3}
\end{equation}
It is obviously antisymmetric since $\mathcal{P}$ is antisymmetric and the Jacobi's identity follows from relation (\ref{eq:7.1.1}). In fact there holds
\begin{align*}
\{f,g\} & = -\{g,f\} & \textrm{(antisymmetry)}, \\
\{f,gh\} & = \{f,g\}h + g\{f,h\} & \textrm{(Leibniz's rule)}, \\
0 & = \{f,\{g,h\}\} + \{h,\{f,g\}\}  + \{g,\{h,f\}\} & \textrm{(Jacobi's identity)}.
\end{align*}
The bracket (\ref{eq:7.1.3}) is called a \emph{Poisson bracket} and the algebra $C^\infty(M)$ endowed with a Poisson bracket is called a \emph{Poisson
algebra}.

To describe a physical system besides a phase space also an \emph{algebra of observables} is needed. Lets introduce a notation
$\mathcal{A}_C = C^\infty(M)$. Usually in classical mechanics as the algebra of observables a real subalgebra $\mathcal{O}_C \subset \mathcal{A}_C$ of
all real valued smooth functions on $M$ is taken, but for further use during the introduction of the ordinary description of quantum mechanics it will
be better to define the algebra of observables in a different way. First, lets define an algebra $\hat{\mathcal{A}}_C$ of all operators defined on the
space $C^\infty(M)$ of the form $\hat{A} = A \cdot {}$, where $A \in \mathcal{A}_C$ and $\cdot$ denotes an ordinary point-wise product of functions from
$C^\infty(M)$. The algebra $\hat{\mathcal{A}}_C$ have a structure of a Lie algebra with a Lie bracket defined by the formula
\begin{equation*}
\{ \hat{A}, \hat{B} \} := \{ A,B \} \cdot {}, \quad \hat{A},\hat{B} \in \hat{\mathcal{A}}_C.
\end{equation*}
Now, the algebra of observables can be defined as a real subalgebra $\hat{\mathcal{O}}_C$ of all operators from $\hat{\mathcal{A}}_C$ induced by real
valued smooth functions on $M$. One of the admissible observables from $\hat{\mathcal{O}}_C$ have a special purpose, namely a \emph{Hamiltonian}
$\hat{H}$. This observable corresponds to the total energy of the system and it governs the time evolution of the system. A triple
$(M,\mathcal{P},\hat{H})$ is then called a \emph{classical Hamiltonian system}.

Local coordinates $q^i,p_i$ ($i = 1,\ldots,N$) in which a Poisson tensor $\mathcal{P}$ have (locally) a form
\begin{equation*}
\mathcal{P} = \frac{\partial}{\partial q^i} \wedge \frac{\partial}{\partial p_i}
= \frac{\partial}{\partial q^i} \otimes \frac{\partial}{\partial p_i} - \frac{\partial}{\partial p_i} \otimes \frac{\partial}{\partial q^i}
\quad \textrm{i.e.} \quad
\mathcal{P}^{ij} = \left( \begin{array}{cc}
0_N & \mathbb{I}_N \\
-\mathbb{I}_N & 0_N
\end{array} \right)
\end{equation*}
are called \emph{canonical coordinates}. Furthermore, in the canonical coordinates a Hamiltonian field $\zeta_f$ and a Poisson bracket take a form
\begin{equation}
\zeta_f = \frac{\partial f}{\partial p_i} \frac{\partial}{\partial q^i} - \frac{\partial f}{\partial q^i} \frac{\partial}{\partial p_i},\qquad
\{f,g\} = \frac{\partial f}{\partial q^i} \frac{\partial g}{\partial p_i} - \frac{\partial f}{\partial p_i} \frac{\partial g}{\partial q^i}.
\label{eq:7.1.11}
\end{equation}
It can be proved that for every Poisson manifold canonical coordinates always exist, at least locally.

\subsection{Pure states, mixed states and expectation values of observables}
\label{subsec:7.2}

As mentioned earlier \emph{pure states} of a classical Hamiltonian system are points in a phase space $M$. They represent generalized positions and
momenta of a phase space particle. Values of generalized positions and momenta of the particle can be extracted from a point in $M$ (a pure state)
by writing this point in canonical coordinates $q^i,p_i$. Then, $q^i$ are the values of generalized positions and $p_i$ are the values of generalized
momenta of the particle.

When one does not know the exact positions and momenta of the phase space particle, but only a probability distribution that the system is in some
point of the phase space then there is a need to extend the concept of a state. It is natural to generalize the states to probability distribution
functions defined on the phase space $M$, i.e. to smooth functions $\rho$ on $M$ satisfying
\begin{enumerate}
\item $\rho(\xi) \ge 0$,
\item $\displaystyle{\int_M \rho(\xi) \d{\xi} = 1}$.
\end{enumerate}
Such generalized states are called \emph{mixed states}. The probability distribution functions $\rho$ on $M$ will be also called a \emph{classical
distribution functions}. Pure states $\xi_0 \in M$ can be then identified with Dirac delta distributions $\delta(\xi - \xi_0)$.

An expectation value $\braket{\hat{A}}_\rho$ of an observable $\hat{A} \in \hat{\mathcal{A}}_C$ in a state $\rho$ is defined by
\begin{equation}
\braket{\hat{A}}_\rho := \int_M (\hat{A}\rho)(\xi) \d{\xi} = \int_M A(\xi) \cdot \rho(\xi) \d{\xi}.
\label{eq:7.2.16}
\end{equation}
Note that an expectation value of the observable $\hat{A}$ in a pure state $\delta(\xi - \xi_0)$ is just equal $A(\xi_0)$.

In classical mechanics (to be compared with quantum mechanics) the uncertainty relations take the form
\begin{equation}
\Delta q^i \Delta p_j \ge 0, \qquad i,j = 1,\ldots,N,
\label{eq:7.2.2}
\end{equation}
where $\Delta q^i$ and $\Delta p_j$ are the uncertainties of the $i$-th position observable and the $j$-th momentum observable respectively, in some state.
The above uncertainty relations state that it is possible to simultaneously measure the position and momentum of a particle with an arbitrary precision.
Obviously the equality in (\ref{eq:7.2.2}) takes place for pure states.

\subsection{Time evolution of classical Hamiltonian systems}
\label{subsec:7.3}
For a given Hamiltonian system $(M,\mathcal{P},\hat{H})$ the Hamiltonian $\hat{H}$ governs the time evolution of the system. Namely, the Hamiltonian
$\hat{H}$ generates a Hamiltonian field $\zeta_H$. The flow $\Phi^H_t$ (called the \emph{phase flow} or the \emph{Hamiltonian flow}) of this Hamiltonian
field moves the points of $M$, which is interpreted as the time development of pure states ($\xi(t) = \Phi^H_t(\xi(0))$). A trajectory of a point
$\xi \in M$ (a pure state) can be then calculated from the equation
\begin{equation}
\dot{\xi} = \zeta_H = \mathcal{P} \d{H}.
\label{eq:7.3.17}
\end{equation}
In canonical coordinates $q^i,p_i$, using formulae (\ref{eq:7.1.11}), equation (\ref{eq:7.3.17}) takes a form of the following system of differential
equations called the \emph{canonical Hamilton equations}
\begin{equation}
\dot{q}^i = \frac{\partial H}{\partial p_i}, \quad \dot{p}_i = -\frac{\partial H}{\partial q^i}.
\label{eq:7.3.27}
\end{equation}
The equation of motion of mixed states can be derived from the probability conservation law. From this law and the Hamilton equations the following
equation can be derived
\begin{equation}
L(H,\rho) = \frac{\partial \rho}{\partial t} - \{ H,\rho \} = 0.
\label{eq:7.3.18}
\end{equation}
Equation (\ref{eq:7.3.18}) is called the \emph{Liouville equation} and it describes the time development of an arbitrary state $\rho$. Of course for a
pure state $\delta(\xi - \xi(t))$ the Liouville equation (\ref{eq:7.3.18}) is equivalent to the Hamilton equations (\ref{eq:7.3.27}).

From (\ref{eq:7.3.18}) it follows that a time dependent expectation value of an observable $\hat{A} \in \hat{\mathcal{A}}_C$ in a state $\rho(t)$, i.e.
$\braket{\hat{A}}_{\rho(t)}$, fulfills the following equation of motion
\begin{equation}
\braket{\hat{A}}_{L(H,\rho)} = 0 \iff \frac{\d{}}{\d{t}} \braket{\hat{A}}_{\rho(t)} - \braket{\{\hat{A},\hat{H}\}}_{\rho(t)} = 0.
\label{eq:7.3.20}
\end{equation}

Until now the states undergo the time development whereas the observables do not. This is called a \emph{Schr\"odinger-like picture} of the time
evolution. There is also a dual point of view (which, in turn, is referred to as a \emph{Heisenberg-like picture}), in which states remain still whereas
the observables undergo a time development. A pull-back of the Hamiltonian flow $(\Phi^H_t)^* = e^{t\mathcal{L}_{\zeta_H}}$ induces, for every $t$, an
automorphism $U^H_t$ of the algebra $\hat{\mathcal{A}}_C$ (a one-to-one map preserving the linear structure as well as the dot-product and the Lie
bracket) given by the equation
\begin{equation*}
U^H_t \hat{A} := (\Phi^H_t)^* A \cdot {}, \quad \hat{A} \in \hat{\mathcal{A}}_C.
\end{equation*}
Its action on an arbitrary observable $\hat{A} \in \hat{\mathcal{A}}_C$ is interpreted as the time development of $\hat{A}$
\begin{equation}
\hat{A}(t) = U^H_t \hat{A}(0) = e^{t \mathcal{L}_{\zeta_H}} A(0) \cdot {} = e^{t \zeta_H} A(0) \cdot {}.
\label{eq:7.3.21}
\end{equation}
Differentiating equation (\ref{eq:7.3.21}) with respect to $t$ one receives
\begin{equation*}
\frac{\d{\hat{A}}}{\d{t}}(t) = \zeta_H A(t) \cdot {}.
\end{equation*}
From above equation the following time evolution equation for an observable $\hat{A}$ follows
\begin{equation}
\frac{\d{\hat{A}}}{\d{t}}(t) - \{\hat{A}(t),\hat{H}\} = 0.
\label{eq:7.3.22}
\end{equation}
Both presented approaches to the time development yield equal predictions concerning the results of measurements, since
\begin{equation*}
\braket{\hat{A}(0)}_{\rho(t)} = \braket{\hat{A}(t)}_{\rho(0)}.
\end{equation*}

\section{Quantization procedure on a phase space}
\label{sec:6}
\subsection{Basics of deformation quantization}
\label{subsec:6.1}
This section contains basic information about deformation quantization present in the literature. The idea of quantization of the classical mechanics
relays on the postulate that on the microscopic level the classical uncertainty relations (\ref{eq:7.2.2}) violate and have to be modified to the Heisenberg uncertainty relations, i.e. that the
uncertainties $\Delta q^i$ and $\Delta p_j$ of the $i$-th position observable and the $j$-th momentum observable satisfy the following equation
\begin{equation*}
\Delta q^i \Delta p_j \ge \frac{1}{2} \hbar \delta^i_j, \qquad i,j = 1,\ldots,N,
\end{equation*}
instead of (\ref{eq:7.2.2}) one, where $\hbar$ is some fundamental constant (the Planck constant) to be determined in a physical experiment.
There exist two equivalent ways of introducing the deformation (quantization) of classical mechanics such that the Heisenberg uncertainty principle is fulfilled.
The first way is based on an appropriate deformation (coarse graining) of the classical probability distribution functions \cite{Wetterich:2010}.
In this approach the coarse grained classical states can be seen as the states of some subsystem of the classical statistical ensemble which describes
the quantum particle in phase space and contains no longer any information beyond what is available in the usual quantum formalism. In addition, the time evolution of the coarse grained classical states should remain unitary if the quantum system remains isolated and the states in question
should satisfy the Heisenberg uncertainty principle.

In the second way the quantization of classical mechanics is introduced by an appropriate deformation of the phase space resulting from a deformation
of the point-wise product of smooth complex valued functions on the phase space to an appropriate noncommutative product. This two ways of quantizing
classical mechanics are equivalent since the algebraic structure of the algebra of observables introduces the admissible states, and vice verse the
space of admissible states can be used to introduce the algebraic structure on the space of observables. In the paper the second way of quantizing
classical mechanics will be presented. Hence, the quantization procedure will look as follows.

Let $(M,\mathcal{P})$ be a $2N$-dimensional phase space (i.e. a smooth Poisson manifold) and $\mathcal{A}_C = C^\infty(M)$ be the algebra of all smooth
complex valued functions on $M$ with respect to the standard point-wise product and a Poisson bracket associated to $\mathcal{P}$. The quantization of
the classical mechanics relays on a deformation, with respect to some parameter $\hbar$, of the classical phase space $(M,\mathcal{P})$
to a noncommutative phase space \cite{Dias:2004b, Fairlie:1964, Hirshfeld:2002a, Sternheimer:1998, Zachos:2002, Zachos:2005}. The noncommutative phase
space is composed of a noncommutative smooth manifold and a Poisson tensor defined on it. The noncommutative manifold is just an ordinary manifold $M$
whose algebra of smooth functions $C^\infty(M)$ has the point-wise multiplication deformed to a noncommutative multiplication denoted by $\star$. Since,
for a given manifold $M$ geometrical structures on it (in our case the Poisson tensor) can be defined in terms of the algebra $C^\infty(M)$, the
deformation of the algebra $C^\infty(M)$ results in the deformation of the geometrical structures on the manifold $M$.

In the paper only such quantized (deformed) Poisson tensors will be considered which associated deformed Poisson brackets
$\lshad \,\cdot\, , \,\cdot\, \rshad$ are given by the formula
\begin{equation}
\lshad f,g \rshad := \frac{1}{i\hbar} [f,g], \quad f,g \in C^\infty(M),
\label{eq:6.1.2}
\end{equation}
where
\begin{equation*}
[f,g] := f \star g - g \star f, \quad f,g \in C^\infty(M)
\end{equation*}
is the $\star$-commutator.

The algebra, of all smooth complex valued functions on the quantum phase space $M$, with respect to the $\star$-product and the deformed Poisson bracket
$\lshad \,\cdot\, , \,\cdot\, \rshad$ will be denoted by $\mathcal{A}_Q$. Moreover, an involution in the algebra $\mathcal{A}_Q$ will be introduced as
a deformation of the complex conjugation of functions. This involution will be denoted by $\dagger$. Furthermore, by $\mathcal{O}_Q$ a subspace of
$\mathcal{A}_Q$ of all self-adjoint functions, with respect to the involution $\dagger$, will be denoted. Note that in general $\mathcal{O}_Q$ do not
need to constitute an algebra with respect to the $\star$-product. As a space of admissible quantum observables the space $\mathcal{O}_Q$ can be taken,
but for the same reason as in Section \ref{sec:7} it will be better to define the space of quantum observables in a similar way as in Section \ref{sec:7}.
First, lets introduce an algebra $\hat{\mathcal{A}}_Q$ of all operators defined on the space $C^\infty(M)$ of the form $\hat{A} = A \star {}$, where
$A \in \mathcal{A}_Q$. The algebra $\hat{\mathcal{A}}_Q$ have a structure of a Lie algebra with a Lie bracket defined by the formula
\begin{equation*}
\lshad \hat{A},\hat{B} \rshad := \lshad A,B \rshad \star {} = \frac{1}{i\hbar}(\hat{A}\hat{B} - \hat{B}\hat{A}), \quad \hat{A},\hat{B} \in \hat{\mathcal{A}}_Q.
\end{equation*}
All operators from $\hat{\mathcal{A}}_Q$ induced by self-adjoint functions are defined as the admissible quantum observables. The space of all admissible
quantum observables will be denoted by $\hat{\mathcal{O}}_Q$. Note, that the algebra $\hat{\mathcal{A}}_Q$ is a deformation of the algebra
$\hat{\mathcal{A}}_C$.

\begin{remark}
In the paper the problem of convergence in the algebra $\mathcal{A}_Q$ will not be considered. It will be assumed that the algebra $\mathcal{A}_Q$ is
endowed with an appropriate topology in which the $\star$-product and the involution $\dagger$ are continuous, and in which all limits of functions
from $\mathcal{A}_Q$ appearing in the paper are convergent. Although, finding such topology is a perfectly legitimate mathematical problem it is beyond the content of the paper.
\end{remark}

The deformed noncommutative multiplication on $\mathcal{A}_Q$ should satisfy such natural conditions
\begin{enumerate}
\item $f \star (g \star h) = (f \star g) \star h$ (associativity),
\item $f \star g = fg + o(\hbar)$,
\item $\lshad f,g \rshad = \{f,g\} + o(\hbar)$,
\item $f \star 1 = 1 \star f = f$,
\end{enumerate}
for $f,g,h \in \mathcal{A}_Q$. Moreover, it is assumed that the $\star$-product can be expanded in the following infinite series with respect to the
parameter $\hbar$
\begin{equation*}
f \star g = \sum_{k=0}^\infty \hbar^k B_k(f,g), \quad f,g \in \mathcal{A}_Q,
\end{equation*}
where $B_k \colon \mathcal{A}_Q \times \mathcal{A}_Q \to \mathcal{A}_Q$ are bilinear operators. From the construction of the $\star$-product it can be
immediately seen that in the limit $\hbar \to 0$ the quantized algebra $\hat{\mathcal{A}}_Q$ reduces to the classical algebra $\hat{\mathcal{A}}_C$.
From the associativity of the $\star$-product it follows that the bracket (\ref{eq:6.1.2}) is a well-defined Lie bracket. In fact, it satisfies the
following relations
\begin{align*}
\lshad f,g \rshad & = -\lshad g,f \rshad & \textrm{(antisymmetry)}, \\
\lshad f,g \star h \rshad & = \lshad f,g \rshad \star h + g \star \lshad f,h \rshad & \textrm{(Leibniz's rule)}, \\
0 & = \lshad f, \lshad g, h \rshad \rshad + \lshad h, \lshad f, g \rshad \rshad + \lshad g, \lshad h, f \rshad \rshad & \textrm{(Jacobi's identity)}.
\end{align*}

From the definition of the $\star$-product it follows that
\begin{equation*}
B_0(f,g) = fg
\end{equation*}
and
\begin{equation*}
B_1(f,g) - B_1(g,f) = \{f,g\}.
\end{equation*}
The associativity of the $\star$-product implies that the bilinear maps $B_k$ satisfy the equations \cite{Hirshfeld:2002a, Blaszak:2003}
\begin{equation*}
\sum_{s=0}^k (B_s(B_{k-s}(f,g),h) - B_s(f,B_{k-s}(g,h))) = 0
\end{equation*}
for $k = 1,2,\ldots$. Hence, in particular $B_1$ satisfies the equation
\begin{equation*}
B_1(f,g)h - fB_1(g,h) + B_1(fg,h) - B_1(f,gh) = 0.
\end{equation*}

Let $S \colon \mathcal{A}_Q \to \mathcal{A}_Q$ be a vector space automorphism, such that
\begin{equation}
Sf = \sum_{k=0}^\infty \hbar^k S_k f, \quad S_0 = 1,
\label{eq:6.1.5}
\end{equation}
where $S_k$ are linear operators. Such an automorphism produces a new $\star'$ in $\mathcal{A}_Q$ in the following way \cite{Hirshfeld:2002a, Blaszak:2003}
\begin{equation}
f \star' g := S(S^{-1}f \star S^{-1}g).
\label{eq:6.1.6}
\end{equation}
Indeed, the associativity of the new $\star'$ follows from the associativity of the old $\star$-product, as
\begin{align*}
f \star' (g \star' h) & = f \star' S(S^{-1}g \star S^{-1}h) = S(S^{-1}f \star (S^{-1}g \star S^{-1}h)) \nonumber \\
& = S((S^{-1}f \star S^{-1}g) \star S^{-1}h) = S(S^{-1}f \star S^{-1}g) \star' h
= (f \star' g) \star' h.
\end{align*}
Using formula (\ref{eq:6.1.5}) one finds the following expression
\begin{equation*}
B_1'(f,g) = B_1(f,g) - fS_1(g) - S_1(f)g + S_1(fg).
\end{equation*}
Then
\begin{equation*}
B_1'(f,g) - B_1'(g,f) = B_1(f,g) - B_1(g,f) = \{f,g\}
\end{equation*}
and
\begin{equation*}
\lim_{\hbar \to 0} \lshad f,g \rshad' = \lim_{\hbar \to 0} \lshad f,g \rshad = \{f,g\}.
\end{equation*}
Hence, the new $\star'$ is the second well-defined $\star$-product.

Two $\star$-products: $\star$ and $\star'$ are called \emph{gauge equivalent} or simply \emph{equivalent} if there exists a vector space automorphism
$S \colon \mathcal{A}_Q \to \mathcal{A}_Q$ of the form (\ref{eq:6.1.5}) such that (\ref{eq:6.1.6}) holds. Note that such automorphism $S$ is an isomorphism
of the algebra $(\mathcal{A}_Q, \star)$ onto the algebra $(\mathcal{A}_Q, \star')$. It becomes evident that there can be many equivalent quantizations
of the classical mechanics. Even though, all this quantizations are mathematically equivalent they do not need to be physically equivalent.

\subsection{Canonical $\star_\sigma$-products}
\label{subsec:6.2}
The following section reviews standard results about a construction of a $\sigma$-family of quantizations together with its basic properties.
In the rest of the paper as the phase space $M$ the manifold $\mathbb{R}^2$ will be chosen. Moreover, for simplicity only the case of a two dimensional
phase space will be considered. Generalisation to the $2N$ dimensions is straightforward. For the Poisson tensor the canonical Poisson tensor
\begin{equation*}
\mathcal{P} = \partial_x \wedge \partial_p = \partial_x \otimes \partial_p - \partial_p \otimes \partial_x
\end{equation*}
will be taken. The corresponding Poisson bracket takes the form
\begin{equation}
\{f,g\}_{\mathcal{P}}  = f \left( \partial_x \otimes \partial_p - \partial_p \otimes \partial_x \right) g
= f \left( \overleftarrow{\partial}_x \overrightarrow{\partial}_p - \overleftarrow{\partial}_p \overrightarrow{\partial}_x \right) g
= \partial_x f \partial_p g - \partial_p f \partial_x g,
\label{eq:6.2.8}
\end{equation}
where the arrows over the vector fields $\partial_x,\partial_p$ denote that a given vector field acts only on a function standing on the left or on
the right side of the vector field.

The simplest natural deformation of the algebra $\mathcal{A}_C$ with the Poisson bracket (\ref{eq:6.2.8}) is given by such deformed $\star$-multiplication
\begin{subequations}
\label{eq:6.2.9}
\begin{align}
f \star_\sigma g & = f \exp \left( i\hbar \sigma \overleftarrow{\partial}_x \overrightarrow{\partial}_p
- i\hbar \bar{\sigma} \overleftarrow{\partial}_p \overrightarrow{\partial}_x \right) g
\label{eq:6.2.9a} \\
& = \sum_{n,m=0}^\infty (-1)^m (i\hbar)^{n+m} \frac{\sigma^n \bar{\sigma}^m}{n!m!} (\partial_x^n \partial_p^m f)(\partial_x^m \partial_p^n g)
\label{eq:6.2.9b} \\
& = \sum_{k=0}^\infty (i\hbar)^k \frac{1}{k!} \sum_{m=0}^k \binom{k}{m} \sigma^{k-m} (-\bar{\sigma})^m
(\partial_x^{k-m} \partial_p^{m} f)(\partial_x^{m} \partial_p^{k-m} g),
\label{eq:6.2.9c}
\end{align}
\end{subequations}
where $\sigma$ is some real parameter and $\bar{\sigma} = 1 - \sigma$. The product (\ref{eq:6.2.9}) is associative. Moreover, it is a well-defined
$\star$-product.

Note that two star-products $\star_\sigma$ and $\star_{\sigma'}$ are equivalent. In fact the automorphism (\ref{eq:6.1.5}) intertwining the
$\star_\sigma$-product with the $\star_{\sigma'}$-product takes the form
\begin{equation*}
S_{\sigma' - \sigma} = \exp \left( i\hbar (\sigma' - \sigma) \partial_x \partial_p \right).
\end{equation*}

The well-known particular cases of the product (\ref{eq:6.2.9}) are: for $\sigma = 0$ the Kupershmidt-Manin product \cite{Kupershmidt:1990},
for $\sigma = \frac{1}{2}$ the Moyal (or Groenewold) product \cite{Moyal:1949, Groenewold:1946}.

The $\star_\sigma$-product (\ref{eq:6.2.9}) can be written in the following integral form (see \cite{Hirshfeld:2002a} for a derivation in the case of
the Moyal product, the general case can be derived analogically)
\begin{align}
(f \star_\sigma g)(x,p) & = \frac{1}{2\pi\hbar} \iint \mathcal{F}f(\xi,\eta) g(x - \bar{\sigma} \eta, p - \sigma \xi)
e^{\frac{i}{\hbar}(\xi x - \eta p)} \d{\xi} \d{\eta} \nonumber \displaybreak[0] \\
& = \frac{1}{2\pi\hbar} \iint f(x + \sigma \eta, p + \bar{\sigma} \xi) \mathcal{F}g(\xi,\eta)
e^{\frac{i}{\hbar}(\xi x - \eta p)} \d{\xi} \d{\eta}.
\label{eq:6.2.10}
\end{align}
For a special cases $\sigma \neq 0,1$ the above formula, after performing an appropriate change of variables under the integral sign, can be written
differently
\begin{equation*}
(f \star_\sigma g)(x,p) = \frac{1}{(2\pi\hbar)^2 |\sigma \bar{\sigma}|} \iiiint f(x',p') g(x'',p'')
e^{\frac{i}{\hbar} \sigma^{-1} (x - x')(p - p'')} e^{-\frac{i}{\hbar} \bar{\sigma}^{-1} (p - p')(x - x'')} \d{x'} \d{p'} \d{x''} \d{p''}.
\end{equation*}

Using equations (\ref{eq:6.2.9b}) and (\ref{eq:6.2.9c}), and the integration by parts a useful property of the $\star_\sigma$-product can be derived.
Namely, there holds
\begin{theorem}
\label{thm:6.2.7}
Let $f,g \in \mathcal{A}_Q$ be such that $f \star_\sigma g$ and $g \star_\sigma f$ are integrable functions. Then there holds
\begin{equation*}
\iint (f \star_\sigma g)(x,p) \d{x} \d{p} = \iint (g \star_\sigma f)(x,p) \d{x} \d{p}.
\end{equation*}
Moreover, for the Moyal $\star$-product (the case of $\sigma = \frac{1}{2}$) there holds
\begin{equation*}
\iint (f \star_{\frac{1}{2}} g)(x,p) \d{x} \d{p} = \iint f(x,p)g(x,p) \d{x} \d{p}.
\end{equation*}
\end{theorem}

From equation (\ref{eq:6.2.9b}) immediately follows another two interesting properties of the $\star_\sigma$-product. Namely, there holds
\begin{theorem}
For $f,g \in \mathcal{A}_Q$ there holds
\begin{gather*}
(f \star_\sigma g)^* = g^* \star_{\bar{\sigma}} f^*, \\
\partial_x(f \star_\sigma g) = (\partial_x f) \star_\sigma g + f \star_\sigma (\partial_x g), \qquad
\partial_p(f \star_\sigma g) = (\partial_p f) \star_\sigma g + f \star_\sigma (\partial_p g).
\end{gather*}
In particular, for the case of $\sigma = \frac{1}{2}$ the complex conjugation of functions is an involution of the algebra $\mathcal{A}_Q$.
\end{theorem}

The involution $\dagger$ of the algebra $(\mathcal{A}_Q, \star_\sigma)$ will be defined by the equation
\begin{equation*}
A^\dagger := S_{\sigma - \bar{\sigma}} A^*, \quad A \in \mathcal{A}_Q.
\end{equation*}
It is easy to check that the involution defined in such a way is a proper involution of the algebra $\mathcal{A}_Q$. Note that for $\sigma = \frac{1}{2}$
the involution $\dagger$ is the usual complex conjugation of functions.

\subsection{Space of states}
\label{subsec:6.3}
In what follows the problem of defining a space of states will be discussed. In analogy with the classical Hamiltonian mechanics one could try to
define admissible states of the quantum Hamiltonian system as probabilistic distributions on the phase space. After doing this one would quickly realize
that it is necessary to extend the space of admissible states to \emph{pseudo-probabilistic distributions}, i.e. complex valued functions on the phase
space which are normalized but need not to be non-negative. Hence, it is postulated that the space, which contains all admissible states is the space
$L^2(M)$ of all square integrable functions on the phase space $M = \mathbb{R}^2$ with respect to the Lebesgue measure.

It is possible to introduce the $\star_\sigma$-product between functions from $L^2(M)$, as to make from $L^2(M)$ an algebra with respect to the
$\star_\sigma$-multiplication \cite{Gracia-Bondia:1988}. First, note that, by using the integral form (\ref{eq:6.2.10}) of the
$\star_\sigma$-product, the $\star_\sigma$-product of two Schwartz functions can be defined. Moreover, the $\star_\sigma$-product of two
Schwartz functions is again a Schwartz function, hence the Schwartz space $\mathcal{S}(M)$ is an algebra with respect to the
$\star_\sigma$-multiplication. Indeed, the Schwartz space $\mathcal{S}(M)$ is the space of all smooth functions $f \in C^\infty(M)$ such that
$\| x^n p^m \partial_x^r \partial_p^s f \|_\infty = \displaystyle{\sup_{(x,p) \in M}} | x^n p^m \partial_x^r \partial_p^s f(x,p) | < \infty$ for every
$n,m,r,s \in \mathbb{N}$. For $f,g \in \mathcal{S}(M)$ from (\ref{eq:6.2.10}) it immediately follows that
\begin{align*}
\partial_x(f \star_\sigma g) & = (\partial_x f) \star_\sigma g + f \star_\sigma (\partial_x g), \displaybreak[0] \\
\partial_p(f \star_\sigma g) & = (\partial_p f) \star_\sigma g + f \star_\sigma (\partial_p g), \displaybreak[0] \\
x (f \star_\sigma g) & = f \star_\sigma (x g) + i\hbar \bar{\sigma} (\partial_p f) \star_\sigma g
= (x f) \star_\sigma g - i\hbar \sigma f \star_\sigma (\partial_p g), \displaybreak[0] \\
p (f \star_\sigma g) & = f \star_\sigma (p g) - i\hbar \sigma (\partial_x f) \star_\sigma g
= (p f) \star_\sigma g + i\hbar \bar{\sigma} f \star_\sigma (\partial_x g).
\end{align*}
By induction on these formulas, $f \star_\sigma g \in C^\infty(M)$ and $\| x^n p^m \partial_x^r \partial_p^s (f \star_\sigma g) \|_\infty < \infty$ for
every $n,m,r,s \in \mathbb{N}$.

The below theorem says about the possibility of extension of the $\star_\sigma$-product to the whole space $L^2(M)$.
\begin{theorem}
For $\Psi,\Phi \in \mathcal{S}(M)$ there holds
\begin{equation}
\| \Psi \star_\sigma \Phi \|_{L^2} \le \frac{1}{(2\pi\hbar)^{1/2}} \| \Psi \|_{L^2} \| \Phi \|_{L^2}.
\label{eq:6.3.24}
\end{equation}
There exists an unique extension of the $\star_\sigma$-product from the space $\mathcal{S}(M)$ to the whole space $L^2(M)$, such that relation
(\ref{eq:6.3.24}) holds.
\end{theorem}
\begin{proof}
First, lets prove that the $\star_\sigma$-product on $\mathcal{S}(M)$ is separately continuous in the $L^2$-norm (i.e. the $\star_\sigma$-product as a
map $\mathcal{S}(M) \times \mathcal{S}(M) \to \mathcal{S}(M)$ is continuous with respect to the first and second argument separately). Note that from
Jensen's inequality for $f,g \in \mathcal{S}(M)$ there holds (see \ref{sec:3})
\begin{equation*}
\left| \iint f(x,p) g(x,p) \d{x} \d{p} \right|^2 \le \iint |g(x,p)| \d{x} \d{p} \iint |f(x,p)|^2 |g(x,p)| \d{x} \d{p}.
\end{equation*}
For $\Psi,\Phi \in \mathcal{S}(M)$ from the above equation and the integral form (\ref{eq:6.2.10}) of the $\star_\sigma$-product there holds
\begin{align*}
\| \Psi \star_\sigma \Phi \|_{L^2}^2 & = \iint | \Psi \star_\sigma \Phi |^2 \d{x} \d{p}
= \frac{1}{(2\pi\hbar)^{2}} \iint \left| \iint \mathcal{F}\Psi(\xi,\eta) \Phi(x - \bar{\sigma} \eta,p - \sigma \xi)
e^{\frac{i}{\hbar}(\xi x - \eta p)} \d{\xi} \d{\eta} \right|^2 \d{x} \d{p} \displaybreak[0] \\
& \le \frac{1}{(2\pi\hbar)^{2}} \iint \| \mathcal{F}\Psi \|_{L^1} \iint |\Phi(x - \bar{\sigma} \eta,p - \sigma \xi)|^2 |\mathcal{F}\Psi(\xi,\eta)|
\d{\xi} \d{\eta} \d{x} \d{p}
= \frac{1}{(2\pi\hbar)^{2}} \| \mathcal{F}\Psi \|_{L^1}^2 \| \Phi \|_{L^2}^2.
\end{align*}
Hence
\begin{equation*}
\| \Psi \star_\sigma \Phi \|_{L^2} \le \frac{1}{2\pi\hbar} \| \mathcal{F}\Psi \|_{L^1} \| \Phi \|_{L^2}.
\end{equation*}
Analogically, one proves that
\begin{equation*}
\| \Psi \star_\sigma \Phi \|_{L^2} \le \frac{1}{2\pi\hbar} \| \mathcal{F}\Phi \|_{L^1} \| \Psi \|_{L^2}.
\end{equation*}
The above equations show that the $\star_\sigma$-product on $\mathcal{S}(M)$ is separately continuous in the $L^2$-norm.

Now, assume that $\Psi_{ij} \in \mathcal{S}(M)$ is an orthonormal basis in $L^2(M)$ satisfying
\begin{equation}
\Psi_{ij} \star_\sigma \Psi_{kl} = \frac{1}{(2\pi\hbar)^{1/2}} \delta_{il} \Psi_{kj}.
\label{eq:6.3.26}
\end{equation}
Such basis always exists (see Section \ref{sec:9}). Every $\Psi \in L^2(M)$ can be expanded in this basis
\begin{equation*}
\Psi = \sum_{i,j=0}^\infty c_{ij} \Psi_{ij},
\end{equation*}
where the convergence is in the $L^2$-norm. First note that for $\Psi = \sum_{i,j} c_{ij} \Psi_{ij} \in \mathcal{S}(M)$ and
$\Phi = \sum_{k,l} b_{kl} \Psi_{kl} \in \mathcal{S}(M)$, using (\ref{eq:6.3.26}) and the continuity of the $\star_\sigma$-product in the
$L^2$-norm, the $\star_\sigma$-product of two Schwartz functions can be written in a form
\begin{equation}
\Psi \star_\sigma \Phi = \left( \sum_{i,j=0}^\infty c_{ij} \Psi_{ij} \right) \star_\sigma \left( \sum_{k,l=0}^\infty b_{kl} \Psi_{kl} \right)
= \frac{1}{(2\pi\hbar)^{1/2}} \sum_{i,j,k=0}^\infty c_{ij} b_{ki} \Psi_{kj}.
\label{eq:6.3.27}
\end{equation}
Now, for $\Psi = \sum_{i,j} c_{ij} \Psi_{ij} \in L^2(M)$ and $\Phi = \sum_{k,l} b_{kl} \Psi_{kl} \in L^2(M)$ lets define the
$\star_\sigma$-product of functions $\Psi$ and $\Phi$ by the formula
\begin{equation*}
\Psi \star_\sigma \Phi = \frac{1}{(2\pi\hbar)^{1/2}} \sum_{k,j=0}^\infty \left( \sum_{i=0}^\infty c_{ij} b_{ki} \right) \Psi_{kj}.
\end{equation*}
From (\ref{eq:6.3.27}) above definition of the $\star_\sigma$-product, for Schwartz functions, is consistent with the previous one.

From Schwartz inequality it follows that
\begin{equation*}
\| \Psi \star_\sigma \Phi \|_{L^2}^2 \le \frac{1}{2\pi\hbar} \sum_{k,j=0}^\infty \left( \sum_{i=0}^\infty |c_{ij}| |b_{ki}| \right)^2
\le \frac{1}{2\pi\hbar} \sum_{i,j=0}^\infty |c_{ij}|^2 \sum_{k,l=0}^\infty |b_{kl}|^2
 = \frac{1}{2\pi\hbar} \| \Psi \|_{L^2}^2 \| \Phi \|_{L^2}^2.
\end{equation*}
Hence, the $\star_\sigma$-product in $L^2(M)$ satisfies the relation (\ref{eq:6.3.24}). The uniqueness of the presented extension of the
$\star_\sigma$-product is evident from the fact that $\mathcal{S}(M)$ is dense in $L^2(M)$.
\end{proof}

\subsection{Weyl operator calculus for $\star_\sigma$-family of quantizations}
\label{subsec:6.4}
During a quantization procedure it is needed to assign to a function of commuting variables respective operators, i.e. admissible functions of
noncommuting variables. This can be achieved be means of \emph{Weyl operator (pseudo-differential) calculus}. The Weyl correspondence rule, associating
with a symbol (a function defined on a phase space) an operator defined on a Hilbert space, can be expressed by a unitary irreducible representation of
the Heisenberg group. For more details refer to \cite{Agarwal:1970, Gosson:2005}. In this section the equivalent (and slightly simpler) approach to the
Weyl calculus will be considered, following the results of \citet{Cohen:1966}.

From definition, an \emph{operator function} assigned to the function $A$ is given by the equation
\begin{equation}
A_{\sigma}(\hat{q},\hat{p}) := \frac{1}{2\pi\hbar} \iint \mathcal{F}A(\xi,\eta) e^{\frac{i}{\hbar}(\xi \hat{q} - \eta \hat{p})}
e^{\frac{i}{\hbar}(\frac{1}{2} - \sigma) \xi \eta} \d{\xi} \d{\eta},
\label{eq:6.4.1}
\end{equation}
where $\sigma \in \mathbb{R}$ is a parameter describing different orderings. Examples of most common orderings are: for the case $\sigma = 0$ the
\emph{standard ordering} \cite{Mehta:1964} (operator function will have all position operators on the left and all momentum operators on the right),
for the case $\sigma = 1$ the \emph{anti-standard ordering} (operator function will have all position operators on the right and all momentum
operators on the left), for the case $\sigma = \frac{1}{2}$ the \emph{symmetric ordering} (\emph{Weyl ordering}) \cite{Wigner:1932} (operator function
will be arithmetic average of all possible permutations of position and momentum operators).

Equation (\ref{eq:6.4.1}), using the Baker-Campbell-Hausdorff formula (see \ref{sec:1}), can be rewritten in a form
\begin{equation}
A_{\sigma}(\hat{q},\hat{p}) := \frac{1}{2\pi\hbar} \iint \mathcal{F}A(\xi,\eta) e^{\frac{i}{\hbar}\xi \hat{q}}
e^{-\frac{i}{\hbar} \eta \hat{p}} e^{-\frac{i}{\hbar} \sigma \xi \eta} \d{\xi} \d{\eta}.
\label{eq:6.4.15}
\end{equation}

To illustrate different orderings lets consider an examples of the operator function. Using (\ref{eq:6.4.1}) it can be easily calculated that for the
function $A(x,p) = xp^2$ the operator function is equal
\begin{align*}
A_{\sigma}(\hat{q},\hat{p}) & = (\hat{q} \hat{p}^2)_{\sigma} = \hat{q} \hat{p}^2 - 2i\hbar \sigma \hat{p} \\
& = \hat{q} \hat{p}^2 - 2\sigma \lambda [\hat{q},\hat{p}] \hat{p} - 2\sigma(1 - \lambda) \hat{p} [\hat{q},\hat{p}] \\
& = (1 - 2\sigma \lambda)\hat{q} \hat{p}^2 + 2\sigma (2\lambda - 1) \hat{p} \hat{q} \hat{p} + 2\sigma(1 - \lambda) \hat{p}^2 \hat{q}.
\end{align*}
In this case the $\lambda$-family of $\sigma$-orderings was received for $\lambda \in \mathbb{R}$.
The case when $\sigma = \frac{1}{2}$ (the Weyl ordering) gives
\begin{equation*}
(\hat{q} \hat{p}^2)_{\sigma=\frac{1}{2}} = (1 - \lambda)\hat{q} \hat{p}^2 + (2\lambda - 1) \hat{p} \hat{q} \hat{p}
+ (1 - \lambda) \hat{p}^2 \hat{q}.
\end{equation*}
In particular
\begin{equation*}
(\hat{q} \hat{p}^2)_{\sigma=\frac{1}{2}} = \frac{1}{3} \hat{q} \hat{p}^2 + \frac{1}{3} \hat{p} \hat{q} \hat{p}
+ \frac{1}{3} \hat{p}^2 \hat{q}
= \frac{1}{2} \hat{q} \hat{p}^2 + \frac{1}{2} \hat{p}^2 \hat{q} = \hat{p} \hat{q} \hat{p},
\end{equation*}
for $\lambda = \frac{2}{3},\frac{1}{2},1$.

There is a useful property of operator functions. Namely, there holds
\begin{equation}
A_{\sigma}^\dagger(\hat{q},\hat{p}) = A^*_{\bar{\sigma}}(\hat{q},\hat{p})
\label{eq:6.4.12}
\end{equation}
Indeed
\begin{equation*}
A_{\sigma}^\dagger(\hat{q},\hat{p}) = \frac{1}{2\pi\hbar} \iint \mathcal{F}A^*(-\xi,-\eta) e^{-\frac{i}{\hbar}(\xi \hat{q}
- \eta \hat{p} + (\frac{1}{2} - \sigma) \xi \eta)} \d{\xi} \d{\eta}
\end{equation*}
and above equation after the change of coordinates: $\xi \to -\xi$, $\eta \to -\eta$ can be written in a form
\begin{equation*}
A_{\sigma}^\dagger(\hat{q},\hat{p}) = \frac{1}{2\pi\hbar} \iint \mathcal{F}A^*(\xi,\eta) e^{\frac{i}{\hbar}(\xi \hat{q}
- \eta \hat{p} + (\frac{1}{2} - \bar{\sigma}) \xi \eta)} \d{\xi} \d{\eta}
= A^*_{\bar{\sigma}}(\hat{q},\hat{p}).
\end{equation*}

For further use it will be useful to introduce an operator function from Hermitian operators $\underline{\hat{q}}$, $\underline{\hat{p}}$ satisfying the
following commutation relation
\begin{equation*}
[\underline{\hat{q}}, \underline{\hat{p}}] = -i\hbar.
\end{equation*}
This will be needed below where it will be shown that operators of the form ${} \star A$ ($A \in \mathcal{A}_Q$) can be written as
operator functions of appropriate operators $\underline{\hat{q}}$, $\underline{\hat{p}}$.
For this purpose the same defining equation (equation~(\ref{eq:6.4.1})) can be used as in the previous definition of operator functions. Using the
Baker-Campbell-Hausdorff formula equation~(\ref{eq:6.4.1}) can be now rewritten in the form
\begin{equation}
A_{\sigma}(\underline{\hat{q}},\underline{\hat{p}}) := \frac{1}{2\pi\hbar} \iint \mathcal{F}A(\xi,\eta)
e^{\frac{i}{\hbar}\xi \underline{\hat{q}}} e^{-\frac{i}{\hbar} \eta \underline{\hat{p}}}
e^{\frac{i}{\hbar} \bar{\sigma} \xi \eta} \d{\xi} \d{\eta}.
\label{eq:6.4.16}
\end{equation}
From above equation it can be seen that in this case the roles of $\sigma$ and $\bar{\sigma}$ are reversed. For example an operator function is standard
ordered for $\sigma = 1$ and anti-standard ordered for $\sigma = 0$. All properties and equations derived for the previous case can be rederived in
a similar way for this case.

Using the integral representation (\ref{eq:6.2.10}) of the $\star_\sigma$-product it is possible to define a left and right $\star_\sigma$-product of
a function $A \in \mathcal{A}_Q$ with functions from some subspace of the space of states $L^2(M)$ receiving again a function from $L^2(M)$. Note that
in general the function $A \in \mathcal{A}_Q$ cannot be multiplied by every function from $L^2(M)$ in such a way, as to receive again a function from
$L^2(M)$. Above arguments state that operators from the algebra $\hat{\mathcal{A}}_Q$, hence in particular observables, can be treated as operators
defined on the Hilbert space $L^2(M)$.

It can be shown that a left action of some function $A \in \mathcal{A}_Q$ through the $\star_\sigma$-product on some function
$\Psi \in L^2(M)$ can be treated as an action of an operator function $A_{\sigma}(\hat{q}_{\sigma}, \hat{p}_{\sigma})$ on
function $\Psi$ where $\hat{q}_{\sigma},\hat{p}_{\sigma}$ are some operators defined on the Hilbert space $L^2(M)$. Similarly, a right
action of $A \in \mathcal{A}_Q$ through the $\star_{\sigma}$-product on $\Psi \in L^2(M)$ can be treated as an action of an operator
function $A_{\sigma}(\hat{q}^*_{\bar{\sigma}}, \hat{p}^*_{\bar{\sigma}})$ on function $\Psi$. First, note that by using
equation~(\ref{eq:6.2.9a}) and the identity
\begin{equation}
e^{a \partial_x} f(x) = f(x + a), \quad a \in \mathbb{R}
\label{eq:6.4.23}
\end{equation}
valid for any smooth function $f \colon \mathbb{R} \to \mathbb{C}$, the $\star_\sigma$-product of functions $A \in \mathcal{A}_Q$ and
$\Psi \in L^2(M)$ can be formally written in a form
\begin{align*}
A_L \star_\sigma \Psi := A \star_\sigma \Psi & = A(x + i\hbar \sigma \overrightarrow{\partial}_p, p - i \hbar \bar{\sigma} \overrightarrow{\partial}_x)
\Psi, \\
A_R \star_\sigma \Psi := \Psi \star_\sigma A & = A(x - i\hbar \bar{\sigma} \overrightarrow{\partial}_p, p + i \hbar \sigma \overrightarrow{\partial}_x)
\Psi.
\end{align*}
Lets define the operators $\hat{q}_{\sigma},\hat{p}_{\sigma}$ \cite{Chruscinski:2005, Gosson:2005} by the equations
\begin{equation*}
\hat{q}_{\sigma} := x + i\hbar \sigma \partial_p = x \star_\sigma {}, \qquad
\hat{p}_{\sigma} := p - i\hbar \bar{\sigma} \partial_x = p \star_\sigma {}.
\end{equation*}
It can be easily checked that the operators $\hat{q}_{\sigma},\hat{p}_{\sigma}$ satisfy the commutation relation $[\hat{q}_{\sigma}, \hat{p}_{\sigma}]
= i\hbar$. The operators $\hat{q}^*_{\bar{\sigma}},\hat{p}^*_{\bar{\sigma}}$ take then the form
\begin{equation*}
\hat{q}^*_{\bar{\sigma}} := x - i\hbar \bar{\sigma} \partial_p = {} \star_\sigma x, \qquad
\hat{p}^*_{\bar{\sigma}} := p + i\hbar \sigma \partial_x = {} \star_\sigma p,
\end{equation*}
and they satisfy the commutation relation $[\hat{q}^*_{\bar{\sigma}}, \hat{p}^*_{\bar{\sigma}}] = -i\hbar$. There holds
\begin{theorem}
\label{thm:6.4.1}
For any function $A \in \mathcal{A}_Q$
\begin{equation*}
A_L \star_\sigma {} = A(\overrightarrow{\hat{q}}_{\sigma},\overrightarrow{\hat{p}}_{\sigma})
= A_{\sigma}(\hat{q}_{\sigma},\hat{p}_{\sigma}), \qquad
A_R \star_\sigma {} = A(\overrightarrow{\hat{q}}^*_{\bar{\sigma}},\overrightarrow{\hat{p}}^*_{\bar{\sigma}})
= A_{\sigma}(\hat{q}^*_{\bar{\sigma}},\hat{p}^*_{\bar{\sigma}}).
\end{equation*}
\end{theorem}
\begin{proof}
Above theorem can be proved using the integral form (\ref{eq:6.2.10}) of the $\star_\sigma$-product. First, lets check how the
operator $e^{\frac{i}{\hbar} \xi \hat{q}_{\sigma}} e^{-\frac{i}{\hbar} \eta \hat{p}_{\sigma}}$ acts on some function
$\Psi \in L^2(M)$. Using the Baker-Campbell-Hausdorff formulae (see \ref{sec:1}) and the identity~(\ref{eq:6.4.23}) one receives
\begin{align*}
e^{\frac{i}{\hbar} \xi \hat{q}_{\sigma}} e^{-\frac{i}{\hbar} \eta \hat{p}_{\sigma}} \Psi(x,p)
& = e^{\frac{i}{\hbar} \xi (x + i\hbar \sigma \partial_p)} e^{-\frac{i}{\hbar} \eta (p - i\hbar \bar{\sigma} \partial_x)} \Psi(x,p)
= e^{\frac{i}{\hbar} \xi x} e^{-\sigma \xi \partial_p} e^{-\frac{i}{\hbar} \eta p} e^{-\bar{\sigma} \eta \partial_x} \Psi(x,p) \\
& = e^{\frac{i}{\hbar} \xi x} e^{-\frac{i}{\hbar} \eta p} e^{\frac{i}{\hbar} \sigma \xi \eta} e^{-\sigma \xi \partial_p}
e^{-\bar{\sigma} \eta \partial_x} \Psi(x,p)
= e^{\frac{i}{\hbar} \xi x} e^{-\frac{i}{\hbar} \eta p} e^{\frac{i}{\hbar} \sigma \xi \eta} \Psi(x - \bar{\sigma} \eta, p - \sigma \xi).
\end{align*}
Using above equation and equation (\ref{eq:6.4.15}) it follows immediately that
\begin{equation*}
A_{\sigma}(\hat{q}_{\sigma},\hat{p}_{\sigma}) \Psi(x,p) = \frac{1}{2\pi\hbar} \iint \mathcal{F}A(\xi,\eta) \Psi(x - \bar{\sigma} \eta, p - \sigma \xi)
e^{\frac{i}{\hbar} \xi x} e^{-\frac{i}{\hbar} \eta p} \d{\xi} \d{\eta},
\end{equation*}
which is just the integral form (\ref{eq:6.2.10}) of the product $A \star_\sigma \Psi$.
\end{proof}

From Theorem \ref{thm:6.4.1} it follows that operators from the algebra $\hat{\mathcal{A}}_Q$, hence in particular observables, can be written as operator
functions of the operators $\hat{q}_{\sigma},\hat{p}_{\sigma}$.

It is possible to introduce adjoint of left and right $\star_\sigma$-multiplication in a standard way
\begin{equation*}
\braket{(A_L \star_\sigma {})^\dagger \Psi_1 | \Psi_2}_{L^2} = \braket{\Psi_1 | A_L \star_\sigma \Psi_2}_{L^2}, \qquad
\braket{(A_R \star_\sigma {})^\dagger \Psi_1 | \Psi_2}_{L^2} = \braket{\Psi_1 | A_R \star_\sigma \Psi_2}_{L^2}.
\end{equation*}
From this it then follows that
\begin{subequations}
\label{eq:6.4.17}
\begin{align}
(A_L \star_\sigma {})^\dagger & = A^\dagger_{\sigma}(\hat{q}_{\sigma},\hat{p}_{\sigma})
= A^*_{\bar{\sigma}}(\hat{q}_{\sigma},\hat{p}_{\sigma}),
\label{eq:6.4.17a} \\
(A_R \star_\sigma {})^\dagger & = A^\dagger_{\sigma}(\hat{q}^*_{\bar{\sigma}},\hat{p}^*_{\bar{\sigma}})
= A^*_{\bar{\sigma}}(\hat{q}^*_{\bar{\sigma}},\hat{p}^*_{\bar{\sigma}}).
\label{eq:6.4.17b}
\end{align}
\end{subequations}
It is easy to check that
\begin{equation*}
(A_L \star_\sigma {})^\dagger = A_L^\dagger \star_\sigma {}, \qquad
(A_R \star_\sigma {})^\dagger = A_R^\dagger \star_\sigma {}.
\end{equation*}

\subsection{Family of admissible canonical quantizations}
\label{subsec:6.5}
In the previous sections the $\sigma$-family of equivalent admissible quantizations of the classical Hamiltonian mechanics was reviewed, together with
its basic properties. In the following section the $\sigma$-family of quantizations will be extended to a very broad class of equivalent admissible
quantizations, which shares the same properties as the $\sigma$-family of quantizations. This extended family of quantizations will contain all canonical
quantizations ever considered in the literature, and more.

Let $S$ be a linear automorphism of the algebra of observables $\mathcal{A}_Q$, of the form (\ref{eq:6.1.5}). Moreover, lets assume that $S$ satisfies
\begin{equation}
Sx = x, \quad Sp = p
\label{eq:6.5.1}
\end{equation}
and that $S$ vanishes under the integral sign, i.e. for $f \in \mathcal{A}_Q$ such that $f$ and $Sf$ are integrable
\begin{equation}
\iint Sf(x,p) \d{x}\d{p} = \iint f(x,p) \d{x}\d{p}.
\label{eq:6.5.2}
\end{equation}
The automorphism $S$ can be used to define a new $\star$-product equivalent to the $\star_\sigma$-product by the equation
\begin{equation*}
f \star_{\sigma,S} g := S(S^{-1}f \star_\sigma S^{-1}g), \quad f,g \in \mathcal{A}_Q.
\end{equation*}
It will be shown that all important properties of the $\star_\sigma$-product are also valid for the $\star_{\sigma,S}$-product.

First, lets define a conjugation $\bar{S}$ of the automorphism $S$ by the equation
\begin{equation*}
\bar{S}f := (Sf^*)^*, \quad f \in \mathcal{A}_Q.
\end{equation*}
$\bar{S}$ is also an automorphism of the algebra $\mathcal{A}_Q$, which satisfies relations (\ref{eq:6.5.1}) and (\ref{eq:6.5.2}). There holds
\begin{theorem}
Let $S$, $S_1$, $S_2$ be automorphisms of the algebra $\mathcal{A}_Q$. Then
\begin{equation*}
\overline{S_1 S_2} = \bar{S}_1 \bar{S}_2
\end{equation*}
and in particular
\begin{equation*}
\overline{S^{-1}} = \bar{S}^{-1}.
\end{equation*}
\end{theorem}
\begin{proof}
For $f \in \mathcal{A}_Q$
\begin{equation*}
\overline{S_1 S_2}f = (S_1 S_2 f^*)^* = \bar{S}_1 (S_2 f^*)^* = \bar{S}_1 \bar{S}_2 f.
\end{equation*}
\end{proof}

The $\star_{\sigma,S}$-product commutes under the integral sign. Indeed
\begin{theorem}
Let $f,g \in \mathcal{A}_Q$ be such that $f \star_{\sigma,S} g$ and $g \star_{\sigma,S} f$ are integrable functions. Then there holds
\begin{equation*}
\iint (f \star_{\sigma,S} g)(x,p) \d{x}\d{p} = \iint (g \star_{\sigma,S} f)(x,p) \d{x}\d{p}.
\end{equation*}
\end{theorem}
\begin{proof}
Using the property (\ref{eq:6.5.2}) and the fact that the $\star_\sigma$-product commutes under the integral sign one gets
\begin{align*}
\iint (f \star_{\sigma,S} g)(x,p) \d{x}\d{p} & = \iint SS^{-1}(f \star_{\sigma,S} g)(x,p) \d{x}\d{p}
= \iint S(S^{-1}f \star_\sigma S^{-1}g)(x,p) \d{x}\d{p} \\
& = \iint (S^{-1}f \star_\sigma S^{-1}g)(x,p) \d{x}\d{p} = \iint (S^{-1}g \star_\sigma S^{-1}f)(x,p) \d{x}\d{p} \\
& = \iint S(S^{-1}g \star_\sigma S^{-1}f)(x,p) \d{x}\d{p} = \iint SS^{-1}(g \star_{\sigma,S} f)(x,p) \d{x}\d{p} \\
& = \iint (g \star_{\sigma,S} f)(x,p) \d{x}\d{p}.
\end{align*}
\end{proof}

Another interesting property of the $\star_{\sigma,S}$-product can be derived. Namely
\begin{theorem}
For $f,g \in \mathcal{A}_Q$ there holds
\begin{equation*}
(f \star_{\sigma,S} g)^* = g^* \star_{\bar{\sigma},\bar{S}} f^*.
\end{equation*}
\end{theorem}
\begin{proof}
\begin{align*}
(f \star_{\sigma,S} g)^* & = (SS^{-1}(f \star_{\sigma,S} g))^* = (S(S^{-1}f \star_\sigma S^{-1}g))^* = \bar{S}((S^{-1}g)^* \star_{\bar{\sigma}} (S^{-1}f)^*)
= \bar{S}(\bar{S}^{-1} g^* \star_{\bar{\sigma}} \bar{S}^{-1} f^*) \\
& = g^* \star_{\bar{\sigma},\bar{S}} f^*.
\end{align*}
\end{proof}

By means of the automorphism $S$ the involution of the algebra $(\mathcal{A}_Q, \star_{\sigma,S})$ can be defined by the formula
\begin{equation*}
A^\dagger := S(S^{-1}A)^\dagger = S S_{\sigma - \bar{\sigma}} \bar{S}^{-1} A^*.
\end{equation*}

In what follows the space of states of the quantum Hamiltonian system associated to the $\star_{\sigma,S}$-product will be defined. First note that the
automorphism $S \colon \mathcal{A}_Q \to \mathcal{A}_Q$ can be restricted to the Schwartz space $\mathcal{S}(M)$, i.e. one receives the following map:
$S \colon \mathcal{S}(M) \to C^\infty(M)$. To simplify considerations, in the rest of the paper it will be assumed that $S$ has values in the Hilbert
space $L^2(M,\mu)$, where $\mu$ is some measure on $M$. Moreover, it will be assumed that $S$ is a continuous mapping with respect to $L^2$-topology.
Now, since $\mathcal{S}(M)$ is dense in $L^2(M)$ the map $S \colon \mathcal{S}(M) \to L^2(M,\mu)$ can be uniquely extended to the continuous map $S$
defined on the whole Hilbert space $L^2(M)$ (being the space of states associated to the $\star_\sigma$-product) and taking values onto the Hilbert space
$\mathcal{H} = L^2(M,\mu)$ with a scalar product given by
\begin{equation*}
\braket{\Phi|\Psi}_\mathcal{H} = \braket{S^{-1} \Phi | S^{-1} \Psi}_{L^2}, \quad \Phi,\Psi \in \mathcal{H}.
\end{equation*}
Note that it is not assumed that the scalar product $\braket{\,\cdot\,|\,\cdot\,}_\mathcal{H}$ is equal to the standard scalar product of the Hilbert
space $L^2(M,\mu)$. The Hilbert space $\mathcal{H}$ is the space of states of the quantum Hamiltonian system associated to the $\star_{\sigma,S}$-product.
Note that the isomorphism $S$ induces the $\star_{\sigma,S}$-product on $\mathcal{H}$ from the $\star_\sigma$-product on $L^2(M)$ by the formula
\begin{equation*}
\Phi \star_{\sigma,S} \Psi := S(S^{-1}\Phi \star_\sigma S^{-1}\Psi), \quad \Phi,\Psi \in \mathcal{H},
\end{equation*}
which satisfies the relation
\begin{equation}
\| \Phi \star_{\sigma,S} \Psi \|_\mathcal{H} \le \frac{1}{\sqrt{2\pi\hbar}} \| \Phi \|_\mathcal{H} \| \Psi \|_\mathcal{H}, \quad \Phi,\Psi \in \mathcal{H}.
\label{eq:6.5.7}
\end{equation}
From above relation it is evident that the $\star_{\sigma,S}$-product is a continuous mapping $\mathcal{H} \times \mathcal{H} \to \mathcal{H}$ and that
$\mathcal{H}$ has a structure of a Hilbert algebra.

From the assumption that $S \colon \mathcal{A}_Q \to \mathcal{A}_Q$ vanishes under the integral sign, the induced Hilbert space isomorphism
$S \colon L^2(M) \to L^2(M,\mu)$ also vanishes under the integral sign. Indeed, it is clear that $S \colon L^2(M) \to L^2(M,\mu)$ vanishes under the
integral sign for all Schwartz functions $\Psi$, provided that $S \Psi$ is integrable. From this observation the general statement follows easily,
as for $\Psi \in L^2(M)$ integrable such that $S \Psi$ is also integrable and noting that $\Psi$ can be written as a limit of Schwartz functions $\Psi_n$
one obtains
\begin{align*}
\iint (S \Psi)(x,p) \d{x}\d{p} & = \iint S \left(\lim_{n \to \infty} \Psi_n \right)(x,p) \d{x}\d{p} = \iint \lim_{n \to \infty} S(\Psi_n)(x,p) \d{x}\d{p}
= \lim_{n \to \infty} \iint S(\Psi_n)(x,p) \d{x}\d{p} \\
& = \lim_{n \to \infty} \iint \Psi_n(x,p) \d{x}\d{p} = \iint \lim_{n \to \infty} \Psi_n(x,p) \d{x}\d{p} = \iint \Psi(x,p) \d{x}\d{p}.
\end{align*}

Similarly as in the case of the $\star_\sigma$-product lets define the $\star_{\sigma,S}$-product of an observable $A \in \mathcal{A}_Q$ and a state
$\Psi \in \mathcal{H}$. The $\star_{\sigma,S}$-product of $A$ and $\Psi$ will be defined by
\begin{equation*}
A \star_{\sigma,S} \Psi := S(S^{-1}A \star_\sigma S^{-1}\Psi), \qquad
\Psi \star_{\sigma,S} A := S(S^{-1}\Psi \star_\sigma S^{-1}A).
\end{equation*}

A description of the left and right action of some observable $A \in \mathcal{A}_Q$ through the $\star_{\sigma,S}$-product requires the generalisation
of the $\sigma$-ordering. Thus, for $A \in \mathcal{A}_Q$ and some operators $\hat{q}$, $\hat{p}$ defined on some Hilbert space $\mathcal{H}$ and
satisfying the canonical commutation relation
\begin{equation}
[\hat{q}, \hat{p}] = i\hbar
\label{eq:6.5.3}
\end{equation}
lets introduce a $(\sigma,S)$-ordered operator function by the formula
\begin{equation*}
A_{\sigma,S}(\hat{q},\hat{p}) := (S^{-1}A)_\sigma(\hat{q},\hat{p}).
\end{equation*}
Similarly, for operators $\underline{\hat{q}}$, $\underline{\hat{p}}$ satisfying the following commutation relation
\begin{equation}
[\underline{\hat{q}}, \underline{\hat{p}}] = -i\hbar
\label{eq:6.5.4}
\end{equation}
the $(\sigma,S)$-ordered operator function takes the form
\begin{equation*}
A_{\sigma,S}(\underline{\hat{q}},\underline{\hat{p}}) := (S^{-1}A)_\sigma(\underline{\hat{q}},\underline{\hat{p}}).
\end{equation*}
The Hermitian conjugation of a $(\sigma,S)$-ordered operator function reads
\begin{equation*}
A^\dagger_{\sigma,S}(\hat{q},\hat{p}) = A^*_{\bar{\sigma},\bar{S}}(\hat{q},\hat{p}).
\end{equation*}
Indeed
\begin{equation*}
A^\dagger_{\sigma,S}(\hat{q},\hat{p}) = (S^{-1}A)^\dagger_\sigma(\hat{q},\hat{p}) = (S^{-1}A)^*_{\bar{\sigma}}(\hat{q},\hat{p})
= (\bar{S}^{-1} A^*)_{\bar{\sigma}}(\hat{q},\hat{p}) = A^*_{\bar{\sigma},\bar{S}}(\hat{q},\hat{p}).
\end{equation*}

\begin{remark}
For a special case of the automorphism $S$ such that
\begin{equation}
S^{-1} = F(-i\hbar \partial_x, i\hbar \partial_p),
\label{eq:6.5.6}
\end{equation}
where $F \colon \mathbb{R}^2 \to \mathbb{C}$ is some general analytic function, the $(\sigma,S)$-ordered operator function reads
\begin{equation*}
A_{\sigma,S}(\hat{q},\hat{p}) = \frac{1}{2\pi\hbar} \iint \mathcal{F}A(\xi,\eta) e^{\frac{i}{\hbar}\xi \hat{q}}
e^{-\frac{i}{\hbar} \eta \hat{p}} e^{-\frac{i}{\hbar} \sigma \xi \eta} F(\xi,\eta) \d{\xi} \d{\eta}.
\end{equation*}
The above formula for the operator function was first considered by Cohen \cite{Cohen:1966}.
\end{remark}

\begin{example}
An interesting special case of two parameter family of automorphisms (\ref{eq:6.5.6}) is received for $F(\xi,\eta) = e^{\frac{1}{2\hbar} \alpha \xi^2
+ \frac{1}{2\hbar} \beta \eta^2}$, where $\alpha,\beta$ are some real numbers. In this case the automorphisms $S$ take the form
\begin{equation*}
S_{\alpha,\beta} = \exp \left( \tfrac{1}{2} \hbar \alpha \partial_x^2 + \tfrac{1}{2} \hbar \beta \partial_p^2 \right).
\end{equation*}
Moreover, the $\star_{\sigma,S}$-product (denoted hereafter by $\star_{\sigma,\alpha,\beta}$) reads
\begin{equation*}
f \star_{\sigma,\alpha,\beta} g = f \exp \left( i\hbar \sigma \overleftarrow{\partial}_x \overrightarrow{\partial}_p
- i\hbar \bar{\sigma} \overleftarrow{\partial}_p \overrightarrow{\partial}_x
+ \hbar \alpha \overleftarrow{\partial}_x \overrightarrow{\partial}_x
+ \hbar \beta \overleftarrow{\partial}_p \overrightarrow{\partial}_p \right) g.
\end{equation*}
The $(\sigma,\alpha,\beta)$-family of quantizations will serve as an illustrative example throughout the paper.
The particular case of the $\star_{\sigma,\alpha,\beta}$-product is widely used in the \emph{holomorphic coordinates}
$a(x,p) = (\omega x + ip)/\sqrt{2\omega}$, $\bar{a}(x,p) = (\omega x - ip)/\sqrt{2\omega}$ ($\omega > 0$). Namely, for $\sigma = \frac{1}{2}$,
$\alpha = \frac{2\lambda - 1}{2\omega}$, $\beta = \omega^2 \alpha$ where $\lambda \in \mathbb{R}$ the $\star_{\sigma,\alpha,\beta}$-product
(denoted hereafter by $\star_\lambda$) written in the holomorphic coordinates takes the form \cite{Hirshfeld:2002a}
\begin{equation*}
f \star_\lambda g = f \exp \left( \hbar \lambda \overleftarrow{\partial}_a \overrightarrow{\partial}_{\bar{a}}
- \hbar \bar{\lambda} \overleftarrow{\partial}_{\bar{a}} \overrightarrow{\partial}_a \right) g.
\end{equation*}
\end{example}

\begin{example}
In general the automorphism $S$ is not of the form (\ref{eq:6.5.6}). As an example the following three parameter family of automorphisms may serve
\begin{equation*}
S = \exp \left( -i\hbar a \partial_x \partial_p + i\hbar b x \partial_p^2 - \hbar^2 c \partial_p^3 \right),
\end{equation*}
where $a,b,c \in \mathbb{R}$. This shows that the family of quantizations considered in the paper is more general than the broad family of quantizations
considered by Cohen and others.
As an illustration lets consider the following observable $A(x,p) = \frac{1}{2}p^2 + \frac{1}{6}p^3 + \frac{1}{2}x^2$. Then, one easily finds that
\begin{equation*}
(S^{-1}A)(x,p) = \tfrac{1}{2}p^2 + \tfrac{1}{6}p^3 + \tfrac{1}{2}x^2 - i\hbar bx(1 + p) + \hbar^2(\tfrac{1}{2}ab + c)
\end{equation*}
and that
\begin{align*}
A_{\sigma,S}(\hat{q},\hat{p}) & = \tfrac{1}{2}\hat{p}^2 + \tfrac{1}{6}\hat{p}^3 + \tfrac{1}{2}\hat{q}^2
- i\hbar b(\hat{q} + \bar{\sigma}\hat{q}\hat{p} + \sigma \hat{p}\hat{q}) + \hbar^2(\tfrac{1}{2}ab + c) \\
& = \tfrac{1}{2}\hat{p}^2 + \tfrac{1}{6}\hat{p}^3 + \tfrac{1}{2}\hat{q}^2 - b\hat{q}\hat{p}\hat{q} + b\hat{p}\hat{q}^2
- (\tfrac{1}{2}ab + b\bar{\sigma} + c)\hat{q}\hat{p}\hat{q}\hat{p} + (\tfrac{1}{2}ab - b\sigma + c)\hat{q}\hat{p}^2\hat{q} \\
& \quad {} + (\tfrac{1}{2}ab + b\bar{\sigma} + c)\hat{p}\hat{q}^2\hat{p} - (\tfrac{1}{2}ab - b\sigma + c)\hat{p}\hat{q}\hat{p}\hat{q}.
\end{align*}
\end{example}

Lets define the following operators
\begin{align*}
\hat{q}_{\sigma,S} & := x \star_{\sigma,S} {}, & \hat{p}_{\sigma,S} & := p \star_{\sigma,S} {}, \\
\hat{q}^*_{\bar{\sigma},S} & := {} \star_{\sigma,S} x, & \hat{p}^*_{\bar{\sigma},S} & := {} \star_{\sigma,S} p.
\end{align*}
There holds
\begin{theorem}
\begin{align*}
\hat{q}_{\sigma,S} & = S \hat{q}_\sigma S^{-1}, & \hat{p}_{\sigma,S} & = S \hat{p}_\sigma S^{-1}, \\
\hat{q}^*_{\bar{\sigma},S} & = S \hat{q}^*_{\bar{\sigma}} S^{-1}, & \hat{p}^*_{\bar{\sigma},S} & = S \hat{p}^*_{\bar{\sigma}} S^{-1},
\end{align*}
where $\hat{q}_\sigma = x \star_\sigma {}$, $\hat{p}_\sigma = p \star_\sigma {}$, $\hat{q}^*_{\bar{\sigma}} = {} \star_\sigma x$,
$\hat{p}^*_{\bar{\sigma}} = {} \star_\sigma p$.
\end{theorem}
\begin{proof}
Let $\Psi \in \mathcal{H}$. Using relations (\ref{eq:6.5.1}) one finds that
\begin{equation*}
\hat{q}_{\sigma,S} \Psi = x \star_{\sigma,S} \Psi = SS^{-1}(x \star_{\sigma,S} \Psi) = S(S^{-1}x \star_\sigma S^{-1}\Psi)
= S(x \star_\sigma S^{-1}\Psi) = S \hat{q}_\sigma S^{-1} \Psi.
\end{equation*}
Similarly one proves other equalities.
\end{proof}

Of course, the operators $\hat{q}_{\sigma,S}$ and $\hat{p}_{\sigma,S}$ satisfy the commutation relation (\ref{eq:6.5.3}), whereas the operators
$\hat{q}^*_{\bar{\sigma},S}$ and $\hat{p}^*_{\bar{\sigma},S}$ satisfy the commutation relation (\ref{eq:6.5.4}).

There holds
\begin{theorem}
\label{thm:6.5.1}
For any function $A \in \mathcal{A}_Q$
\begin{equation*}
A_L \star_{\sigma,S} {} = A_{\sigma,S}(\hat{q}_{\sigma,S},\hat{p}_{\sigma,S}), \qquad
A_R \star_{\sigma,S} {} = A_{\sigma,S}(\hat{q}^*_{\bar{\sigma},S},\hat{p}^*_{\bar{\sigma},S}).
\end{equation*}
\end{theorem}
\begin{proof}
For $\Psi \in \mathcal{H}$ there holds
\begin{align*}
A_L \star_{\sigma,S} \Psi & = SS^{-1}(A \star_{\sigma,S} \Psi) = S(S^{-1}A \star_\sigma S^{-1}\Psi)
= S (S^{-1}A)_\sigma(\hat{q}_\sigma,\hat{p}_\sigma) S^{-1} \Psi \\
& = \frac{1}{2\pi\hbar} \iint \mathcal{F}(S^{-1}A)(\xi,\eta) S e^{\frac{i}{\hbar} \xi \hat{q}_\sigma} S^{-1}S e^{-\frac{i}{\hbar} \eta \hat{p}_\sigma}
S^{-1} e^{-\frac{i}{\hbar} \xi \eta} \d{\xi}\d{\eta} \\
& = \frac{1}{2\pi\hbar} \iint \mathcal{F}(S^{-1}A)(\xi,\eta) e^{\frac{i}{\hbar} \xi S \hat{q}_\sigma S^{-1}} e^{-\frac{i}{\hbar} \eta S \hat{p}_\sigma
S^{-1}} e^{-\frac{i}{\hbar} \xi \eta} \d{\xi}\d{\eta} \\
& = (S^{-1}A)_\sigma(\hat{q}_{\sigma,S},\hat{p}_{\sigma,S}) \Psi
= A_{\sigma,S}(\hat{q}_{\sigma,S},\hat{p}_{\sigma,S}) \Psi.
\end{align*}
Similarly one proves the second equality.
\end{proof}

\subsection{Pure states, mixed states and expectation values of observables in the general quantization scheme}
\label{subsec:6.6}
As was presented earlier all admissible states of the quantum Hamiltonian system are contained in the Hilbert space $\mathcal{H}$. It is necessary to
determine which functions from $\mathcal{H}$ can be considered as \emph{pure states} and \emph{mixed states}. \emph{Pure states} will be defined as
functions $\Psi_{\textrm{pure}} \in \mathcal{H}$ which satisfy the following conditions
\begin{enumerate}
\item $\Psi_{\textrm{pure}} \star_{\sigma,S} {} = (\Psi_{\textrm{pure}} \star_{\sigma,S} {})^\dagger$ (hermiticity),
\item $\Psi_{\textrm{pure}} \star_{\sigma,S} \Psi_{\textrm{pure}} = \dfrac{1}{\sqrt{2\pi\hbar}} \Psi_{\textrm{pure}}$ (idempotence),
\item $\| \Psi_{\textrm{pure}} \|_{\mathcal{H}} = 1$ (normalization).
\end{enumerate}
\emph{Mixed states} $\Psi_{\textrm{mix}} \in \mathcal{H}$ will be defined in a standard way, as linear combinations, possibly infinite, of some families
of pure states $\Psi_{\textrm{pure}}^{(\lambda)}$
\begin{equation*}
\Psi_{\textrm{mix}} := \sum_{\lambda} p_{\lambda} \Psi_{\textrm{pure}}^{(\lambda)},
\end{equation*}
where $0 \le p_{\lambda} \le 1$ and $\sum_{\lambda} p_{\lambda} = 1$. Such definition of mixed states reflects the lack of knowledge about the state of
the system, where $p_{\lambda}$ is the probability of finding the system in a state $\Psi_{\textrm{pure}}^{(\lambda)}$.

For an admissible quantum state $\Psi \in \mathcal{H}$ lets define a \emph{quantum distribution function} $\rho$ on the phase space by the equation
\begin{equation*}
\rho := \frac{1}{\sqrt{2\pi\hbar}} \Psi.
\end{equation*}
Note that from Theorem \ref{thm:9.3} it follows that the function $\rho$ is a quasi-probabilistic distribution function, i.e.
\begin{equation*}
\iint \rho(x,p) \d{x} \d{p} = 1.
\end{equation*}
The quantum distribution functions $\rho$ are the analogue of the classical distribution functions representing states of the classical
Hamiltonian system. The difference between classical and quantum distribution functions is that the latter do not have to be non-negative everywhere.
Thus, $\rho(x,p)$ cannot be interpreted as a probability density of finding a particle in a point $(x,p)$ of the phase space. This is a
reflection of the fact that $x$ and $p$ coordinates do not commute with respect to the $\star_{\sigma,S}$-multiplication, which yield, from
the Heisenberg uncertainty principle, that it is impossible to measure simultaneously the position and momentum of a particle. Hence, the point position
of a particle in the phase space does not make sense anymore. On the other hand, from Theorem \ref{thm:9.4} it follows that \emph{marginal distributions}
\begin{equation*}
P(x) := \int (S^{-1} \rho)(x,p) \d{p}, \qquad
P(p) := \int (S^{-1} \rho)(x,p) \d{x},
\end{equation*}
are probabilistic distribution functions and can be interpreted as probability densities that a particle in the phase space have position $x$ or momentum
$p$. The result is not surprising as each marginal distribution depends on commuting coordinates only. Note however, that only in the case of the
$\star_\sigma$-product the marginal distributions are received by simple integration of a distribution function with respect to $x$ or $p$ variable.
In general the distribution function first has to be transformed with the isomorphism $S$.

The expectation value of an observable $\hat{A} \in \hat{\mathcal{A}}_Q$ in an admissible state $\Psi \in \mathcal{H}$ can be defined like in its
classical analogue (\ref{eq:7.2.16}), i.e. as a mean value with respect to a quantum distribution function $\rho = \frac{1}{\sqrt{2\pi\hbar}} \Psi$
\begin{equation*}
\braket{\hat{A}}_{\Psi} = \iint (\hat{A} \rho)(x,p) \d{x} \d{p} = \iint (A \star_{\sigma,S} \rho)(x,p) \d{x} \d{p}.
\end{equation*}

\subsection{Time evolution of quantum Hamiltonian systems in the general quantization scheme}
\label{subsec:6.7}
In this section the time evolution of a quantum Hamiltonian system will be presented. Analogically as in classical mechanics, the time evolution of a
quantum Hamiltonian system is governed by a Hamiltonian $\hat{H}$. It will be assumed that $\hat{H} \in \hat{\mathcal{O}}_Q$, i.e. $H = H^\dagger$.
The time evolution of a quantum distribution function $\rho$ is defined like in its classical counterpart (\ref{eq:7.3.18})
\begin{gather}
L(H,\rho) := \frac{\partial \rho}{\partial t} - \lshad H,\rho \rshad = 0 \quad
\Leftrightarrow \quad
i\hbar \frac{\partial \rho}{\partial t} - [H,\rho] = 0.
\label{eq:6.7.1}
\end{gather}
States $\Psi \in \mathcal{H}$ which do not change during the time development, i.e. such that $\frac{\partial \Psi}{\partial t} = 0$ are called
\emph{stationary states}. From the time evolution equation (\ref{eq:6.7.1}) it follows that stationary states $\Psi$ satisfy
\begin{equation*}
[H,\Psi] = 0.
\end{equation*}
If a stationary state $\Psi$ is a pure state then, from Theorem \ref{thm:9.6}, it follows that the above equation is equivalent to a pair of
$\star_{\sigma,S}$-genvalue equations
\begin{equation*}
H \star_{\sigma,S} \Psi = E \Psi, \quad \Psi \star_{\sigma,S} H = E \Psi,
\end{equation*}
for some $E \in \mathbb{R}$. Note that $E$ in the above equations is the expectation value of the Hamiltonian $\hat{H}$ in a stationary state $\Psi$,
hence it is an energy of the system in the state $\Psi$.

The formal solution of (\ref{eq:6.7.1}) takes the form
\begin{equation*}
\rho(t) = U(t) \star_{\sigma,S} \rho(0) \star_{\sigma,S} U(-t),
\end{equation*}
where
\begin{equation}
U(t) = e_{\star_{\sigma,S}}^{-\frac{i}{\hbar} t H} := \sum_{k=0}^{\infty} \frac{1}{k!} \left( -\frac{i}{\hbar} t \right)^k
\underbrace{H \star_{\sigma,S} \ldots \star_{\sigma,S} H}_k
\label{eq:6.7.4}
\end{equation}
is an unitary function in $\mathcal{H}$ as $\hat{H}$ is self-adjoint. Hence, the time evolution of states can be alternatively expressed in terms of the
one parameter group of unitary functions $U(t)$.

From (\ref{eq:6.7.1}) it follows that a time dependent expectation value of an observable $\hat{A} \in \hat{\mathcal{A}}_Q$ in a state $\rho(t)$, i.e.
$\braket{\hat{A}}_{\rho(t)}$, fulfills the following equation of motion
\begin{equation}
\braket{\hat{A}}_{L(H,\rho)} = 0 \iff \frac{\d{}}{\d{t}} \braket{\hat{A}}_{\rho(t)} - \braket{ \lshad \hat{A},\hat{H} \rshad }_{\rho(t)} = 0.
\label{eq:6.7.5}
\end{equation}
Indeed
\begin{align*}
& \iint \d{x}\d{p} A \star_{\sigma,S} \frac{\partial \rho}{\partial t}(t) = \frac{\d{}}{\d{t}} \iint \d{x}\d{p} A \star_{\sigma,S}
\rho(t) = \frac{\d{}}{\d{t}} \braket{\hat{A}}_{\rho(t)}, \displaybreak[0] \\
& \iint \d{x}\d{p} A \star_{\sigma,S} \frac{1}{i\hbar} (H \star_{\sigma,S} \rho(t) - \rho(t) \star_{\sigma,S} H)
 = \iint \d{x}\d{p} \frac{1}{i\hbar} (A \star_{\sigma,S} H - H \star_{\sigma,S} A) \star_{\sigma,S} \rho(t)
 = \braket{ \lshad \hat{A},\hat{H} \rshad }_{\rho(t)}.
\end{align*}
Equation (\ref{eq:6.7.5}) is the quantum analogue of the classical equation (\ref{eq:7.3.20}).

Until now the time evolution in the \emph{Schr\"odinger picture} were considered, i.e. only states undergo a time development. It is also possible to
consider a dual approach to the time evolution, namely the \emph{Heisenberg picture}. In this picture states remain still whereas the observables undergo
a time development. The time development of an observable $\hat{A} \in \hat{\mathcal{A}}_Q$ is given by the action of the unitary function $U(t)$ from
(\ref{eq:6.7.4}) on $\hat{A}$
\begin{equation}
\hat{A}(t) = U(-t) \star_{\sigma,S} A(0) \star_{\sigma,S} U(t) \star_{\sigma,S} {} = \hat{U}(-t) \hat{A}(0) \hat{U}(t).
\label{eq:6.7.6}
\end{equation}
Differentiating equation (\ref{eq:6.7.6}) with respect to $t$ results in such evolution equation for $\hat{A}$
\begin{equation}
\frac{\d{\hat{A}}}{\d{t}}(t) - \lshad \hat{A}(t),\hat{H} \rshad = 0.
\label{eq:6.7.7}
\end{equation}
Equation (\ref{eq:6.7.7}) is the quantum analogue of the classical equation (\ref{eq:7.3.22}).

Both presented approaches to the time development yield equal predictions concerning the results of measurements, since
\begin{align*}
\braket{\hat{A}(0)}_{\rho(t)} & = \iint \d{x}\d{p} A(0) \star_{\sigma,S} \rho(t)
= \iint \d{x}\d{p} A(0) \star_{\sigma,S} U(t) \star_{\sigma,S} \rho(0) \star_{\sigma,S} U(-t) \\
& = \iint \d{x}\d{p} (U(-t) \star_{\sigma,S} A(0) \star_{\sigma,S} U(t)) \star_{\sigma,S} \rho(0)
= \iint \d{x}\d{p} A(t) \star_{\sigma,S} \rho(0)
= \braket{\hat{A}(t)}_{\rho(0)}.
\end{align*}

\section{Ordinary description of quantum mechanics}
\label{sec:9}
In this section the equivalence of the phase space formulation of quantum mechanics and the Schr\"odinger, Dirac and Heisenberg formulation of quantum
mechanics will be proved by observing that the Wigner-Moyal transform (see e.g. \cite{Gosson:2005}) have all properties of the tensor product. This
observation allows writing many previous results found in the literature in a lucid and elegant way, from which the equivalence of the two formulations
of quantum mechanics is easier to see. Moreover, this observation will also provide an argument to treat the phase space quantum mechanics as a more
fundamental formalism of quantum mechanics, than the one developed by Schr\"odinger, Dirac and Heisenberg.

First, it will be shown that $\mathcal{H}$ can be considered as a tensor product of Hilbert spaces $L^2(\mathbb{R})$ and a space dual to it
$\left( L^2(\mathbb{R}) \right)^*$. It is well known that the Hilbert space $\left( L^2(\mathbb{R}) \right)^*$ dual to $L^2(\mathbb{R})$ can be
identified with $L^2(\mathbb{R})$ where the anti-linear duality map $* \colon L^2(\mathbb{R}) \to \left( L^2(\mathbb{R}) \right)^*$ is the complex
conjugation of functions. The tensor product of $\left( L^2(\mathbb{R}) \right)^*$ and $L^2(\mathbb{R})$ is defined up to an isomorphism. The most
natural choice for the tensor product of $\left( L^2(\mathbb{R}) \right)^*$ and $L^2(\mathbb{R})$ is the Hilbert space $L^2(\mathbb{R}^2)$ where
the tensor product of $\varphi^*, \psi$ ($\varphi,\psi \in L^2(\mathbb{R})$) is defined as
\begin{equation*}
(\varphi^* \otimes \psi)(x,y) := \varphi^*(x) \psi(y)
\end{equation*}
and the scalar product in $L^2(\mathbb{R}^2)$ satisfies the equation
\begin{equation*}
\braket{\varphi_1^* \otimes \psi_1 | \varphi_2^* \otimes \psi_2}_{L^2} = \braket{\varphi_2 | \varphi_1}_{L^2} \braket{\psi_1 | \psi_2}_{L^2},
\end{equation*}
for $\varphi_1,\varphi_2,\psi_1,\psi_2 \in L^2(\mathbb{R})$.

Now, an isomorphism of $L^2(\mathbb{R}^2)$ onto $\mathcal{H}$ will be defined, which will make from $\mathcal{H}$ a tensor product of
$\left( L^2(\mathbb{R}) \right)^*$ and $L^2(\mathbb{R})$. First, note that the Fourier transform $\mathcal{F}_y$ is an isomorphism of
$L^2(\mathbb{R}^2)$. For $\Psi(x,y) \in L^2(\mathbb{R}^2)$, the function
\begin{equation*}
\Psi(x,p) = \mathcal{F}_y(\Psi(x,y)) = \frac{1}{\sqrt{2\pi \hbar}} \int \d{y} e^{-\frac{i}{\hbar} py} \Psi(x,y)
\end{equation*}
will be called an $(x,p)$-representation of $\Psi(x,y)$ and it will be considered as a function on the phase space $M = \mathbb{R}^2$ in the canonical
coordinates of position $x$ and momentum $p$. Lets introduce another isomorphism of $L^2(\mathbb{R}^2)$ by the equation
\begin{equation*}
T_\sigma \Psi(x,y) := \Psi(x - \bar{\sigma} y, x + \sigma y), \quad \Psi \in L^2(\mathbb{R}^2).
\end{equation*}
As the searched isomorphism of $L^2(\mathbb{R}^2)$ onto $\mathcal{H}$ the map $S \mathcal{F}_y T_\sigma$ will be taken. A tensor product of
$\left( L^2(\mathbb{R}) \right)^*$ and $L^2(\mathbb{R})$ induced by this isomorphism will be denoted by
$\left( L^2(\mathbb{R}) \right)^* \otimes_{\sigma,S} L^2(\mathbb{R})$ and called a \emph{$(\sigma,S)$-twisted tensor product} of
$\left( L^2(\mathbb{R}) \right)^*$ and $L^2(\mathbb{R})$. Hence
\begin{equation}
\mathcal{H} = \left( L^2(\mathbb{R}) \right)^* \otimes_{\sigma,S} L^2(\mathbb{R})
= S \mathcal{F}_y T_\sigma \left( \left( L^2(\mathbb{R}) \right)^* \otimes L^2(\mathbb{R}) \right)
\label{eq:9.1}
\end{equation}
and the scalar product in $\mathcal{H}$ satisfies
\begin{equation*}
\braket{\varphi_1^* \otimes_{\sigma,S} \psi_1 | \varphi_2^* \otimes_{\sigma,S} \psi_2}_{\mathcal{H}}
= \braket{\varphi_2 | \varphi_1}_{L^2} \braket{\psi_1 | \psi_2}_{L^2},
\end{equation*}
for $\varphi_1,\varphi_2,\psi_1,\psi_2 \in L^2(\mathbb{R})$.
The relevance of the representation (\ref{eq:9.1}) will be revealed in the key theorem \ref{thm:9.5}.

The generators of $\mathcal{H}$ are of the form
\begin{align}
\Psi^{\sigma,S}(x,p) & = (\varphi^* \otimes_{\sigma,S} \psi)(x,p)
= \frac{1}{\sqrt{2\pi \hbar}} S \int \d{y} e^{-\frac{i}{\hbar} py} \varphi^*(x - \bar{\sigma} y) \psi(x + \sigma y) \nonumber \displaybreak[0] \\
& = \frac{1}{\sqrt{2\pi \hbar}} \iiint \d{x'} \d{p'} \d{y} \varphi^*(x' - \bar{\sigma} y) \psi(x' + \sigma y) S(x,p,x',p') e^{-\frac{i}{\hbar} p'y},
\label{eq:9.2}
\end{align}
where $\varphi, \psi \in L^2(\mathbb{R})$ and $S(x,p,x',p')$ is an integral kernel of the isomorphism $S$.

\begin{example}
For a special case of the $\star_{\sigma,\alpha,\beta}$-product ($\alpha > 0$, $\beta > 0$) the integral kernel of the isomorphism $S = S_{\alpha,\beta}$
takes the form
\begin{equation*}
S_{\alpha,\beta}(x,p,x',p') = \frac{1}{2\pi\hbar\sqrt{\alpha \beta}} e^{-\frac{1}{2\hbar\alpha}(x - x')^2} e^{-\frac{1}{2\hbar\beta}(p - p')^2}.
\end{equation*}
Hence, the generators $\Psi^{\sigma,\alpha,\beta}$ take the form
\begin{equation}
\Psi^{\sigma,\alpha,\beta}(x,p) = \frac{1}{(2\pi \hbar)^{3/2} \sqrt{\alpha \beta}} \iiint \d{x'} \d{p'} \d{y} \varphi^*(x' - \bar{\sigma} y)
\psi(x' + \sigma y) e^{-\frac{1}{2\hbar\alpha}(x - x')^2} e^{-\frac{1}{2\hbar\beta}(p - p')^2} e^{-\frac{i}{\hbar} p'y}.
\label{eq:9.4}
\end{equation}
In a special case of Weyl ordering $\sigma = \frac{1}{2}$, $\alpha = \beta = 0$ generators $\Psi^{\sigma,\alpha,\beta}$ are the well known Wigner
functions related to the Moyal $\star$-product. Another special case are generators $\Psi^{\sigma,\alpha,\beta}$ related to the $\star_\lambda$-product,
i.e. the case when $\sigma = \frac{1}{2}$, $\alpha = \frac{2\lambda - 1}{2\omega}$, $\beta = \omega^2 \alpha$ ($\omega > 0$) and when generators
$\Psi^{\sigma,\alpha,\beta} = \Psi^\lambda$ are written in the holomorphic coordinates. The generators $\Psi^\lambda$ in the case $\lambda = 0$ are the
well known Glauber-Sudarshan distribution functions, and in the case $\lambda = 1$ are the well known Husimi distribution functions. Many other particular
examples of the quantum phase-space distribution functions (\ref{eq:9.4}) considered in the past are listed and described in the review paper
\cite{Lee:1995}.
\end{example}

Observe, that if $\{ \varphi_i \}$ is an orthonormal basis in $L^2(\mathbb{R})$, then $\{ \Psi_{ij} \} = \{ \varphi_i^* \otimes_{\sigma,S} \varphi_j \}$
is an orthonormal basis in $\mathcal{H}$ and for any $\Psi = \varphi^* \otimes_{\sigma,S} \psi \in \mathcal{H}$ where $\varphi,\psi \in L^2(\mathbb{R})$,
one have
\begin{gather*}
\varphi = \sum_i b_i \varphi_i, \qquad \psi = \sum_j c_j \varphi_j, \qquad \textrm{for some $b_i, c_i \in \mathbb{C}$}, \\
\Psi = \sum_{i,j} a_{ij} \Psi_{ij}, \qquad a_{ij} = b_i^* c_j.
\end{gather*}
The interesting property of the basis functions $\Psi_{ij}$ is they idempotence. Namely, there holds
\begin{theorem}
\begin{equation}
\Psi_{ij} \star_{\sigma,S} \Psi_{kl} = \frac{1}{\sqrt{2\pi\hbar}} \delta_{il} \Psi_{kj}
\label{eq:9.3}
\end{equation}
\end{theorem}
\begin{proof}
For the case of the Moyal product the proof can be found in \cite{Curtright:1998}. The general case follows by applying the isomorphism
$S S_{\sigma - \frac{1}{2}}$ to the both sides of the equation.
\end{proof}

Hence, if $\Psi_1 = \varphi_1^* \otimes_{\sigma,S} \psi_1$ and $\Psi_2 = \varphi_2^* \otimes_{\sigma,S} \psi_2$ where
$\varphi_1, \psi_1, \varphi_2, \psi_2 \in L^2(\mathbb{R})$, then
\begin{subequations}
\label{eq:9.56}
\begin{align}
\Psi_1 \star_{\sigma,S} \Psi_2 & = \frac{1}{\sqrt{2\pi\hbar}} \braket{\varphi_1 | \psi_2}_{L^2}
(\varphi_2^* \otimes_{\sigma,S} \psi_1),
\label{eq:9.56a} \\
\Psi_2 \star_{\sigma,S} \Psi_1 & = \frac{1}{\sqrt{2\pi\hbar}} \braket{\varphi_2 | \psi_1}_{L^2}
(\varphi_1^* \otimes_{\sigma,S} \psi_2).
\label{eq:9.56b}
\end{align}
\end{subequations}

Using the basis $\{ \Psi_{ij} \} = \{ \varphi_i^* \otimes_{\sigma,S} \varphi_j \}$ some interesting properties of the admissible states can be proved.
Namely, there holds
\begin{theorem}
\label{thm:9.1}
Every pure state $\Psi_{\mathrm{pure}} \in \mathcal{H}$ is of the form
\begin{equation}
\Psi_{\mathrm{pure}} = \varphi^* \otimes_{\sigma,S} \varphi,
\label{eq:9.5}
\end{equation}
for some normalized function $\varphi \in L^2(\mathbb{R})$. Conversely, every function $\Psi \in \mathcal{H}$ of the form (\ref{eq:9.5}) is a pure
state.
\end{theorem}
\begin{proof}
From formula (\ref{eq:9.3}) it follows that every function $\Psi \in \mathcal{H}$ of the form (\ref{eq:9.5}) is a pure state. If now one assumes that
$\Psi_{\textrm{pure}} \in \mathcal{H}$ is a pure state then $\Psi_{\textrm{pure}}$ can be written in a form
\begin{equation*}
\Psi_{\textrm{pure}} = \sum_{i,j} c_{ij} \Psi_{ij},
\end{equation*}
where $\{ \Psi_{ij} \} = \{ \varphi_i^* \otimes_{\sigma,S} \varphi_j \}$ is an induced basis in $\mathcal{H}$ by the basis $\{ \varphi_i \}$
in $L^2(\mathbb{R})$. The assumptions that $\Psi_{\textrm{pure}}$ is Hermitian, idempotent and normalized can be restated saying that the matrix
$\check{c}$ of the coefficients $c_{ij}$ is Hermitian ($\check{c} = \check{c}^\dagger$), idempotent ($\check{c}^2 = \check{c}$) and normalized
($\tr \check{c} = 1$). Since the matrix $\check{c}$ is Hermitian it can be diagonalized, i.e. there exist an unitary matrix $\check{T}$
such that $c_{ij} = \sum_{k,l} T^\dagger_{ik} (a_k \delta_{kl}) T_{lj} = \sum_k T^*_{ki} a_k T_{kj}$ for some $a_k \in \mathbb{R}$. Hence,
$\Psi_{\textrm{pure}}$ takes the form
\begin{equation*}
\Psi_{\textrm{pure}} = \sum_{i,j,k} T^*_{ki} a_k T_{kj} (\varphi_i^* \otimes_{\sigma,S} \varphi_j)
= \sum_k a_k \left( \Big(\textstyle{\sum_i} T_{ki} \varphi_i \Big)^* \otimes_{\sigma,S} \Big(\textstyle{\sum_j} T_{kj} \varphi_j \Big) \right)
= \sum_k a_k (\psi_k^* \otimes_{\sigma,S} \psi_k),
\end{equation*}
where $\psi_k = \sum_i T_{ki} \varphi_i$. The conditions that $\check{c}^2 = \check{c}$ and $\tr \check{c} = 1$ give that $a_k^2 = a_k$ and
$\sum_k a_k = 1$. Hence $a_k = \delta_{k_0 k}$ for some $k_0$, from which follows that $\Psi_{\textrm{pure}} = \psi_{k_0}^* \otimes_{\sigma,S} \psi_{k_0}$.
\end{proof}

\begin{theorem}
\label{thm:9.3}
Every admissible (pure or mixed) state $\Psi \in \mathcal{H}$ satisfies
\begin{equation*}
\frac{1}{\sqrt{2\pi\hbar}} \iint \Psi(x,p) \d{x} \d{p} = 1.
\end{equation*}
\end{theorem}
\begin{proof}
It is enough to prove the theorem for the case when $\Psi$ is a pure state. In that case $\Psi$ can be written in a form $\Psi = \varphi^*
\otimes_{\sigma,S} \varphi$ for some normalized function $\varphi \in L^2(\mathbb{R})$. Hence, one have that
\begin{align*}
\frac{1}{\sqrt{2\pi\hbar}} \iint \Psi(x,p) \d{x} \d{p} & = \frac{1}{\sqrt{2\pi\hbar}} \iint (\varphi^* \otimes_{\sigma,S} \varphi)(x,p) \d{x} \d{p}
\nonumber \displaybreak[0] \\
& = \frac{1}{2\pi\hbar} \iiint \d{x}\d{p}\d{y} e^{-\frac{i}{\hbar} py} \varphi^*(x - \bar{\sigma} y) \varphi(x + \sigma y) \nonumber \displaybreak[0] \\
& = \iint \d{x}\d{y} \delta(y) \varphi^*(x - \bar{\sigma} y) \varphi(x + \sigma y)
= \int \d{x} \varphi^*(x) \varphi(x) = 1.
\end{align*}
\end{proof}

\begin{theorem}
\label{thm:9.4}
Every admissible (pure or mixed) state $\Psi = \sum_\lambda p_\lambda \left( \big( \varphi^{(\lambda)} \big)^* \otimes_{\sigma,S} \varphi^{(\lambda)} \right)$
for some normalized $\varphi^{(\lambda)} \in L^2(\mathbb{R})$ satisfies
\begin{subequations}
\label{eq:9.6}
\begin{align}
\frac{1}{\sqrt{2\pi\hbar}} \int (S^{-1} \Psi)(x,p) \d{p} & = \sum_\lambda p_\lambda |\varphi^{(\lambda)}(x)|^2,
\label{eq:9.6a} \\
\frac{1}{\sqrt{2\pi\hbar}} \int (S^{-1} \Psi)(x,p) \d{x} & = \sum_\lambda p_\lambda |\tilde{\varphi}^{(\lambda)}(p)|^2,
\label{eq:9.6b}
\end{align}
\end{subequations}
where $\tilde{\varphi}^{(\lambda)}$ denotes the Fourier transform of $\varphi^{(\lambda)}$.
\end{theorem}
\begin{proof}
It is enough to prove the theorem for a pure state $\Psi = \varphi^* \otimes_{\sigma,S} \varphi$. Equation (\ref{eq:9.6a}) follows from
\begin{equation*}
\frac{1}{\sqrt{2\pi\hbar}} \int (S^{-1} \Psi)(x,p) \d{p} = \frac{1}{2\pi\hbar} \int \d{p} \int \d{y} e^{-\frac{i}{\hbar} py}
\varphi^*(x - \bar{\sigma} y) \varphi(x + \sigma y)
= \int \d{y} \delta(y) \varphi^*(x - \bar{\sigma} y) \varphi(x + \sigma y) = |\varphi(x)|^2.
\end{equation*}
Equation (\ref{eq:9.6b}) follows from
\begin{equation*}
\frac{1}{\sqrt{2\pi\hbar}} \int (S^{-1} \Psi)(x,p) \d{x} = \frac{1}{2\pi\hbar} \iint \d{x} \d{p} e^{-\frac{i}{\hbar} py}
\varphi^*(x - \bar{\sigma} y) \varphi(x + \sigma y).
\end{equation*}
Introducing new coordinates $x_1 = x - \bar{\sigma} y, x_2 = x + \sigma y$ gives
\begin{equation*}
\frac{1}{\sqrt{2\pi\hbar}} \int (S^{-1} \Psi)(x,p) \d{x} = \frac{1}{2\pi\hbar} \iint \d{x_1} \d{x_2} e^{-\frac{i}{\hbar} p x_2} e^{\frac{i}{\hbar} p x_1}
\varphi^*(x_1) \varphi(x_2) = |\tilde{\varphi}(p)|^2.
\end{equation*}
\end{proof}

From Theorem \ref{thm:9.1} follows that there is a one to one correspondence between pure states of the phase space quantum mechanics and the normalized
functions from the Hilbert space $L^2(\mathbb{R})$.

Elements of the algebra $\hat{\mathcal{A}}_Q$, hence in particular observables, are operators on $\mathcal{H}$ of the form $A \star_{\sigma,S} {}$.
In the canonical case discussed so far, from Theorem \ref{thm:6.5.1}, these operators are equal to operator functions
$A_{\sigma,S}(\hat{q}_{\sigma,S},\hat{p}_{\sigma,S})$. Moreover, states $\Psi \in \mathcal{H}$ can also be considered as operators on $\mathcal{H}$
given by the formula
\begin{equation}
\hat{\Psi} = \sqrt{2\pi\hbar} \Psi \star_{\sigma,S} {}.
\label{eq:9.8}
\end{equation}
The space of all operators $\hat{\Psi}$ given by (\ref{eq:9.8}) will be denoted by $\hat{\mathcal{H}}$. Note that $\hat{\mathcal{H}}$ inherits from
$\mathcal{H}$ a structure of a Hilbert algebra with the scalar product of $\hat{\Psi}_1 = \sqrt{2\pi\hbar} \Psi_1 \star_{\sigma,S} {}$ and
$\hat{\Psi}_2 = \sqrt{2\pi\hbar} \Psi_2 \star_{\sigma,S} {}$ defined by
\begin{equation*}
\braket{\hat{\Psi}_1 | \hat{\Psi}_2}_{\hat{\mathcal{H}}} := \braket{\Psi_1 | \Psi_2}_{\mathcal{H}}.
\end{equation*}
Note also, that from (\ref{eq:6.5.7}) $\| \,\cdot\, \|_{\hat{\mathcal{H}}}$ satisfies the following relation
\begin{equation*}
\| \hat{\Psi}_1 \hat{\Psi}_2 \|_{\hat{\mathcal{H}}} \le \| \hat{\Psi}_1 \|_{\hat{\mathcal{H}}} \| \hat{\Psi}_2 \|_{\hat{\mathcal{H}}}.
\end{equation*}

Now, it will be proved that operators from $\hat{\mathcal{H}}$ can be naturally identified with Hilbert-Schmidt operators defined on the Hilbert space
$L^2(\mathbb{R})$. The space of Hilbert-Schmidt operators $\mathcal{S}^2(L^2(\mathbb{R}))$ is a space of all bounded operators
$\hat{A} \in \mathcal{B}(L^2(\mathbb{R}))$ for which $\| \hat{A} \|_{\mathcal{S}^2} < \infty$, where $\| \,\cdot\, \|_{\mathcal{S}^2}$ is a norm
induced by a scalar product
\begin{equation}
\braket{\hat{A}|\hat{B}}_{\mathcal{S}^2} := \tr(\hat{A}^\dagger \hat{B}), \quad \hat{A},\hat{B} \in \mathcal{S}^2(L^2(\mathbb{R})).
\label{eq:9.11}
\end{equation}
The space of Hilbert-Schmidt operators $\mathcal{S}^2(L^2(\mathbb{R}))$ with the scalar product (\ref{eq:9.11}) is a Hilbert algebra. From the well
known relation between the $\mathcal{S}^2$-norm and the usual operator norm
\begin{equation*}
\| \hat{A} \| \le \| \hat{A} \|_{\mathcal{S}^2}, \quad \hat{A} \in \mathcal{S}^2(L^2(\mathbb{R}))
\end{equation*}
it follows that the inclusion $\mathcal{S}^2(L^2(\mathbb{R})) \subset \mathcal{B}(L^2(\mathbb{R}))$ is continuous and hence, every sequence of
Hilbert-Schmidt operators convergent in $\mathcal{S}^2(L^2(\mathbb{R}))$ is also convergent in $\mathcal{B}(L^2(\mathbb{R}))$.

There holds
\begin{theorem}
\label{thm:9.8}
For every $\hat{\Psi} \in \hat{\mathcal{H}}$
\begin{equation}
\hat{\Psi} = \hat{1} \otimes_{\sigma,S} \hat{\rho},
\label{eq:9.13}
\end{equation}
where $\hat{\rho} \in \mathcal{S}^2(L^2(\mathbb{R}))$ is some Hilbert-Schmidt operator defined on the Hilbert space $L^2(\mathbb{R})$. Conversely,
for every $\hat{\rho} \in \mathcal{S}^2(L^2(\mathbb{R}))$ the operator $\hat{1} \otimes_{\sigma,S} \hat{\rho}$ is an element of
$\hat{\mathcal{H}}$.

In particular, for $\Psi = \varphi^* \otimes_{\sigma,S} \psi$ the corresponding operator $\hat{\Psi}$ takes the form
\begin{equation}
\hat{\Psi} = \hat{1} \otimes_{\sigma,S} \hat{\rho},
\label{eq:9.14}
\end{equation}
where $\hat{\rho} = \braket{\varphi | \,\cdot\,}_{L^2} \psi$.

Moreover, for $\hat{\Psi}_1 = \hat{1} \otimes_{\sigma,S} \hat{\rho}_1$ and $\hat{\Psi}_2 = \hat{1} \otimes_{\sigma,S} \hat{\rho}_2$
\begin{equation}
\braket{\hat{\Psi}_1 | \hat{\Psi}_2}_{\hat{\mathcal{H}}} = \braket{\hat{\rho}_1 | \hat{\rho}_2}_{\mathcal{S}^2}.
\label{eq:9.15}
\end{equation}
\end{theorem}
\begin{proof}
From equation (\ref{eq:9.56a}) for $\Psi = \varphi^* \otimes_{\sigma,S} \psi$ and the basis functions $\Psi_{ij} = \varphi_i^* \otimes_{\sigma,S} \varphi_j$
it follows that
\begin{equation*}
\hat{\Psi} \Psi_{ij} = \sqrt{2\pi\hbar} (\varphi^* \otimes_{\sigma,S} \psi) \star_{\sigma,S} (\varphi_i^* \otimes_{\sigma,S} \varphi_j)
= \braket{\varphi | \varphi_j}_{L^2} (\varphi_i^* \otimes_{\sigma,S} \psi)
= \varphi_i^* \otimes_{\sigma,S} (\hat{\rho} \varphi_j)
= (\hat{1} \otimes_{\sigma,S} \hat{\rho}) \Psi_{ij},
\end{equation*}
where $\hat{\rho} = \braket{\varphi | \,\cdot\,}_{L^2} \psi$, which proves formula (\ref{eq:9.14}).

It is sufficient to prove formula (\ref{eq:9.15}) for basis functions $\Psi_{ij}$. From (\ref{eq:9.14}) it follows that operators corresponding to the
basis functions $\Psi_{ij}$ can be written in a form
\begin{equation*}
\hat{\Psi}_{ij} = \hat{1} \otimes_{\sigma,S} \hat{\rho}_{ij},
\end{equation*}
where $\hat{\rho}_{ij} = \braket{\varphi_i | \,\cdot\,}_{L^2} \varphi_j$. This implies that
\begin{equation*}
\braket{\hat{\Psi}_{ij} | \hat{\Psi}_{kl}}_{\hat{\mathcal{H}}} = \delta_{ik} \delta_{jl} = \braket{\hat{\rho}_{ij} | \hat{\rho}_{kl}}_{\mathcal{S}^2},
\end{equation*}
which proves formula (\ref{eq:9.15}). Formula (\ref{eq:9.13}) is an immediate consequence of formulae (\ref{eq:9.14}) and (\ref{eq:9.15}).
\end{proof}

From the above theorem follows that states can be naturally identified with appropriate operators on the Hilbert space $L^2(\mathbb{R})$. For instance,
if $\Psi_{\textrm{pure}} = \varphi^* \otimes_{\sigma,S} \varphi$ is a pure state then $\hat{\Psi}_{\textrm{pure}}
= \hat{1} \otimes_{\sigma,S} \hat{\rho}_{\textrm{pure}}$ where $\hat{\rho}_{\textrm{pure}} = \braket{\varphi | \,\cdot\,}_{L^2} \varphi$.
If $\Psi_{\textrm{mix}} = \sum_\lambda p_\lambda \Psi_{\textrm{pure}}^{(\lambda)} = \sum_\lambda p_\lambda \left( \varphi^{(\lambda)} \right)^*
\otimes_{\sigma,S} \varphi^{(\lambda)}$ is a mixed state then $\hat{\Psi}_{\textrm{mix}} = \hat{1} \otimes_{\sigma,S} \hat{\rho}_{\textrm{mix}}$ where
\begin{equation*}
\hat{\rho}_{\textrm{mix}} = \sum_\lambda p_\lambda \hat{\rho}_{\textrm{pure}}^{(\lambda)}
= \sum_\lambda p_\lambda \braket{\varphi^{(\lambda)} | \,\cdot\,}_{L^2} \varphi^{(\lambda)}.
\end{equation*}
Pure and mixed state operators $\hat{\rho} \in \mathcal{S}^2(L^2(\mathbb{R}))$ are called \emph{density operators}.

From the below theorem follows that observables can be naturally identified with operators defined on the Hilbert space $L^2(\mathbb{R})$. This theorem
is also the key theorem from which formulae for the expectation values of observables and the time evolution of the observables and states, represented
as operators in $L^2(\mathbb{R})$, follows.
\begin{theorem}
\label{thm:9.5}
Let $A \in \mathcal{A}_Q$ and $\Psi \in \mathcal{H}$ be such that $\Psi = \varphi^* \otimes_{\sigma,S} \psi$ for $\varphi,\psi \in L^2(\mathbb{R})$, then
\begin{equation*}
A_L \star_{\sigma,S} \Psi = A_{\sigma,S}(\hat{q}_{\sigma,S},\hat{p}_{\sigma,S}) \Psi
= \varphi^* \otimes_{\sigma,S} A_{\sigma,S}(\hat{q},\hat{p}) \psi,
\end{equation*}
if $\psi \in D(A_{\sigma,S}(\hat{q},\hat{p}))$ and
\begin{equation*}
A_R \star_{\sigma,S} \Psi = A_{\sigma,S}(\hat{q}^*_{\bar{\sigma},S},\hat{p}^*_{\bar{\sigma},S}) \Psi
= \left( A^\dagger_{\sigma,S}(\hat{q},\hat{p}) \varphi \right)^* \otimes_{\sigma,S} \psi,
\end{equation*}
if $\varphi \in D(A^\dagger_{\sigma,S}(\hat{q},\hat{p}))$,
where $A_{\sigma,S}(\hat{q},\hat{p})$ is a $(\sigma,S)$-ordered operator function of canonical operators of position $\hat{q} = x$ and momentum
$\hat{p} = -i\hbar \partial_x$, acting in the Hilbert space $L^2(\mathbb{R})$, and $D(\hat{A})$ denotes a domain of an operator $\hat{A}$.
\end{theorem}
Since the proof of the above theorem is quite long and tedious it was moved to \ref{sec:2}.

From the proof of Theorem \ref{thm:9.5} a connection between eigenvalues and eigenfunctions of the star-genvalue equation and the corresponding
eigenvalue equation can be derived. For the Moyal star-product these spectral results were derived in \cite{Gosson:2008}, see also \cite{Gosson:2005}.
\begin{corollary}
\label{col:9.1}
Every solution of the $\star_{\sigma,S}$-genvalue equation
\begin{equation}
A \star_{\sigma,S} \Psi = a \Psi
\label{eq:9.46}
\end{equation}
for $A \in \mathcal{A}_Q$ and $a \in \mathbb{C}$ is of the form
\begin{equation}
\Psi = \sum_i \varphi_i^* \otimes_{\sigma,S} \psi_i,
\label{eq:9.64}
\end{equation}
where $\varphi_i \in L^2(\mathbb{R})$ are arbitrary and $\psi_i \in L^2(\mathbb{R})$ are the eigenvectors of the operator
$A_{\sigma,S}(\hat{q},\hat{p})$ corresponding to the eigenvalue $a$ spanning the subspace of all eigenvectors of
$A_{\sigma,S}(\hat{q},\hat{p})$ corresponding to the eigenvalue $a$, i.e. $\psi_i$ satisfy the eigenvalue equation
\begin{equation*}
A_{\sigma,S}(\hat{q},\hat{p}) \psi_i = a \psi_i.
\end{equation*}
In particular, when $a$ is nondegenerate, every solution of (\ref{eq:9.46}) is of the form
\begin{equation*}
\Psi = \varphi^* \otimes_{\sigma,S} \psi,
\end{equation*}
where $\varphi \in L^2(\mathbb{R})$ is arbitrary and $\psi \in L^2(\mathbb{R})$ satisfies the eigenvalue equation
\begin{equation*}
A_{\sigma,S}(\hat{q},\hat{p}) \psi = a \psi.
\end{equation*}

Similarly, every solution of the $\star_{\sigma,S}$-genvalue equation
\begin{equation}
\Psi \star_{\sigma,S} B = b \Psi
\label{eq:9.48}
\end{equation}
for $B \in \mathcal{A}_Q$ and $b \in \mathbb{C}$ is of the form
\begin{equation*}
\Psi = \sum_i \psi_i^* \otimes_{\sigma,S} \varphi_i,
\end{equation*}
where $\varphi_i \in L^2(\mathbb{R})$ are arbitrary and $\psi_i \in L^2(\mathbb{R})$ are the eigenvectors of the operator
$B_{\sigma,S}^\dagger(\hat{q},\hat{p})$ corresponding to the eigenvalue $b^*$ spanning the subspace of all eigenvectors of
$B_{\sigma,S}^\dagger(\hat{q},\hat{p})$ corresponding to the eigenvalue $b^*$, i.e. $\psi_i$ satisfy the eigenvalue equation
\begin{equation*}
B_{\sigma,S}^\dagger(\hat{q},\hat{p}) \psi_i = b^* \psi_i.
\end{equation*}
In particular, when $b^*$ is nondegenerate, every solution of (\ref{eq:9.48}) is of the form
\begin{equation*}
\Psi = \psi^* \otimes_{\sigma,S} \varphi,
\end{equation*}
where $\varphi \in L^2(\mathbb{R})$ is arbitrary and $\psi \in L^2(\mathbb{R})$ satisfies the eigenvalue equation
\begin{equation*}
B_{\sigma,S}^\dagger(\hat{q},\hat{p}) \psi = b^* \psi.
\end{equation*}
\end{corollary}
\begin{proof}
Replacing $\Phi_L$, in the proof of Theorem \ref{thm:9.5}, by $a \Psi$ one gets from equation (\ref{eq:2.36a}) that the $\star_{\sigma,S}$-genvalue
equation (\ref{eq:9.46}) is equivalent to the following equation
\begin{equation*}
A_{\sigma,S}(\xi, -i\hbar \partial_\xi) \tilde{\Psi}_1(\xi,z) = a \tilde{\Psi}_1(\xi,z),
\end{equation*}
i.e. $\tilde{\Psi}_1(\xi,z)$ is an eigenvector of the operator $A_{\sigma,S}(\hat{q},\hat{p})$ for every $z$. If $\{ \psi_i \in L^2(\mathbb{R}) \}$
is the basis in the subspace of all eigenvectors of the operator $A_{\sigma,S}(\hat{q},\hat{p})$ corresponding to the eigenvalue $a$ then
$\tilde{\Psi}_1(\xi,z)$, for every $z$, can be written as a linear combination of the basis vectors $\psi_i$
\begin{equation}
\tilde{\Psi}_1(\xi,z) = \sum_i \kappa_i(z) \psi_i(\xi),
\label{eq:9.72}
\end{equation}
where the coefficients $\kappa_i(z) \in \mathbb{C}$ depend on $z$. Since $\Psi \in \mathcal{H}$ the functions $\kappa_i \in L^2(\mathbb{R})$. Now,
from equation (\ref{eq:9.72}), using the analogous arguments as in the proof of Theorem \ref{thm:9.5} it can be proved that $\Psi$ is of the form
(\ref{eq:9.64}) where $\varphi^*_i(x - \bar{\sigma} y) = \kappa_i(y - \bar{\sigma}^{-1} x)$. The second part of the corollary can be proved analogically.
\end{proof}

From Theorem \ref{thm:9.5} it follows that for $\Psi_1 = \varphi_1^* \otimes_{\sigma,S} \psi_1$ and
$\Psi_2 = \varphi_2^* \otimes_{\sigma,S} \psi_2$ where $\varphi_1, \psi_1, \varphi_2, \psi_2 \in L^2(\mathbb{R})$
\begin{subequations}
\label{eq:9.40}
\begin{align}
\braket{\Psi_1 | A_L \star_{\sigma,S} \Psi_2}_{\mathcal{H}} & = \braket{\varphi_2 | \varphi_1}_{L^2}
\braket{\psi_1 | A_{\sigma,S}(\hat{q},\hat{p}) \psi_2}_{L^2},
\label{eq:9.40a} \\
\braket{\Psi_1 | A_R \star_{\sigma,S} \Psi_2}_{\mathcal{H}} & = \braket{A^\dagger_{\sigma,S}(\hat{q},\hat{p}) \varphi_2 | \varphi_1}_{L^2}
\braket{\psi_1 | \psi_2}_{L^2}
= \braket{\varphi_2 | A_{\sigma,S}(\hat{q},\hat{p}) \varphi_1}_{L^2} \braket{\psi_1 | \psi_2}_{L^2}.
\label{eq:9.40b}
\end{align}
\end{subequations}
From equations (\ref{eq:6.4.17}) it follows that
\begin{subequations}
\label{eq:9.44}
\begin{align}
(A_L \star_{\sigma,S} {})^\dagger \Psi & = A^*_{\bar{\sigma},\bar{S}}(\hat{q}_{\sigma,S},\hat{p}_{\sigma,S}) \Psi
= \varphi^* \otimes_{\sigma,S} A^\dagger_{\sigma,S}(\hat{q},\hat{p}) \psi,
\label{eq:9.44a} \\
(A_R \star_{\sigma,S} {})^\dagger \Psi & = A^*_{\bar{\sigma},\bar{S}}(\hat{q}^*_{\bar{\sigma},S},\hat{p}^*_{\bar{\sigma},S}) \Psi
= \left( A_{\sigma,S}(\hat{q},\hat{p}) \varphi \right)^* \otimes_{\sigma,S} \psi,
\label{eq:9.44b}
\end{align}
\end{subequations}
for $\Psi = \varphi^* \otimes_{\sigma,S} \psi$ where $\varphi, \psi \in L^2(\mathbb{R})$. Note, that Corollary \ref{col:9.1} implies that in
the nondegenerate case the solution $\Psi$ to the following pair of $\star_{\sigma,S}$-genvalue equations
\begin{equation}
A \star_{\sigma,S} \Psi = a \Psi, \qquad \Psi \star_{\sigma,S} B = b \Psi,
\label{eq:9.70}
\end{equation}
is unique up to a multiplication constant and is of the form $\Psi = \varphi^* \otimes_{\sigma,S} \psi$, where $\varphi, \psi \in L^2(\mathbb{R})$
satisfy the following eigenvalue equations
\begin{align}
A_{\sigma,S}(\hat{q},\hat{p}) \psi = a \psi, \qquad B_{\sigma,S}^\dagger(\hat{q},\hat{p}) \varphi = b^* \varphi.
\label{eq:9.71}
\end{align}
Hence, the pair of $\star_{\sigma,S}$-genvalue equations (\ref{eq:9.70}) is equivalent to the pair of eigenvalue equations (\ref{eq:9.71}).
In particular, from formula (\ref{eq:9.44b}) it follows that a pair of $\star_{\sigma,S}$-genvalue equations
\begin{align*}
A_L \star_{\sigma,S} \Psi = a \Psi, \qquad (A_R \star_{\sigma,S} {})^\dagger \Psi = a^* \Psi
\end{align*}
have a solution $\Psi$ in the form of a pure state $\Psi = \varphi^* \otimes_{\sigma,S} \varphi$, where $\varphi \in L^2(\mathbb{R})$ is a
solution to the eigenvalue equation
\begin{equation*}
A_{\sigma,S}(\hat{q},\hat{p}) \varphi = a \varphi.
\end{equation*}

From Theorem \ref{thm:9.5} follows also that operators $\hat{A} = A \star_{\sigma,S} {}$, hence in particular observables, can be written as
\begin{equation*}
A \star_{\sigma,S} {} = A_{\sigma,S}(\hat{q}_{\sigma,S},\hat{p}_{\sigma,S})
= \hat{1} \otimes_{\sigma,S} A_{\sigma,S}(\hat{q},\hat{p}).
\end{equation*}
Hence, operators from $\hat{\mathcal{A}}_Q$ can be naturally identified with operator functions $A_{\sigma,S}(\hat{q},\hat{p})$ defined on the
Hilbert space $L^2(\mathbb{R})$. Moreover, from Theorems \ref{thm:9.8} and \ref{thm:9.5} it follows that the action of observables, treated as operator
functions $A_{\sigma,S}(\hat{q}_{\sigma,S},\hat{p}_{\sigma,S})$, on states, treated as operators $\hat{\Psi} \in \hat{\mathcal{H}}$,
is equivalent to the action of observables, treated as operator functions $A_{\sigma,S}(\hat{q},\hat{p})$, on states, treated as operators
$\hat{\rho} \in \mathcal{S}^2(L^2(\mathbb{R}))$. In fact, there holds
\begin{equation*}
A_{\sigma,S}(\hat{q}_{\sigma,S},\hat{p}_{\sigma,S}) \hat{\Psi} = \hat{1} \otimes_{\sigma,S} A_{\sigma,S}(\hat{q},\hat{p}) \hat{\rho}, \qquad
\hat{\Psi} A_{\sigma,S}(\hat{q}_{\sigma,S},\hat{p}_{\sigma,S}) = \hat{1} \otimes_{\sigma,S} \hat{\rho} A_{\sigma,S}(\hat{q},\hat{p}),
\end{equation*}
where $\hat{\Psi} = \hat{1} \otimes_{\sigma,S} \hat{\rho}$.

Using Theorem \ref{thm:9.5} a formula for the expectation value of observables represented as operators on the Hilbert space $L^2(\mathbb{R})$ can
be derived. Namely, there holds
\begin{theorem}
Let $\hat{A} \in \hat{\mathcal{A}}_Q$ be some observable and $A_{\sigma,S}(\hat{q},\hat{p})$ be a corresponding operator in the Hilbert space
$L^2(\mathbb{R})$. Moreover, let $\Psi = \sum_\lambda p_\lambda \Psi_{\mathrm{pure}}^{(\lambda)} = \sum_\lambda p_\lambda \left( \varphi^{(\lambda)}
\right)^* \otimes_{\sigma,S} \varphi^{(\lambda)} \in \mathcal{H}$ be some mixed state and $\hat{\rho} = \sum_\lambda p_\lambda
\braket{\varphi^{(\lambda)}|\,\cdot\,}_{L^2} \varphi^{(\lambda)}$ the corresponding density operator. Then there holds
\begin{equation}
\braket{\hat{A}}_{\Psi} = \sum_\lambda p_\lambda \braket{\varphi^{(\lambda)} | A_{\sigma,S}(\hat{q},\hat{p}) \varphi^{(\lambda)}}_{L^2}
= \tr(\hat{\rho} A_{\sigma,S}(\hat{q},\hat{p})).
\label{eq:9.50}
\end{equation}
\end{theorem}
\begin{proof}
From Theorem \ref{thm:9.5} it follows
\begin{align*}
\braket{\hat{A}}_{\Psi} & = \frac{1}{\sqrt{2\pi\hbar}} \iint \d{x}\d{p} (A \star_{\sigma,S} \Psi)(x,p)
= \frac{1}{\sqrt{2\pi\hbar}} \sum_{\lambda} p_{\lambda} \iint \d{x}\d{p} (A \star_{\sigma,S} \Psi^{(\lambda)}_{\textrm{pure}})(x,p) \displaybreak[0] \\
& = \frac{1}{\sqrt{2\pi\hbar}} \sum_{\lambda} p_{\lambda} \iint \d{x}\d{p} \left( \big( \varphi^{(\lambda)} \big)^* \otimes_{\sigma,S}
A_{\sigma,S}(\hat{q},\hat{p}) \varphi^{(\lambda)} \right)(x,p) \displaybreak[0] \\
& = \frac{1}{2\pi\hbar} \sum_{\lambda} p_{\lambda} \iint \d{x}\d{p} \int \d{y} e^{-\frac{i}{\hbar} py} \big( \varphi^{(\lambda)} \big)^* (x - \bar{\sigma} y)
A_{\sigma,S}(\hat{q},\hat{p}) \varphi^{(\lambda)}(x + \sigma y) \displaybreak[0] \\
& = \sum_{\lambda} p_{\lambda} \iint \d{x}\d{y} \delta(y) \big( \varphi^{(\lambda)} \big)^* (x - \bar{\sigma} y) A_{\sigma,S}(\hat{q},\hat{p})
\varphi^{(\lambda)}(x + \sigma y) \displaybreak[0] \\
& = \sum_{\lambda} p_{\lambda} \int \d{x} \big( \varphi^{(\lambda)} \big)^* (x) A_{\sigma,S}(\hat{q},\hat{p}) \varphi^{(\lambda)}(x) \displaybreak[0] \\
& = \sum_{\lambda} p_{\lambda} \braket{\varphi^{(\lambda)}|A_{\sigma,S}(\hat{q},\hat{p}) \varphi^{(\lambda)}}_{L^2}
= \tr(\hat{\rho} A_{\sigma,S}(\hat{q},\hat{p})).
\end{align*}
\end{proof}

\begin{corollary}
\begin{equation}
\braket{\hat{A}}_{\Psi} = \sum_{\lambda} p_{\lambda} \braket{\Psi_{\mathrm{pure}}^{(\lambda)} | A_L \star_{\sigma,S}
\Psi_{\mathrm{pure}}^{(\lambda)}}_{\mathcal{H}}
= \sum_{\lambda} p_{\lambda} \braket{\Psi_{\mathrm{pure}}^{(\lambda)} | A_R \star_{\sigma,S} \Psi_{\mathrm{pure}}^{(\lambda)}}_{\mathcal{H}}
\label{eq:9.51}
\end{equation}
\end{corollary}
\begin{proof}
Equation (\ref{eq:9.51}) follows immediately from (\ref{eq:9.50}) and (\ref{eq:9.40}) as from one side
\begin{equation*}
\braket{\hat{A}}_{\Psi} = \sum_{\lambda} p_{\lambda} \braket{\varphi^{(\lambda)} | A_{\sigma,S}(\hat{q},\hat{p}) \varphi^{(\lambda)}}_{L^2}
= \sum_{\lambda} p_{\lambda} \braket{\varphi^{(\lambda)} | \varphi^{(\lambda)}}_{L^2}
\braket{\varphi^{(\lambda)} | A_{\sigma,S}(\hat{q},\hat{p}) \varphi^{(\lambda)}}_{L^2}
= \sum_{\lambda} p_{\lambda} \braket{\Psi_{\textrm{pure}}^{(\lambda)} | A_L \star_{\sigma,S} \Psi_{\textrm{pure}}^{(\lambda)}}_{\mathcal{H}}
\end{equation*}
and from the other side
\begin{equation*}
\braket{\hat{A}}_{\Psi} = \sum_{\lambda} p_{\lambda} \braket{\varphi^{(\lambda)} | A_{\sigma,S}(\hat{q},\hat{p}) \varphi^{(\lambda)}}_{L^2}
= \sum_{\lambda} p_{\lambda} \braket{A^\dagger_{\sigma,S}(\hat{q},\hat{p}) \varphi^{(\lambda)} | \varphi^{(\lambda)}}_{L^2}
\braket{\varphi^{(\lambda)} | \varphi^{(\lambda)}}_{L^2}
= \sum_{\lambda} p_{\lambda} \braket{\Psi_{\textrm{pure}}^{(\lambda)} | A_R \star_{\sigma,S} \Psi_{\textrm{pure}}^{(\lambda)}}_{\mathcal{H}}.
\end{equation*}
\end{proof}

Using the results of this section it is possible to prove a useful property of pure states used in Section \ref{subsec:6.7}. Namely
\begin{theorem}
\label{thm:9.6}
Let $A \in \mathcal{A}_Q$. A pure state function $\Psi = \varphi^* \otimes_{\sigma,S} \varphi \in \mathcal{H}$ satisfies the equation
\begin{equation}
[A,\Psi] = 0
\label{eq:9.52}
\end{equation}
if and only if it satisfies the following pair of $\star_{\sigma,S}$-genvalue equations
\begin{equation}
A \star_{\sigma,S} \Psi = a \Psi, \quad \Psi \star_{\sigma,S} A = a \Psi,
\label{eq:9.53}
\end{equation}
for some $a \in \mathbb{C}$.
\end{theorem}
\begin{proof}
It is obvious that if $\Psi$ satisfies (\ref{eq:9.53}) then it also satisfies (\ref{eq:9.52}). Lets assume that $\Psi$ satisfies (\ref{eq:9.52}).
Hence, it also satisfies
\begin{equation*}
A \star_{\sigma,S} \Psi \star_{\sigma,S} \Psi = \Psi \star_{\sigma,S} A \star_{\sigma,S} \Psi.
\end{equation*}
From the idempotent property of pure states the above equation implies
\begin{equation}
\frac{1}{\sqrt{2\pi\hbar}} A \star_{\sigma,S} \Psi = \Psi \star_{\sigma,S} A \star_{\sigma,S} \Psi.
\label{eq:9.54}
\end{equation}
From Theorem \ref{thm:9.5} it follows that
\begin{equation}
A \star_{\sigma,S} \Psi = \varphi^* \otimes_{\sigma,S} A_{\sigma,S}(\hat{q},\hat{p}) \varphi.
\label{eq:9.55}
\end{equation}
Now, equations (\ref{eq:9.54}) and (\ref{eq:9.55}) with the help of (\ref{eq:9.56}) give
\begin{equation*}
A \star_{\sigma,S} \Psi = \sqrt{2\pi\hbar} \Psi \star_{\sigma,S} (A \star_{\sigma,S} \Psi)
= \braket{\varphi | A_{\sigma,S}(\hat{q},\hat{p}) \varphi}_{L^2} \Psi = a \Psi,
\end{equation*}
where $a = \braket{\varphi | A_{\sigma,S}(\hat{q},\hat{p}) \varphi}_{L^2}$. The second $\star_{\sigma,S}$-genvalue equation can be derived analogically.
\end{proof}

Finally, lets derive the time evolution of the observables and states represented as operators on the Hilbert space $L^2(\mathbb{R})$. In this case
the time evolution is governed by a Hermitian operator $H_{\sigma,S}(\hat{q},\hat{p})$ corresponding to the Hamiltonian $\hat{H}$. From the
time evolution equation (\ref{eq:6.7.1}) one receives the following evolution equation for density operators $\hat{\rho}$, called the \emph{von Neumann
equation}
\begin{equation}
i\hbar \frac{\partial \hat{\rho}}{\partial t} - [H_{\sigma,S}(\hat{q},\hat{p}),\hat{\rho}] = 0.
\label{eq:9.57}
\end{equation}

For a pure state density operator $\hat{\rho} = \braket{\varphi| \,\cdot\,}_{L^2} \varphi$ equation (\ref{eq:9.57}) takes the form
\begin{equation}
i\hbar \frac{\partial \varphi}{\partial t} = H_{\sigma,S}(\hat{q},\hat{p}) \varphi.
\label{eq:9.58}
\end{equation}
The above equation is called the \emph{time dependent Schr\"odinger equation}.

The equation for stationary states takes now the form
\begin{equation*}
[H_{\sigma,S}(\hat{q},\hat{p}),\hat{\rho}] = 0,
\end{equation*}
which for pure states $\hat{\rho} = \braket{\varphi| \,\cdot\,}_{L^2} \varphi$ is equivalent to such eigenvalue equation
\begin{equation*}
H_{\sigma,S}(\hat{q},\hat{p}) \varphi = E \varphi
\end{equation*}
called the \emph{stationary Schr\"odinger equation}.

The representation in the Hilbert space $L^2(\mathbb{R})$ of the one parameter group of unitary functions $U(t)$ from equation (\ref{eq:6.7.4}) is a
one parameter group of unitary operators
\begin{equation*}
U_{\sigma,S}(\hat{q},\hat{p},t) = e^{-\frac{i}{\hbar}t H_{\sigma,S}(\hat{q},\hat{p})}.
\end{equation*}
The time evolution of a density operator $\hat{\rho}$ can be alternatively expressed by the equation
\begin{equation*}
\hat{\rho}(t) = U_{\sigma,S}(\hat{q},\hat{p},t) \hat{\rho}(0) U_{\sigma,S}(\hat{q},\hat{p},-t).
\end{equation*}
It is then easy to check that the above equation is indeed a solution to the von Neumann equation (\ref{eq:9.57}). Using the unitary operators
$U_{\sigma,S}(\hat{q},\hat{p},t)$ also the time evolution of observables $A_{\sigma,S}(\hat{q},\hat{p})$ can be expressed, similarly as in equation
(\ref{eq:6.7.6})
\begin{equation*}
A_{\sigma,S}(\hat{q},\hat{p},t) = U_{\sigma,S}(\hat{q},\hat{p},-t) A_{\sigma,S}(\hat{q},\hat{p},0) U_{\sigma,S}(\hat{q},\hat{p},t).
\end{equation*}
The corresponding time evolution equation for observables $A_{\sigma,S}(\hat{q},\hat{p})$ from equation (\ref{eq:6.7.7}) takes the form
\begin{equation*}
i\hbar \frac{\d{}}{\d{t}}A_{\sigma,S}(\hat{q},\hat{p},t) - [A_{\sigma,S}(\hat{q},\hat{p},t), H_{\sigma,S}(\hat{q},\hat{p})] = 0.
\end{equation*}
The above equation is called the \emph{Heisenberg equation}.

\section{Examples}
\label{sec:12}
In this section some examples of the presented formalism of the phase space quantum mechanics will be given. First a free particle will be considered and
then a simple harmonic oscillator.

\subsection{Free particle}
\label{subsec:12.1}
In this section a free particle will be considered. For simplicity a one dimensional particle (the case of $N = 1$) will be considered. The free particle
is a system, which time evolution is governed by a Hamiltonian $\hat{H}$ induced by the function
\begin{equation}
H(x,p) = \frac{1}{2} p^2,
\label{eq:12.1.1}
\end{equation}
where the mass of the particle $m = 1$. This Hamiltonian describes only the kinetic energy of the particle. It does not contain any terms describing the
potential energy, i.e. there are no forces acting on the particle (the particle is free).

First, lets find a time evolution of a free particle being initially in an arbitrary pure state. To do this it is necessary to solve the time evolution
equation (\ref{eq:6.7.1})
\begin{equation}
i\hbar \frac{\partial \Psi}{\partial t} - [H,\Psi] = 0,
\label{eq:12.1.2}
\end{equation}
with $H$ given by (\ref{eq:12.1.1}) and with the assumption that the solution $\Psi$ is in a form of a pure state, i.e.
$\Psi = \varphi^* \otimes_{\sigma,S} \varphi$ for some function $\varphi \in L^2(\mathbb{R})$. From Section \ref{sec:9} it is known that
the function $\Psi = \varphi^* \otimes_{\sigma,S} \varphi$ is a solution to (\ref{eq:12.1.2}) if and only if function
$\varphi$ is a solution to the Schr\"odinger equation (\ref{eq:9.58}), which for $H$ given by (\ref{eq:12.1.1}) takes the form
\begin{equation}
i\hbar \frac{\partial \varphi}{\partial t} = -\frac{\hbar^2}{2} \frac{\partial^2 \varphi}{\partial x^2}.
\label{eq:12.1.3}
\end{equation}
The simplest solution to equation (\ref{eq:12.1.3}) is of the form of a plain wave
\begin{equation*}
\varphi_p(x,t) = e^{\frac{i}{\hbar}(px - E(p)t)},
\end{equation*}
where $p \in \mathbb{R}$ and $E(p) = \frac{1}{2} p^2$. The general solution to (\ref{eq:12.1.3}) is in the form of a linear combination of the plain wave
solutions $\varphi_p$, i.e. in the form of a wave packet
\begin{equation}
\varphi(x,t) = \frac{1}{\sqrt{2\pi\hbar}} \int f(p) e^{\frac{i}{\hbar}(px - E(p)t)} \d{p}
= \frac{1}{\sqrt{2\pi\hbar}} \int g(p,t) e^{\frac{i}{\hbar} px} \d{p},
\label{eq:12.1.5}
\end{equation}
where $f \in L^2(\mathbb{R})$ and $g(p,t) = f(p) e^{-\frac{i}{\hbar} E(p)t}$. From (\ref{eq:12.1.5}) the solution $\Psi$ of (\ref{eq:12.1.2}) reads
\begin{equation}
\Psi(x,p,t) = (\varphi^* \otimes_{\sigma,S} \varphi)(x,p,t)
= \frac{1}{\sqrt{2\pi\hbar}} S \int \d{p_1} g^*(p - \sigma p_1,t) g(p + \bar{\sigma} p_1,t) e^{\frac{i}{\hbar}p_1 x}.
\label{eq:12.1.6}
\end{equation}

Lets consider some particular cases of the solution (\ref{eq:12.1.6}). Assume that
\begin{equation*}
f(p) = \frac{1}{(2\pi)^{1/4} (\Delta p)^{1/2}} e^{-\frac{(p - p_0)^2}{4(\Delta p)^2}}
\end{equation*}
is a Gaussian function. By (\ref{eq:12.1.5}) the solution $\varphi$ of the Schr\"odinger equation (\ref{eq:12.1.3}) takes the form
\begin{equation*}
\varphi(x,t) = \frac{1}{(2\pi)^{1/4} \sqrt{\Delta x + i\Delta p t}} \exp \left(-\frac{p_0^2}{4(\Delta p)^2} \right)
\exp \left( -\frac{(x - i\frac{\Delta x}{\Delta p} p_0)^2}{4(\Delta x)^2 + 4i\Delta x \Delta p t} \right),
\end{equation*}
where $\Delta x = \frac{\hbar}{2 \Delta p}$. This solution describes the time evolution of a wave packet initially in the form of a Gaussian-like function
\begin{equation*}
\varphi(x,0) = \frac{1}{(2\pi)^{1/4} (\Delta x)^{1/2}} e^{-\frac{x^2}{4(\Delta x)^2}} e^{\frac{i}{\hbar} p_0 x}.
\end{equation*}
It is now possible to calculate the solution $\Psi$ of the time evolution equation (\ref{eq:12.1.2}). For simplicity only the case of $S = 1$
will be considered. The function $\Psi$ takes the form
\begin{equation}
\Psi(x,p,t) = \frac{1}{\sqrt{2\pi \left( (\bar{\sigma}^2 + \sigma^2) \Delta x \Delta p + i(1 - 2\sigma)(\Delta p)^2 t \right)}}
\exp \left( -\frac{(p - p_0)^2}{2(\Delta p)^2} \right)
\exp \left( -\frac{(x - pt + i(1 - 2\sigma) \frac{\Delta x}{\Delta p}(p - p_0))^2}
{4(\bar{\sigma}^2 + \sigma^2)(\Delta x)^2 + 4i(1 - 2\sigma)\Delta x \Delta p t} \right).
\label{eq:12.1.10}
\end{equation}
This solution describes the time evolution of a free particle initially in the state
\begin{equation*}
\Psi(x,p,0) = \frac{1}{\sqrt{2\pi (\bar{\sigma}^2 + \sigma^2) \Delta x \Delta p}}
\exp \left( -\frac{(p - p_0)^2}{2(\Delta p)^2} \right)
\exp \left( -\frac{(x + i(1 - 2\sigma) \frac{\Delta x}{\Delta p}(p - p_0))^2}{4(\bar{\sigma}^2 + \sigma^2)(\Delta x)^2} \right).
\end{equation*}
The state $\Psi$ from (\ref{eq:12.1.10}) greatly simplifies in the case $\sigma = \frac{1}{2}$. For this special case the state $\Psi$ reads
\begin{equation*}
\Psi(x,p,t) = \frac{1}{\sqrt{\pi \Delta x \Delta p}} \exp \left( -\frac{(p - p_0)^2}{2(\Delta p)^2} \right)
\exp \left( -\frac{(x - pt)^2}{2(\Delta x)^2} \right).
\end{equation*}

Lets calculate the expectation values and uncertainties of the position $x$ and momentum $p$ of a free particle described by the state (\ref{eq:12.1.10}).
One easily calculates that
\begin{align*}
&\braket{x}_{\Psi(t)} = \iint x \star_\sigma \Psi(t) \d{x}\d{p} = p_0t, \qquad
\braket{x^2}_{\Psi(t)} = \iint x^2 \star_\sigma \Psi(t) \d{x}\d{p} = (\Delta x)^2 + (\Delta p)^2 t^2 + p_0^2 t^2, \displaybreak[0] \\
&\Delta x(t)  = \sqrt{\braket{x^2}_{\Psi(t)} - \braket{x}_{\Psi(t)}^2} = \sqrt{(\Delta x)^2 + (\Delta p)^2 t^2}, \displaybreak[0] \\
&\braket{p}_{\Psi(t)}  = \iint p \star_\sigma \Psi(t) \d{x}\d{p} = p_0, \qquad
\braket{p^2}_{\Psi(t)}  = \iint p^2 \star_\sigma \Psi(t) \d{x}\d{p} = (\Delta p)^2 + p_0^2, \displaybreak[0] \\
&\Delta p(t) = \sqrt{\braket{p^2}_{\Psi(t)} - \braket{p}_{\Psi(t)}^2} = \Delta p.
\end{align*}
Note, that during the time evolution the uncertainty of the momentum $\Delta p(t)$ of the free particle described by the state (\ref{eq:12.1.10}) do not
change in time and is equal to its initial value $\Delta p$, whereas the uncertainty of the position $\Delta x(t)$ initially equal $\Delta x$ increases
in time. Note also, that the uncertainties of the position and momentum satisfy the Heisenberg uncertainty principle, i.e. $\Delta x(t) \Delta p(t) \ge
\frac{\hbar}{2}$. Moreover, initially the free particle is in a state which minimizes the Heisenberg uncertainty principle since $\Delta x(0) \Delta p(0)
= \Delta x \Delta p = \frac{\hbar}{2}$. Worth noting is also the fact that the expectation value of the momentum $\braket{p}_{\Psi(t)}$ is constant and
equal $p_0$, whereas the expectation value of the position $\braket{x}_{\Psi(t)}$ is equal $p_0 t$. Hence, the time evolution of the free
particle described by the state (\ref{eq:12.1.10}) can be interpreted as the movement of the particle along a straight line with the constant momentum
equal $p_0$, similarly as in the classical case. The difference between the classical and quantum case is that in the quantum case there is some
uncertainty of the position and momentum, in contrast to the classical case where the position and momentum is known precisely.

It is interesting to calculate to which classical state converges the state (\ref{eq:12.1.10}) in the limit $\hbar \to 0^+$. Assume that $\Delta x \propto
\sqrt{\hbar}$ and $\Delta p \propto \sqrt{\hbar}$, then $\frac{\Delta p}{\Delta x} = c = \textrm{const}$. To calculate to which classical state converges
the state $\Psi$ from (\ref{eq:12.1.10}) in the limit $\hbar \to 0^+$ it is necessary to calculate the limit $\lim_{\hbar \to 0^+} \rho$, where
$\rho = \frac{1}{\sqrt{2\pi\hbar}} \Psi$ is the quantum distribution function induced by $\Psi$. The limit $\lim_{\hbar \to 0^+} \rho$ has to be
calculated in the distributional sense, i.e. one have to calculate the limit $\lim_{\hbar \to 0^+} \braket{\rho,\phi}$ for every test function $\phi$.
One easily calculates that
\begin{equation*}
\lim_{\hbar \to 0^+} \braket{\rho,\phi} = \phi(p_0 t, p_0).
\end{equation*}
Hence
\begin{equation*}
\lim_{\hbar \to 0^+} \rho(x,p,t) = \delta(x - p_0 t) \delta(p - p_0).
\end{equation*}
Above equation implies that the state $\Psi$ from (\ref{eq:12.1.10}) describing the free particle converges in the limit $\hbar \to 0^+$ to the classical
pure state describing the free particle moving along a straight line with the constant momentum equal $p_0$.

Lets consider now another particular case of the solution (\ref{eq:12.1.6}). Assume that
\begin{equation*}
f(p) = \delta(p - p_0).
\end{equation*}
By (\ref{eq:12.1.5}) the solution $\varphi$ of the Schr\"odinger equation (\ref{eq:12.1.3}) takes the form of a plain wave
\begin{equation}
\varphi(x,t) = \frac{1}{\sqrt{2\pi\hbar}} e^{\frac{i}{\hbar}(p_0 x - E(p_0)t)}.
\label{eq:12.1.16}
\end{equation}
From (\ref{eq:12.1.16}) the solution $\Psi$ of (\ref{eq:12.1.2}), for the case of the $\star_{\sigma,\alpha,\beta}$-product with $\beta > 0$, reads
\begin{equation}
\Psi(x,p,t) = (\varphi^* \otimes_{\sigma,\alpha,\beta} \varphi)(x,p,t)
= \frac{1}{2\pi\hbar \sqrt{\beta}} e^{-\frac{1}{2\hbar \beta}(p - p_0)^2}.
\label{eq:12.1.17}
\end{equation}
In the limit $\beta \to 0^+$ equation (\ref{eq:12.1.17}) takes the form
\begin{equation*}
\Psi(x,p,t) = \frac{1}{\sqrt{2\pi\hbar}} \delta(p - p_0).
\end{equation*}
Note, that $\Psi$ is not a proper state since it does not belong to the space of states $\mathcal{H}$. Hence, $\Psi$ does not describe a physical system.
It, however, describes an idealized situation of a particle with the momentum known precisely and the position not known at all. Note also, that $\Psi$
does not depend on time $t$, i.e. $\Psi$ can be thought of as a stationary state of the system. In fact, $\Psi$ is a formal
$\star_{\sigma,\alpha,\beta}$-genfunction of $p$ and $H$, i.e. $\Psi$ formally satisfies the following $\star_{\sigma,\alpha,\beta}$-genvalue equations
\begin{align*}
p \star_{\sigma,\alpha,\beta} \Psi & = p_0 \Psi, &
\Psi \star_{\sigma,\alpha,\beta} p & = p_0 \Psi, \\
H \star_{\sigma,\alpha,\beta} \Psi & = E(p_0) \Psi, &
\Psi \star_{\sigma,\alpha,\beta} H & = E(p_0) \Psi.
\end{align*}
Hence, $p_0$ and $E(p_0)$ can be interpreted as the momentum and energy of the particle.

\subsection{Simple harmonic oscillator}
\label{subsec:12.2}
\subsubsection{Stationary states of the harmonic oscillator}
In the following example only the special case of the star-product will be considered, namely the $\star_{\sigma,\alpha,\beta}$-product.
Lets consider a Hamiltonian system describing a one dimensional ($N = 1$) simple harmonic oscillator. Its Hamiltonian $\hat{H}$ is induced by the
function
\begin{equation*}
H(x,p) = \frac{1}{2} \left( p^2 + \omega^2 x^2 \right),
\end{equation*}
where $\omega$ is the frequency of oscillations. Note that $H$ is a Hermitian function for every $(\sigma,\alpha,\beta)$-ordering, i.e.
$H \star_{\sigma,\alpha,\beta} {} = (H \star_{\sigma,\alpha,\beta} {})^\dagger$. Lets try to find stationary pure states of the harmonic oscillator.
From Section \ref{subsec:6.7} it is known that the stationary pure states are precisely the solutions of the following pair of
$\star_{\sigma,\alpha,\beta}$-genvalue equations
\begin{equation*}
H \star_{\sigma,\alpha,\beta} \Psi = E \Psi, \qquad
\Psi \star_{\sigma,\alpha,\beta} H = E \Psi,
\end{equation*}
for $E \in \mathbb{R}$. To solve the above equations it is convenient to introduce new coordinates called \emph{holomorphic coordinates}
\cite{Gracia-Bondia:1988}
\begin{equation*}
a(x,p) = \frac{\omega x + ip}{\sqrt{2\hbar\omega}}, \qquad \bar{a}(x,p) = \frac{\omega x - ip}{\sqrt{2\hbar\omega}}.
\end{equation*}
The functions $a$ and $\bar{a}$ are called the \emph{annihilation} and \emph{creation} functions since they decrease and increase the number of
excitations of the vibrational mode with frequency $\omega$ (annihilate and create the quanta of vibrations).
Note, that $a \star_{\sigma,\alpha,\beta} {} = (\bar{a} \star_{\sigma,\alpha,\beta} {})^\dagger$, $\bar{a} \star_{\sigma,\alpha,\beta} {}
= (a \star_{\sigma,\alpha,\beta} {})^\dagger$ and
\begin{equation*}
[a,\bar{a}] = a \star_{\sigma,\alpha,\beta} \bar{a} - \bar{a} \star_{\sigma,\alpha,\beta} a = 1.
\end{equation*}
In this new coordinates the function $H$ inducing the Hamiltonian $\hat{H}$ takes the form
\begin{equation*}
H(a,\bar{a}) = \hbar \omega a \bar{a} = \hbar \omega \left( \bar{a} \star_{\sigma,\alpha,\beta} a + \bar{\lambda} \right)
= \hbar \omega \left( a \star_{\sigma,\alpha,\beta} \bar{a} - \lambda \right),
\end{equation*}
where $\lambda = \frac{1}{2}(1 + \omega \alpha + \omega^{-1} \beta)$ and $\bar{\lambda} := 1 - \lambda = \frac{1}{2}(1 - \omega \alpha - \omega^{-1}
\beta)$.

Lets consider a more general problem of finding a solution to the following pair of $\star_{\sigma,\alpha,\beta}$-genvalue equations
\begin{subequations}
\label{eq:12.2.15}
\begin{align}
H \star_{\sigma,\alpha,\beta} \Psi_{mn} & = E_m \Psi_{mn},
\label{eq:12.2.15a} \\
\Psi_{mn} \star_{\sigma,\alpha,\beta} H & = E_n \Psi_{mn},
\label{eq:12.2.15b}
\end{align}
\end{subequations}
where $m,n$ are numbering the $\star_{\sigma,\alpha,\beta}$-genvalues of $H$. It can be shown that $m,n$ are non-negative integer numbers. The energy
levels $E_n$ of the harmonic oscillator are equal
\begin{equation*}
E_n = (n + \bar{\lambda}) \hbar \omega.
\end{equation*}
Since $H = \hbar \omega(\bar{a} \star_{\sigma,\alpha,\beta} a + \bar{\lambda})$, $\star_{\sigma,\alpha,\beta}$-genvalues of the function
$N := \bar{a} \star_{\sigma,\alpha,\beta} a$ are the natural numbers $n = 0,1,2,\ldots$ and $\star_{\sigma,\alpha,\beta}$-genfunctions are the
$\star_{\sigma,\alpha,\beta}$-genfunctions $\Psi_{mn}$ of $H$, i.e.
\begin{equation*}
N \star_{\sigma,\alpha,\beta} \Psi_{mn} = m \Psi_{mn}, \quad \Psi_{mn} \star_{\sigma,\alpha,\beta} N = n \Psi_{mn}.
\end{equation*}
Hence, the function $N = \bar{a} \star_{\sigma,\alpha,\beta} a$ can be interpreted as an observable of the number of excitations of the vibrational mode
with frequency $\omega$.

Moreover, the normalized solutions of equations (\ref{eq:12.2.15}) can be calculated from the ground state $\Psi_{00}$ according to the equation
\begin{equation}
\Psi_{mn} = \frac{1}{\sqrt{m!n!}} \underbrace{\bar{a} \star_{\sigma,\alpha,\beta} \ldots \star_{\sigma,\alpha,\beta} \bar{a}}_m
\star_{\sigma,\alpha,\beta} \Psi_{00} \star_{\sigma,\alpha,\beta} \underbrace{a \star_{\sigma,\alpha,\beta} \ldots \star_{\sigma,\alpha,\beta} a}_n.
\label{eq:12.2.28}
\end{equation}
On the other hand, the ground state $\Psi_{00}$ takes the form
\begin{subequations}
\label{eq:12.2.29}
\begin{align}
\Psi_{00}(a,\bar{a}) & = \frac{1}{\sqrt{2\pi\hbar}} \frac{\sqrt{(1 - 2\sigma)^2 + (1 + 2 \omega \alpha)(1 + 2 \omega^{-1} \beta)}}
{\bar{\sigma}^2 + \sigma^2 + 2 \alpha \beta + \omega \alpha + \omega^{-1} \beta}
\exp \left( \frac{-(1 + \omega \alpha + \omega^{-1} \beta) a \bar{a}}{\bar{\sigma}^2 + \sigma^2 + 2 \alpha \beta + \omega \alpha
+ \omega^{-1} \beta} \right) \nonumber \\
& \quad \cdot \exp \left( \frac{-\frac{1}{2}(1 - 2\sigma - \omega \alpha + \omega^{-1} \beta) a^2
+ \frac{1}{2}(1 - 2\sigma + \omega \alpha - \omega^{-1} \beta) \bar{a}^2}{\bar{\sigma}^2 + \sigma^2 + 2 \alpha \beta + \omega \alpha
+ \omega^{-1} \beta} \right)
\label{eq:12.2.29a}
\end{align}
or after the change of coordinates
\begin{equation}
\Psi_{00}(x,p) = \frac{1}{\sqrt{2\pi\hbar}} \frac{\sqrt{(1 - 2\sigma)^2 + (1 + 2 \omega \alpha)(1 + 2 \omega^{-1} \beta)}}
{\bar{\sigma}^2 + \sigma^2 + 2 \alpha \beta + \omega \alpha + \omega^{-1} \beta}
\exp \left( \frac{-(1 + 2 \omega^{-1} \beta) \omega^2 x^2
- (1 + 2 \omega \alpha) p^2 - i2(1 - 2\sigma) \omega xp}{2 \hbar \omega(\bar{\sigma}^2 + \sigma^2 + 2 \alpha \beta + \omega \alpha + \omega^{-1} \beta)}
\right).
\label{eq:12.2.29b}
\end{equation}
\end{subequations}

In what follows the $\star_{\sigma,\alpha,\beta}$-genfunctions $\Psi_{mn}$ will be calculated using equation (\ref{eq:12.2.28}). To simplify calculations
the special case of the $(\sigma,\alpha,\beta)$-ordering will be considered, namely the case with $\sigma = \frac{1}{2}$ and $\beta = \omega^2 \alpha$.
To simplify the notation the $\star_{\frac{1}{2},\alpha,\omega^2 \alpha}$-product will be denoted by $\star$. For this special ordering equations
(\ref{eq:12.2.29}) for the ground state take the form
\begin{equation*}
\Psi_{00}(a,\bar{a}) = \frac{1}{\sqrt{2\pi\hbar} \lambda} \exp \left( -\frac{a \bar{a}}{\lambda} \right), \qquad
\Psi_{00}(x,p) = \frac{1}{\sqrt{2\pi\hbar} \lambda} \exp \left( -\frac{p^2 + \omega^2 x^2}{2 \lambda \hbar \omega} \right),
\end{equation*}
where now $\lambda = \frac{1}{2}(1 + 2 \omega \alpha)$. The $\star_{\sigma,\alpha,\beta}$-genfunctions $\Psi_{mn}$ can be now calculated giving
\begin{equation}
\Psi_{mn}(a,\bar{a}) = \frac{1}{\sqrt{m!n!}} \sum_{k=0}^n (-1)^k k! \binom{m}{k} \binom{n}{k} \bar{\lambda}^k \left( \frac{1}{\lambda} \right)^{m+n-k}
\bar{a}^{m-k} a^{n-k} \Psi_{00}(a,\bar{a}).
\label{eq:12.2.40}
\end{equation}
Above equation can be written differently when passing to the polar coordinates $(r,\theta)$
\begin{equation*}
\omega x + ip = r e^{i\theta}.
\end{equation*}
In this new coordinates
\begin{gather*}
a(r,\theta) = \frac{1}{\sqrt{2\hbar\omega}} r e^{i\theta}, \quad \bar{a}(r,\theta) = \frac{1}{\sqrt{2\hbar\omega}} r e^{-i\theta},\quad
r^2 = p^2 + \omega^2 x^2,
\end{gather*}
and equation (\ref{eq:12.2.40}) takes the form
\begin{equation}
\Psi_{mn}(r,\theta) = \frac{1}{\sqrt{2\pi\hbar} \lambda} (-1)^n \sqrt{\frac{n!}{m!}} \frac{\bar{\lambda}^n}{\lambda^m}
\left( \frac{r}{\sqrt{2\hbar\omega}} \right)^{m-n} L_n^{m-n} \left( \frac{r^2}{2 \hbar \omega \lambda \bar{\lambda}} \right)
e^{-i(m-n)\theta} \exp \left( -\frac{r^2}{2 \hbar \omega \lambda} \right),
\label{eq:12.2.41}
\end{equation}
where
\begin{equation*}
L_n^s(x) = \frac{x^{-s} e^x}{n!} \frac{\d{}^n}{\d{x^n}} \left( e^{-x} x^{n+s} \right) = \sum_{k=0}^n (-1)^k \frac{(n+s)!}{(n-k)!(s+k)!k!} x^k
\end{equation*}
are the Laguerre's polynomials. This result for $\lambda =\frac{1}{2}$ was derived in \cite{Groenewold:1946, Bartlett:1949, Fairlie:1964}. The stationary
pure states of the harmonic oscillator are the functions
\begin{equation}
\Psi_{nn}(r,\theta) = \frac{1}{\sqrt{2\pi\hbar} \lambda} (-1)^n \left( \frac{\bar{\lambda}}{\lambda} \right)^n L_n \left( \frac{r^2}{2 \hbar \omega
\lambda \bar{\lambda}} \right) \exp \left( -\frac{r^2}{2 \hbar \omega \lambda} \right),
\label{eq:12.2.42}
\end{equation}
where $L_n(x) = L_n^0(x)$. Equation (\ref{eq:12.2.42}) can be also written in the following form, independent on a coordinate system
\begin{equation*}
\Psi_{nn} = \frac{1}{\sqrt{2\pi\hbar} \lambda} (-1)^n \left( \frac{\bar{\lambda}}{\lambda} \right)^n L_n \left( \frac{H}{\hbar \omega \lambda
\bar{\lambda}} \right) \exp \left( -\frac{H}{\hbar \omega \lambda} \right).
\end{equation*}

Equations (\ref{eq:12.2.41}) and (\ref{eq:12.2.42}) are valid for $\lambda \neq 0,1$ but it can be easily calculated how this equations look like in the
limits $\lambda \to 0$ and $\lambda \to 1$. For the case $\lambda \to 1$ one gets
\begin{align*}
\Psi_{mn}(r,\theta) & = \frac{1}{\sqrt{2\pi\hbar m! n!}} \left(\frac{r}{\sqrt{2\hbar\omega}} \right)^{m+n} e^{-i(m-n)\theta}
\exp \left( -\frac{r^2}{2 \hbar \omega} \right), \\
\Psi_{nn}(r,\theta) & = \frac{1}{\sqrt{2\pi\hbar} n!} \left(\frac{r^2}{2\hbar\omega} \right)^n \exp \left( -\frac{r^2}{2 \hbar \omega} \right).
\end{align*}
Moreover, for the case $\lambda \to 0$ one gets
\begin{align*}
\Psi_{mn}(a,\bar{a}) & = \frac{1}{2\hbar \sqrt{m!n!}} (-1)^{m+n} \sum_{k=0}^n k! \binom{m}{k} \binom{n}{k} \frac{\partial^{m-k}}{\partial a^{m-k}}
\frac{\partial^{n-k}}{\partial \bar{a}^{n-k}} \delta^{(2)}(a), \\
\Psi_{nn}(a,\bar{a}) & = \frac{1}{2\hbar} \sum_{k=0}^n \frac{1}{k!} \binom{n}{k} \frac{\partial^{2k}}{\partial a^k \partial \bar{a}^k} \delta^{(2)}(a).
\end{align*}
It is worth noting that the stationary states $\Psi_{nn}$, for the case $\lambda = 0$, are some distributions which cannot be identified with actual
functions. This shows that the space $\mathcal{H}$ is in general, for certain orderings (for $(\sigma,\alpha,\beta)$-orderings for which
$\alpha,\beta$ are negative), the space of distributions.

It is interesting to check to which classical states converge quantum states $\Psi_{nn}$ in the limit $\hbar \to 0^+$. A quantum distribution function
$\rho_n = \frac{1}{\sqrt{2\pi\hbar}} \Psi_{nn}$ reads
\begin{equation*}
\rho_n(x,p) = \frac{1}{2\pi\hbar \lambda} (-1)^n \left( \frac{\bar{\lambda}}{\lambda} \right)^n L_n \left( \frac{p^2 + \omega^2 x^2}{2 \hbar \omega
\lambda \bar{\lambda}} \right) \exp \left( -\frac{p^2 + \omega^2 x^2}{2 \hbar \omega \lambda} \right).
\end{equation*}
The limit $\hbar \to 0^+$ has to be calculated in a distributional sense, hence the limit $\lim_{\hbar \to 0^+} \braket{\rho_n, \phi}$ has to be
calculated for every test function $\phi$. One have that
\begin{equation*}
\lim_{\hbar \to 0^+} \braket{\rho_n, \phi} = \phi(0,0) = \braket{\delta_{(0,0)}, \phi}.
\end{equation*}
Hence
\begin{equation*}
\lim_{\hbar \to 0^+} \rho_n = \delta_{(0,0)},
\end{equation*}
i.e. the quantum stationary pure states $\rho_n$ of the harmonic oscillator converge, in the limit $\hbar \to 0^+$, to the classical state $(x = 0,p = 0)$
describing a particle with the position and momentum equal 0. This result is not surprising as the state $(x = 0,p = 0)$ is the only classical stationary
pure state of the harmonic oscillator.

\subsubsection{Coherent states of the harmonic oscillator}
As with the case of the stationary states of the harmonic oscillator in what follows only the special case of the $\star_{\sigma,\alpha,\beta}$-product
will be considered. Coherent states of the harmonic oscillator are functions $\Psi_{z_1,z_2} \in \mathcal{H}$ which satisfy the following
$\star_{\sigma,\alpha,\beta}$-genvalue equations
\begin{subequations}
\label{eq:12.2.47}
\begin{align}
a_L \star_{\sigma,\alpha,\beta} \Psi_{z_1,z_2} & = z_1 \Psi_{z_1,z_2},
\label{eq:12.2.47a} \\
\bar{a}_R \star_{\sigma,\alpha,\beta} \Psi_{z_1,z_2} & = z_2^* \Psi_{z_1,z_2},
\label{eq:12.2.47b}
\end{align}
\end{subequations}
where $z_1,z_2 \in \mathbb{C}$. Functions $\Psi_z := \Psi_{z,z}$ are then the admissible pure states. It will be shown later that coherent states are
states which resemble the classical pure states the most. In fact, their time evolution is close to the time evolution of the classical pure states.
Moreover, it can be shown that coherent states minimize the Heisenberg uncertainty principle, i.e. $\Delta x \Delta p = \hbar/2$. This once again shows
that coherent states are the best realization of the classical pure states. Indeed, the classical pure states are those states for which the position and
momentum is known precisely. However, in quantum mechanics we cannot know the precise position and momentum of a particle due to the Heisenberg
uncertainty principle $\Delta x \Delta p \ge \hbar/2$, hence the states which minimize the uncertainty principle are the best realizations of the
classical pure states.

Equations (\ref{eq:12.2.47}), for $z_1 = (\omega x_1 + ip_1) / \sqrt{2 \hbar \omega}$ and $z_2 = (\omega x_2 + ip_2)
/ \sqrt{2 \hbar \omega}$, are equivalent to a system of the following two differential equations:
\begin{subequations}
\label{eq:12.2.49}
\begin{align}
(\omega (x - x_1) + i(p - p_1)) \Psi_{z_1,z_2} + \hbar ((\bar{\sigma} + \omega \alpha) \partial_x + i(\sigma \omega + \beta) \partial_p) \Psi_{z_1,z_2}
& = 0,
\label{eq:12.2.49a} \\
(\omega (x - x_2) - i(p - p_2)) \Psi_{z_1,z_2} + \hbar ((\sigma + \omega \alpha) \partial_x - i(\bar{\sigma} \omega + \beta) \partial_p) \Psi_{z_1,z_2}
& = 0.
\label{eq:12.2.49b}
\end{align}
\end{subequations}

To not receive to complicated equations only the case of $\alpha = \beta = 0$ will be considered. So, the solution to the system of differential
equations (\ref{eq:12.2.49}) for $\alpha = \beta = 0$ reads
\begin{align}
\Psi_{z_1,z_2}(x,p) & = \frac{1}{\sqrt{\pi \hbar \omega (\bar{\sigma}^2 + \sigma^2)}} \exp \left( \frac{i(p_1 - p_2)x}{\hbar} \right)
\exp \left( -\omega \frac{(x - x_1)^2 + (x - x_2)^2}{2 \hbar} \right) \nonumber \\
& \quad \cdot \exp \left( \frac{(\sigma(\omega(x - x_1) + i(p - p_1)) - \bar{\sigma}(\omega(x - x_2) - i(p - p_2)))^2}{2 \hbar \omega(\bar{\sigma}^2
+ \sigma^2)} \right).
\label{eq:12.2.51}
\end{align}
Functions $\Psi_{z_1,z_2}$ which correspond to actual states are those for which $z_1 = z_2 = \bar{z} = (\omega \bar{x} + i\bar{p})/\sqrt{2\hbar \omega}$.
The equation (\ref{eq:12.2.51}) can be written then in the form
\begin{equation*}
\Psi_{\bar{z}}(x,p) = \frac{1}{\sqrt{\pi \hbar \omega(\bar{\sigma}^2 + \sigma^2)}} \exp \left( -\frac{\omega^2 (x - \bar{x})^2}
{2 \hbar \omega(\bar{\sigma}^2 + \sigma^2)} \right) \exp \left( -\frac{(p - \bar{p})^2}{2 \hbar \omega(\bar{\sigma}^2 + \sigma^2)} \right)
\exp \left( i\frac{2(2\sigma - 1)\omega(x - \bar{x})(p - \bar{p})}{2 \hbar \omega(\bar{\sigma}^2 + \sigma^2)} \right).
\end{equation*}
A quantum distribution function induced by $\Psi_{\bar{z}}$ is then given by
\begin{align}
\rho(x,p) & = \frac{1}{\sqrt{2 \pi \hbar}} \Psi_{\bar{z}}(x,p) \nonumber \\
& = \frac{1}{\pi \hbar \sqrt{2 \omega(\bar{\sigma}^2 + \sigma^2)}}
\exp \left( -\frac{\omega^2 (x - \bar{x})^2}{2 \hbar \omega(\bar{\sigma}^2 + \sigma^2)} \right) \exp \left( -\frac{(p - \bar{p})^2}
{2 \hbar \omega(\bar{\sigma}^2 + \sigma^2)} \right)
\exp \left( i\frac{2(2\sigma - 1)\omega(x - \bar{x})(p - \bar{p})}{2 \hbar \omega(\bar{\sigma}^2 + \sigma^2)} \right).
\label{eq:12.2.53}
\end{align}
Note that the expectation value of the position and momentum in a coherent state $\Psi_{\bar{z}}$ is equal respectively $\bar{x}$ and $\bar{p}$.

Lets consider now the time evolution of the quantum distribution functions $\rho$ from equation (\ref{eq:12.2.53}). To find out how the functions $\rho$
develop in time it is necessary to solve the time evolution equation (\ref{eq:6.7.1})
\begin{equation}
i\hbar \frac{\partial \rho}{\partial t} - [H,\rho] = 0,
\label{eq:12.2.54}
\end{equation}
where $H(x,p) = \frac{1}{2}(\omega^2 x^2 + p^2)$. From the definition of the $\star_\sigma$-product it easily follows that the time evolution equation
(\ref{eq:12.2.54}) takes the form
\begin{equation*}
\frac{\partial \rho}{\partial t} - \omega^2 x \frac{\partial \rho}{\partial p} + p \frac{\partial \rho}{\partial x}
- i\hbar \omega^2 \frac{1}{2} (2\sigma - 1) \frac{\partial^2 \rho}{\partial p^2}
+ i\hbar \frac{1}{2} (2\sigma - 1) \frac{\partial^2 \rho}{\partial x^2} = 0.
\end{equation*}
The solution of the above equation with the initial condition equal
\begin{equation*}
\rho(x,p,0) = \frac{1}{\pi \hbar \sqrt{2 \omega(\bar{\sigma}^2 + \sigma^2)}}
\exp \left( -\frac{\omega^2 (x - x_0)^2}{2 \hbar \omega(\bar{\sigma}^2 + \sigma^2)} \right) \exp \left( -\frac{(p - p_0)^2}
{2 \hbar \omega(\bar{\sigma}^2 + \sigma^2)} \right)
\exp \left( i\frac{2(2\sigma - 1)\omega(x - x_0)(p - p_0)}{2 \hbar \omega(\bar{\sigma}^2 + \sigma^2)} \right),
\end{equation*}
reads
\begin{equation*}
\rho(x,p,t) = \frac{1}{\pi \hbar \sqrt{2 \omega(\bar{\sigma}^2 + \sigma^2)}}
\exp \left( -\frac{\omega^2 (x - \bar{x}(t))^2}{2 \hbar \omega(\bar{\sigma}^2 + \sigma^2)} \right) \exp \left( -\frac{(p - \bar{p}(t))^2}
{2 \hbar \omega(\bar{\sigma}^2 + \sigma^2)} \right)
\exp \left( i\frac{2(2\sigma - 1)\omega(x - \bar{x}(t))(p - \bar{p}(t))}{2 \hbar \omega(\bar{\sigma}^2 + \sigma^2)} \right),
\end{equation*}
where
\begin{equation*}
\bar{x}(t)  = x_0 \cos \omega t + \frac{p_0}{\omega} \sin \omega t, \qquad
\bar{p}(t) = -\omega x_0 \sin \omega t + p_0 \cos \omega t.
\end{equation*}
Hence, it can be seen that the expectation values $\bar{x}(t)$ and $\bar{p}(t)$ of the position and momentum evolve in time like classical pure states
of the harmonic oscillator. But, it has to be remembered that the coherent states have some uncertainty $\Delta x$ and $\Delta p$ of the position and
momentum in contrary to the classical pure states. Hence, even though the coherent states resemble the classical pure states they differ from them.

It can be shown that the coherent states $\rho$ from equation (\ref{eq:12.2.53}) converge, in the limit $\hbar \to 0^+$, to the classical
pure states $(\bar{x},\bar{p})$ describing a particle with the position and momentum equal $\bar{x}$ and $\bar{p}$. This can be proved by showing that
\begin{equation*}
\lim_{\hbar \to 0^+} \braket{\rho,\phi} = \braket{\delta_{(\bar{x},\bar{p})},\phi} = \phi(\bar{x},\bar{p})
\end{equation*}
for every test function $\phi$, which further implies that
\begin{equation*}
\lim_{\hbar \to 0^+} \rho = \delta_{(\bar{x},\bar{p})}.
\end{equation*}

\section{Final remarks}
\label{sec:15}
In the paper the broad family of mathematically as well as physically equivalent quantizations of classical mechanics was presented. The equivalence of
quantizations is understood in the sense that for two given quantizations of classical mechanics there exist an isomorphism $S$ between algebras of
observables as well as an isomorphism, also denoted by $S$, between spaces of states of the two given quantizations. The isomorphism $S$ preserves the
linear structure, star-multiplication, deformed Poisson bracket and involution of the algebras of observables. Moreover, the isomorphism $S$ preserves
the structure of the Hilbert algebra of the spaces of states. This implies that for a given observable $A \in \mathcal{A}_Q$ describing some physical
quantity and associated to the first quantization scheme the corresponding observable, describing the same physical quantity, associated to the second
quantization scheme will be equal $SA$. The same thing holds for states. If $\Psi \in \mathcal{H}$ is some given state in the first quantization scheme
then $S\Psi$ is a corresponding state in the second quantization scheme.

Note that two equivalent observables $A$ and $A' = SA$ from two quantization schemes have the same spectra. Indeed, a star-genvalue equation of the
observable $A$ reads
\begin{equation*}
A \star \Psi = a \Psi.
\end{equation*}
Acting with the isomorphism $S$ to the both sides of the above equation gives
\begin{equation*}
SA \star' S\Psi = a S\Psi \quad \iff \quad A' \star' \Psi' = a \Psi',
\end{equation*}
which is the star-genvalue equation of the observable $A'$ corresponding to $A$. Since the star-genvalues of both star-genvalue equations are the same
the spectra of $A$ and $A'$ are the same.

Moreover, the expectation values of two equivalent observables $A$ and $A'$ in equivalent states $\Psi$ and $\Psi'$ respectively are the same. This
fact can be easily proved by using the assumption that $S$ vanishes under the integral sign. Note also that the time development in the two quantization
schemes is the same if it is described by two equivalent Hamiltonians $H$ and $H' = SH$.

The above three observations show that two given quantum Hamiltonian systems describe the same physical system if Hamiltonians of this two systems are
equivalent. Such quantum Hamiltonian systems will be referred to as equivalent. Note that two equivalent quantum Hamiltonian systems reduce to the same
classical Hamiltonian system in the limit $\hbar \to 0$. Worth noting is also the fact that with a given classical Hamiltonian system can be associated
the whole family of quantum Hamiltonian systems, described in the same quantization scheme, which reduce to the given classical system. The Hamiltonians
of this family of quantum Hamiltonian systems are, in general, $\hbar$-dependent and of course reduce, in the limit $\hbar \to 0$, to the Hamiltonian
of the classical Hamiltonian system.

As an example lets consider the Moyal quantization scheme, i.e. the case of $\sigma = \frac{1}{2}$. In this case all admissible quantum Hamiltonians are
real valued functions from $\mathcal{A}_Q$. Lets consider some classical Hamiltonian $H_C$. The admissible family of quantum Hamiltonians, reducing in
the limit $\hbar \to 0$ to the classical Hamiltonian $H_C$, is of the form
\begin{equation}
H_Q = H_C + \sum_{n = 1}^\infty \hbar^n H^{(n)},\label{eq:12.2.540}
\end{equation}
where $H^{(n)}$ are arbitrary smooth real valued functions which do not depend on $\hbar$. All quantum Hamiltonian systems corresponding to the
Hamiltonians $H_Q$ are proper quantizations of the classical Hamiltonian system corresponding to $H_C$. Of course, in general Hamiltonians $H_Q$ are
non-equivalent and have different spectra.

On the other hand, choosing any $H_C$ and admissible quantization $\star_{\frac{1}{2},S}$, i.e. such that $H_C$ is self-adjoint with respect to
$\star_{\frac{1}{2},S}$ (what means in our case that $S=\bar{S}$), one can re-express the quantization procedure in terms of Moyal quantization
$\star_{\frac{1}{2}}$ with an appropriate deformed Hamiltonian $H_Q=SH_C$, belonging to the family (\ref{eq:12.2.540}).

An interesting example is the $\sigma$-family of quantum Hamiltonian systems all endowed with the same Hamiltonian $H$, namely a natural Hamiltonian
\begin{equation*}
H(x,p) = \frac{1}{2}p^2 + V(x),
\end{equation*}
where the potential $V$ is some real valued smooth function. All systems from this family are equivalent since $S_{\sigma' - \sigma}H = H$. In general
however, different quantum Hamiltonian systems are endowed with different Hamiltonians which can be even complex.

\section{Conclusions}
\label{sec:14}
In the paper the quantization procedure of a classical Hamiltonian system was presented. In full details the quantization was performed in the canonical
coordinates for the case of the Hamiltonian system on
$\mathbb{R}^2$. In addition to this, it was shown that in the canonical regime, from the presented quantization scheme, immediately follows the
ordinary description of quantum mechanics developed by Schr\"odinger, Dirac and Heisenberg. Finally, there were given some examples of the presented
formalism, namely the free particle and the simple harmonic oscillator.

In the paper the admissible states of the quantum Hamiltonian system were defined as pseudo-probability distributions. It is however possible to associate
quantum states with probability distributions known as the \emph{tomographic probability distributions} or \emph{tomograms} \cite{Mancini:1996,
Manko.Stern:2006, Manko:2006, Ibort:2009}. This is the restatement of the Pauli problem \cite{Pauli:1990}, whether it is possible to reconstruct a quantum
state from the knowledge of the probability distributions for the position and momentum. The answer to the Pauli problem is negative \cite{Reichenbach:1942}
but it can be made positive by weakening the assumptions, namely it is possible to reconstruct a quantum state from the knowledge of a two-parameter
family of the tomographic probability distributions.

The tomographic probability representation of quantum states is strongly related to the description of quantum mechanics on the phase space as tomograms
are just the Radon transforms \cite{Radon:1917} of quantum distribution functions on the phase space. The quantum distribution functions can be then
reconstructed from the knowledge of the two-parameter family of tomograms by using the inverse Radon transform.

The tomographic probability representation of quantum states can be useful in some applications. It is particularly useful in quantum optics and quantum
information. For example measurements of photon quantum states often involve measurements of optical tomograms \cite{Smithey:1993, Mlynek:1996}, which are
easier to perform than straightforward measurements of photon quantum distribution functions.

In conclusion, the paper presents the admissible family of quantization schemes of classical Hamiltonian systems in canonical regime. It
is natural to develop the presented formalism to general Hamiltonian system with constrains. It is also natural to check how this formalism would look
like after the change of coordinates. Especially interesting would be non-canonical formulation of quantum mechanics. The further development of the
presented formalism would be an incorporation of the spin degree of freedom. Some attempts to generalize the presented quantization scheme
were already made \cite{Antonsen:1997b, Hirshfeld:2002b, Galaviz:2008}, but they still need to be systematize and refined.

\appendix
\section{Notation and used conventions}
\label{sec:8}
In the paper the Einstein summation convention is used, i.e. the summation symbol is skipped in terms where the summation index is written in the
subscript and superscript, e.g. $a_i b^i \equiv \sum_i a_i b^i$. Moreover, for $\sigma \in \mathbb{R}$ the symbol $\bar{\sigma}$ denotes the number
$\bar{\sigma} = 1 - \sigma$.

The following notations and conventions for the Fourier transform and convolution are used in the paper. For $f \in L^2(\mathbb{R}^2)$ the Fourier
transform $\mathcal{F}f = g$ and the inverse Fourier transform $\mathcal{F}^{-1}g = f$ are defined by the equations
\begin{align*}
\mathcal{F}f(\xi,\eta) & := \frac{1}{2\pi\hbar} \iint f(x,p) e^{-\frac{i}{\hbar}(\xi x - \eta p)} \d{x}\d{p}, \\
\mathcal{F}^{-1} g(x,p) & := \frac{1}{2\pi\hbar} \iint g(\xi,\eta) e^{\frac{i}{\hbar}(\xi x - \eta p)} \d{\xi}\d{\eta}.
\end{align*}
Also the partial Fourier transforms $\mathcal{F}_1 f = g$, $\mathcal{F}_2 f = h$ and they inverses $\mathcal{F}_1^{-1} g = f$, $\mathcal{F}_2^{-1} h = f$
are defined by the equations
\begin{align*}
\mathcal{F}_1 f(p,y) & := \frac{1}{\sqrt{2\pi\hbar}} \int f(x,y) e^{-\frac{i}{\hbar} xp} \d{x}, &
\mathcal{F}_1^{-1} g(x,y) & := \frac{1}{\sqrt{2\pi\hbar}} \int g(p,y) e^{\frac{i}{\hbar} xp} \d{p}, \displaybreak[0] \\
\mathcal{F}_2 f(x,p) & := \frac{1}{\sqrt{2\pi\hbar}} \int f(x,y) e^{-\frac{i}{\hbar} yp} \d{y}, &
\mathcal{F}_2^{-1} h(x,y) & := \frac{1}{\sqrt{2\pi\hbar}} \int h(x,p) e^{\frac{i}{\hbar} yp} \d{p}.
\end{align*}
Note that $\mathcal{F} f = \mathcal{F}_1 \mathcal{F}_2^{-1} f$ and $\mathcal{F}^{-1} f = \mathcal{F}_1^{-1} \mathcal{F}_2 f$. The partial Fourier
transforms $\mathcal{F}_1$, $\mathcal{F}_1^{-1}$, $\mathcal{F}_2$, $\mathcal{F}_2^{-1}$ are also denoted by $\mathcal{F}_x$, $\mathcal{F}_p^{-1}$,
$\mathcal{F}_y$, $\mathcal{F}_p^{-1}$. The Fourier transform have the following properties
\begin{align*}
\mathcal{F}(\partial_x^n \partial_p^m f)(\xi,\eta) & = \left( \frac{i}{\hbar} \xi \right)^n \left( -\frac{i}{\hbar} \eta \right)^m
\mathcal{F} f(\xi,\eta), \\
\mathcal{F}^{-1}(\partial_{\xi}^n \partial_{\eta}^m g)(x,p) & = \left( -\frac{i}{\hbar} x \right)^n \left( \frac{i}{\hbar} p \right)^m
\mathcal{F}^{-1} g(x,p), \displaybreak[0] \\
\mathcal{F}(x^n p^m f)(\xi,\eta) & = \left( i\hbar \partial_{\xi} \right)^n \left( -i\hbar \partial_{\eta} \right)^m \mathcal{F} f(\xi,\eta), \\
\mathcal{F}^{-1}(\xi^n \eta^m g)(x,p) & = \left( -i\hbar \partial_x \right)^n \left( i\hbar \partial_p \right)^m \mathcal{F}^{-1} g(x,p).
\end{align*}
Moreover, the following convention for Dirac delta distribution is used
\begin{align*}
\delta(x - x_0) & = \frac{1}{2\pi\hbar} \int e^{\frac{i}{\hbar} \xi (x - x_0)} \d{\xi}, \qquad
\delta(p - p_0) = \frac{1}{2\pi\hbar} \int e^{-\frac{i}{\hbar} \eta (p - p_0)} \d{\eta}, \displaybreak[0] \\
\delta(\xi - \xi_0) & = \frac{1}{2\pi\hbar} \int e^{-\frac{i}{\hbar} x (\xi - \xi_0)} \d{x}, \qquad
\delta(\eta - \eta_0) = \frac{1}{2\pi\hbar} \int e^{\frac{i}{\hbar} p (\eta - \eta_0)} \d{p}.
\end{align*}

A convolution of functions $f,g \in L^2(\mathbb{R}^2)$ is denoted by $f * g$ and defined by the equation
\begin{equation*}
(f * g)(x,p) := \iint f(x',p') g(x - x',p - p') \d{x'}\d{p'} \equiv \iint f(x - x',p - p') g(x',p') \d{x'}\d{p'}.
\end{equation*}
There holds the convolution theorem
\begin{equation*}
\mathcal{F}(f \cdot g) = \frac{1}{2\pi\hbar} \mathcal{F}f * \mathcal{F}g, \qquad
\mathcal{F}(f * g) = 2\pi\hbar \mathcal{F}f \cdot \mathcal{F}g.
\end{equation*}

The scalar product in some Hilbert space $\mathcal{H}$ is denoted by $\braket{\,\cdot\,|\,\cdot\,}_{\mathcal{H}}$. Moreover, the symbol
$\braket{\,\cdot\, , \,\cdot\,}$ denotes the bilinear map defined on the Schwartz space $\mathcal{S}(\mathbb{R}^2)$ by the formula
\begin{equation*}
\braket{f,g} := \iint f(x,p) g(x,p) \d{x}\d{p},
\end{equation*}
for $f,g \in \mathcal{S}(\mathbb{R}^2)$.

\section{Baker-Campbell-Hausdorff formulae}
\label{sec:1}
For two operators $\hat{A}$ and $\hat{B}$ defined on some Hilbert space, such that $[\hat{A},[\hat{A},\hat{B}]] = \textrm{const}$ and
$[\hat{B},[\hat{A},\hat{B}]] = \textrm{const}$, there holds
\begin{align*}
e^{\hat{A} + \hat{B}} & = e^{\hat{A}} e^{\hat{B}} e^{-\frac{1}{2} [\hat{A},\hat{B}]}
e^{\frac{1}{6} [\hat{A},[\hat{A},\hat{B}]] + \frac{1}{3} [\hat{B},[\hat{A},\hat{B}]]}, \\
e^{\hat{A}} e^{\hat{B}} & = e^{\hat{B}} e^{\hat{A}} e^{[\hat{A},\hat{B}]}
e^{-\frac{1}{2} [\hat{A},[\hat{A},\hat{B}]] - \frac{1}{2} [\hat{B},[\hat{A},\hat{B}]]}.
\end{align*}
Above equations are called the Baker-Campbell-Hausdorff formulae.

\section{Jensen's inequality}
\label{sec:3}
Let $\mu$ be a positive measure on a $\sigma$-algebra $\mathfrak{M}$ in a set $\Omega$, so that $\mu(\Omega) = 1$. If $f$ is a real function in
$L^1(\Omega,\mu)$ and if $\varphi \colon \mathbb{R} \to \mathbb{R}$ is convex, then \cite{Rudin:1970}
\begin{equation*}
\varphi \left( \int_\Omega f \d{\mu} \right) \le \int_\Omega (\varphi \circ f) \d{\mu}.
\end{equation*}
Above inequality is called the \emph{Jensen's inequality}.

From Jensen's inequality another useful inequality can be derived. Namely, for $f,g \in \mathcal{S}(\mathbb{R}^{2N})$ there holds
\begin{equation*}
\left| \iint f(x,p) g(x,p) \d{x} \d{p} \right|^2 \le \iint |g(x,p)| \d{x} \d{p} \iint |f(x,p)|^2 |g(x,p)| \d{x} \d{p}.
\end{equation*}
Indeed, by taking $\Omega = \mathbb{R}^{2N}$, $\mathfrak{M} = \mathfrak{B}(\mathbb{R}^{2N})$, $\varphi(x) = x^2$ and
\begin{equation*}
\d{\mu(x,p)} = \left( \iint |g(x,p)| \d{x}\d{p} \right)^{-1} |g(x,p)| \d{x}\d{p}
\end{equation*}
from Jensen's inequality there holds
\begin{align*}
\left| \iint f(x,p) g(x,p) \d{x} \d{p} \right|^2 & \le \left( \iint |f(x,p)| |g(x,p)| \d{x} \d{p} \right)^2 \nonumber \displaybreak[0] \\
& = \left( \iint |g(x,p)| \d{x}\d{p} \right)^2 \left( \iint |f(x,p)| \d{\mu(x,p)} \right)^2 \nonumber \displaybreak[0] \\
& \le \left( \iint |g(x,p)| \d{x}\d{p} \right)^2 \iint |f(x,p)|^2 \d{\mu(x,p)} \nonumber \displaybreak[0] \\
& = \iint |g(x,p)| \d{x} \d{p} \iint |f(x,p)|^2 |g(x,p)| \d{x} \d{p}.
\end{align*}

\section{Proof of Theorem \ref{thm:9.5}}
\label{sec:2}
For some $\Psi \in \mathcal{H}$ let
\begin{subequations}
\label{eq:2.28}
\begin{align}
A_L \star_{\sigma,S} \Psi & = A_{\sigma,S}(\hat{q}_{\sigma,S},\hat{p}_{\sigma,S}) \Psi = \Phi_L,
\label{eq:2.28a} \\
A_R \star_{\sigma,S} \Psi & = A_{\sigma,S}(\hat{q}^*_{\bar{\sigma},S},\hat{p}^*_{\bar{\sigma},S}) \Psi = \Phi_R.
\label{eq:2.28b}
\end{align}
\end{subequations}
For equation~(\ref{eq:2.28a}) let
\begin{equation*}
\Psi(x,p) = S e^{c_1 xp} \Psi_1(x,p), \qquad
\Phi_L(x,p) = S e^{c_1 xp} \Phi_1(x,p)
\end{equation*}
and for equation~(\ref{eq:2.28b}) respectively
\begin{equation*}
\Psi(x,p) = S e^{c_2 xp} \Psi_2(x,p), \qquad
\Phi_R(x,p) = S e^{c_2 xp} \Phi_2(x,p),
\end{equation*}
where $c_1, c_2$ are some complex numbers. Notice that by using the Baker-Campbell-Hausdorff formulae one finds
\begin{align*}
e^{\frac{i}{\hbar} \xi \hat{q}_{\sigma,S}} e^{-\frac{i}{\hbar} \eta \hat{p}_{\sigma,S}} S e^{c_1 xp}
& = e^{\frac{i}{\hbar} \xi S \hat{q}_\sigma S^{-1}} e^{-\frac{i}{\hbar} \eta S \hat{p}_\sigma S^{-1}} S e^{c_1 xp}
= S e^{\frac{i}{\hbar} \xi \hat{q}_\sigma} S^{-1} S e^{-\frac{i}{\hbar} \eta \hat{p}_\sigma} S^{-1} S e^{c_1 xp} \\
& = S e^{\frac{i}{\hbar} \xi \hat{q}_\sigma} e^{-\frac{i}{\hbar} \eta \hat{p}_\sigma} e^{c_1 xp}
= S e^{c_1 xp} e^{\frac{i}{\hbar} \xi \hat{Q}_\sigma} e^{-\frac{i}{\hbar} \eta \hat{P}_\sigma},
\end{align*}
where $\hat{Q}_\sigma = \hat{q}_\sigma + i\hbar \sigma c_1 x$, $\hat{P}_\sigma = \hat{p}_\sigma - i\hbar \bar{\sigma} c_1 p$,
$[\hat{Q}_\sigma, \hat{P}_\sigma] = i\hbar$ and similarly
\begin{align*}
e^{\frac{i}{\hbar} \xi \hat{q}^*_{\bar{\sigma},S}} e^{-\frac{i}{\hbar} \eta \hat{p}^*_{\bar{\sigma},S}} S e^{c_2 xp}
& = e^{\frac{i}{\hbar} \xi S \hat{q}^*_{\bar{\sigma}} S^{-1}} e^{-\frac{i}{\hbar} \eta S \hat{p}^*_{\bar{\sigma}} S^{-1}} S e^{c_2 xp}
= S e^{\frac{i}{\hbar} \xi \hat{q}^*_{\bar{\sigma}}} S^{-1} S e^{-\frac{i}{\hbar} \eta \hat{p}^*_{\bar{\sigma}}} S^{-1} S e^{c_2 xp} \\
& = S e^{\frac{i}{\hbar} \xi \hat{q}^*_{\bar{\sigma}}} e^{-\frac{i}{\hbar} \eta \hat{p}^*_{\bar{\sigma}}} e^{c_2 xp}
= S e^{c_2 xp} e^{\frac{i}{\hbar} \xi \hat{Q}^*_{\bar{\sigma}}} e^{-\frac{i}{\hbar} \eta \hat{P}^*_{\bar{\sigma}}},
\end{align*}
where $\hat{Q}^*_{\bar{\sigma}} = \hat{q}^*_{\bar{\sigma}} - i\hbar \bar{\sigma} c_2 x$, $\hat{P}^*_{\bar{\sigma}} = \hat{p}^*_{\bar{\sigma}}
+ i\hbar \sigma c_2 p$, $[\hat{Q}^*_{\bar{\sigma}}, \hat{P}^*_{\bar{\sigma}}] = -i\hbar$. Now, from (\ref{eq:2.28}) one gets
\begin{subequations}
\label{eq:2.30}
\begin{align}
A_{\sigma,S}(\hat{Q}_\sigma,\hat{P}_\sigma) \Psi_1 & = \Phi_1,
\label{eq:2.30a} \\
A_{\sigma,S}(\hat{Q}^*_{\bar{\sigma}},\hat{P}^*_{\bar{\sigma}}) \Psi_2 & = \Phi_2,
\label{eq:2.30b}
\end{align}
\end{subequations}
where
\begin{align*}
\hat{Q}_\sigma & = x + i\hbar \sigma c_1 x + i\hbar \sigma \partial_p, &
\hat{P}_\sigma & = p - i\hbar \bar{\sigma} c_1 p - i\hbar \bar{\sigma} \partial_x, \\
\hat{Q}^*_{\bar{\sigma}} & = x - i\hbar \bar{\sigma} c_2 x - i\hbar \bar{\sigma} \partial_p, &
\hat{P}^*_{\bar{\sigma}} & = p + i\hbar \sigma c_2 p + i\hbar \sigma \partial_x.
\end{align*}
Under the choice $c_1 = -\frac{i}{\hbar} \bar{\sigma}^{-1}$ and $c_2 = \frac{i}{\hbar} \sigma^{-1}$ formulae (\ref{eq:2.30}) take a form
\begin{subequations}
\label{eq:2.32}
\begin{align}
A_{\sigma,S}(\bar{\sigma}^{-1} x + i\hbar \sigma \partial_p, -i\hbar \bar{\sigma} \partial_x) \Psi_1 & = \Phi_1,
\label{eq:2.32a} \\
A_{\sigma,S}(\sigma^{-1} x - i\hbar \bar{\sigma} \partial_p, i\hbar \sigma \partial_x) \Psi_2 & = \Phi_2.
\label{eq:2.32b}
\end{align}
\end{subequations}
Now, taking the inverse Fourier transform of both equations (\ref{eq:2.32}) with respect to $p$ variable one gets
\begin{subequations}
\label{eq:2.34}
\begin{align}
A_{\sigma,S}(\bar{\sigma}^{-1} x + \sigma z, -i\hbar \bar{\sigma} \partial_x) \tilde{\Psi}_1(x,z) & = \tilde{\Phi}_1(x,z),
\label{eq:2.34a} \\
A_{\sigma,S}(\sigma^{-1} x - \bar{\sigma} z, i\hbar \sigma \partial_x) \tilde{\Psi}_2(x,z) & = \tilde{\Phi}_2(x,z).
\label{eq:2.34b}
\end{align}
\end{subequations}
Lets introduce new coordinates
\begin{align*}
\xi  & = \bar{\sigma}^{-1} x + \sigma z, & z & = z, & \Rightarrow \qquad \partial_x & = \bar{\sigma}^{-1} \partial_\xi, \\
\eta & = \sigma^{-1} x - \bar{\sigma} z, & z & = z, & \Rightarrow \qquad \partial_x & = \sigma^{-1} \partial_\eta,
\end{align*}
then equations (\ref{eq:2.34}) can be written as
\begin{subequations}
\label{eq:2.36}
\begin{align}
A_{\sigma,S}(\xi, -i\hbar \partial_\xi) \tilde{\Psi}_1(\xi,z) & = \tilde{\Phi}_1(\xi,z),
\label{eq:2.36a} \\
A_{\sigma,S}(\eta, i\hbar \partial_\eta) \tilde{\Psi}_2(\eta,z) & = \tilde{\Phi}_2(\eta,z).
\label{eq:2.36b}
\end{align}
\end{subequations}
Now, lets restrict functions $\Psi$ to the class for which $\tilde{\Psi}_1(\xi,z) = \varphi_1(\xi) \kappa_1(z)$ and $\tilde{\Psi}_2(\eta,z)
= \varphi^*_2(\eta) \kappa_2(z)$, where $\varphi_1, \varphi_2, \kappa_1, \kappa_2 \in L^2(\mathbb{R})$. Then, obviously $\tilde{\Phi}_1(\xi,z)
= \psi_1(\xi) \kappa_1(z)$ and $\tilde{\Phi}_2(\eta,z) = \psi^*_2(\eta) \kappa_2(z)$ for some $\psi_1, \psi_2 \in L^2(\mathbb{R})$ and equations
(\ref{eq:2.36}) reduce to the form
\begin{subequations}
\label{eq:2.38}
\begin{align}
A_{\sigma,S}(\xi, -i\hbar \partial_\xi) \varphi_1(\xi) & = \psi_1(\xi),
\label{eq:2.38a} \\
A^\dagger_{\sigma,S}(\eta, -i\hbar \partial_\eta) \varphi_2(\eta) & = \psi_2(\eta),
\label{eq:2.38b}
\end{align}
\end{subequations}
as in $L^2(\mathbb{R})$ $(A_{\sigma,S}(\eta, i\hbar \partial_\eta))^* = A^\dagger_{\sigma,S}(\eta, -i\hbar \partial_\eta)$.

Function $\Psi$ can be reconstructed either from $\Psi_1$ or from $\Psi_2$. Indeed, from one side one have
\begin{align*}
\Psi(x,p) & = S e^{-\frac{i}{\hbar} \bar{\sigma}^{-1} xp} \frac{1}{\sqrt{2\pi\hbar}} \int \varphi_1(\bar{\sigma}^{-1} x + \sigma z) \kappa_1(z)
e^{-\frac{i}{\hbar} pz} \d{z} \\
& = S \frac{1}{\sqrt{2\pi\hbar}} \int \varphi_1(\bar{\sigma}^{-1} x + \sigma z) \kappa_1(z)
e^{-\frac{i}{\hbar}(z + \bar{\sigma}^{-1}x)p} \d{z} \\
& = S \frac{1}{\sqrt{2\pi\hbar}} \int \varphi_1(x + \sigma y) \kappa_1(y - \bar{\sigma}^{-1}x)
e^{-\frac{i}{\hbar} yp} \d{y},
\end{align*}
where the new variable $y = z + \bar{\sigma}^{-1}x$ was introduced. From the other side one have
\begin{align*}
\Psi(x,p) & = S e^{\frac{i}{\hbar} \sigma^{-1} xp} \frac{1}{\sqrt{2\pi\hbar}} \int \varphi^*_2(\sigma^{-1} x - \bar{\sigma} z) \kappa_2(z)
e^{-\frac{i}{\hbar} pz} \d{z} \\
& = S \frac{1}{\sqrt{2\pi\hbar}} \int \varphi^*_2(\sigma^{-1} x - \bar{\sigma} z) \kappa_2(z)
e^{-\frac{i}{\hbar}(z - \sigma^{-1}x)p} \d{z} \\
& = S \frac{1}{\sqrt{2\pi\hbar}} \int \varphi^*_2(x - \bar{\sigma} y) \kappa_2(y + \sigma^{-1}x)
e^{-\frac{i}{\hbar} yp} \d{z},
\end{align*}
where the new variable $y = z - \sigma^{-1}x$ was introduced. Thus,
\begin{align}
\kappa_1(y - \bar{\sigma}^{-1} x) & = \varphi^*_2(x - \bar{\sigma} y), & \kappa_2(y + \sigma^{-1} x) & = \varphi_1(x + \sigma y)
\end{align}
and
\begin{equation*}
\Psi(x,p) = S \frac{1}{\sqrt{2\pi\hbar}} \int \varphi_1(x + \sigma y) \varphi^*_2(x - \bar{\sigma} y) e^{-\frac{i}{\hbar} yp} \d{y}
= (\varphi_2^* \otimes_{\sigma,S} \varphi_1)(x,p).
\end{equation*}
Now,
\begin{align*}
\Phi_L(x,p) & = S e^{-\frac{i}{\hbar} \bar{\sigma}^{-1} xp} \frac{1}{\sqrt{2\pi\hbar}} \int \psi_1(\bar{\sigma}^{-1} x + \sigma z) \kappa_1(z)
e^{-\frac{i}{\hbar} pz} \d{z} \\
& = S \frac{1}{\sqrt{2\pi\hbar}} \int \psi_1(\bar{\sigma}^{-1} x + \sigma z) \kappa_1(z)
e^{-\frac{i}{\hbar}(z + \bar{\sigma}^{-1}x)p} \d{z} \\
& = S \frac{1}{\sqrt{2\pi\hbar}} \int \psi_1(x + \sigma y) \kappa_1(y - \bar{\sigma}^{-1}x)
e^{-\frac{i}{\hbar} yp} \d{y} \\
& = S \frac{1}{\sqrt{2\pi\hbar}} \int \psi_1(x + \sigma y) \varphi^*_2(x - \bar{\sigma} y)
e^{-\frac{i}{\hbar} yp} \d{y} \\
& = (\varphi_2^* \otimes_{\sigma,S} \psi_1)(x,p),
\end{align*}
where the new variable $y = z + \bar{\sigma}^{-1}x$ was used. In a similar way one can show that
\begin{equation*}
\Phi_R(x,p) = (\psi_2^* \otimes_{\sigma,S} \varphi_1)(x,p).
\end{equation*}

Above calculations show that for $\Psi = \varphi_2^* \otimes_{\sigma,S} \varphi_1$ one have
\begin{align*}
A_L \star_{\sigma,S} \Psi & = A_{\sigma,S}(\hat{q}_{\sigma,S},\hat{p}_{\sigma,S}) \Psi
= \varphi_2^* \otimes_{\sigma,S} \psi_1
= \varphi_2^* \otimes_{\sigma,S} A_{\sigma,S}(\hat{q},\hat{p}) \varphi_1, \\
A_R \star_{\sigma,S} \Psi & = A_{\sigma,S}(\hat{q}^*_{\bar{\sigma},S},\hat{p}^*_{\bar{\sigma},S}) \Psi
= \psi_2^* \otimes_{\sigma,S} \varphi_1
= \left( A^\dagger_{\sigma,S}(\hat{q},\hat{p}) \varphi_2 \right)^* \otimes_{\sigma,S} \varphi_1.
\end{align*}

\bibliographystyle{model1a-num-names}
\bibliography{aps-jour,dq}

\end{document}